\pdfoutput=1
\documentclass[12pt]{article}

\usepackage[T1]{fontenc}
\usepackage{mathpazo}

\usepackage{amsmath,amsthm,bm}
\newcommand{\norm}[1]{\left\lVert#1\right\rVert}
\newcommand{\RomNum}[1]{%
  \textup{\uppercase\expandafter{\romannumeral#1}}%
}

\usepackage{enumitem}
\usepackage[margin=1in]{geometry}
\usepackage{comment}
\usepackage{float}
\usepackage{natbib}
\usepackage{booktabs}
\usepackage{graphicx}
\usepackage{breakcites}
\usepackage{tabularx}

\usepackage{setspace}

\usepackage{multirow}
\usepackage[flushleft]{threeparttable}

\usepackage{makecell}
\usepackage{caption}
\usepackage[dvipsnames]{xcolor}
\usepackage[colorlinks = true,linkcolor = blue]{hyperref}
\definecolor{darkblue}{rgb}{0, 0, 0.5}
\hypersetup{
     colorlinks = true,
     citecolor = Maroon,
     urlcolor = darkblue,
     linkcolor = darkblue
}

\usepackage{subcaption} 

\DeclareMathOperator*{\argmin}{argmin}

\newtheorem{proposition}{Proposition}
\newtheorem{theorem}{Theorem}
\newtheorem{lemma}{Lemma}
\newtheorem{assumption}{Assumption}
\newtheorem{definition}{Definition}
\newtheorem{corollary}{Corollary}
\newtheorem{example}{Example}
\newtheorem*{remark}{Remark}

\begin{document}

\title{Nonlinear and Nonseparable Structural Functions in Fuzzy Regression Discontinuity Designs\thanks{The author is indebted to his advisors Graham Elliott and Yixiao Sun for their constant support on this paper. For helpful comments, the author also thanks Wei-Lin Chen, Gordon Dahl, Xinwei Ma, Katherine Rittenhouse, Jack Rosetti, Kaspar W\"uthrich, and participants at UC San Diego econometrics seminar, California econometrics conference, and AMES China meeting.}}

\author{Haitian Xie\thanks{Department of Economics, University of California, San Diego. Address: 9500 Gilman Dr. La Jolla, CA 92093. Email: \url{hax082@ucsd.edu}. }}

\date{\textbf{\today}}

\maketitle

\begin{abstract} 
    Many empirical examples of regression discontinuity (RD) designs concern a continuous treatment variable, but the theoretical aspects of such models are less studied. This study examines the identification and estimation of the structural function in fuzzy RD designs with a continuous treatment variable. The structural function fully describes the causal impact of the treatment on the outcome. We show that the nonlinear and nonseparable structural function can be nonparametrically identified at the RD cutoff under shape restrictions, including monotonicity and smoothness conditions. Based on the nonparametric identification equation, we propose a three-step semiparametric estimation procedure and establish the asymptotic normality of the estimator. The semiparametric estimator achieves the same convergence rate as in the case of a binary treatment variable. As an application of the method, we estimate the causal effect of sleep time on health status by using the discontinuity in natural light timing at time zone boundaries.

    \bigskip
    \noindent%
\end{abstract}

\newpage
\section{Introduction}

The \emph{regression discontinuity} (RD) design is one of the most credible approaches to causal inference in non-experimental settings. In an RD design, the researcher is interested in the effect of a treatment $T$ on some outcome $Y$. The basic idea is that there is an observed running variable $R$ (also called score or index or forcing variable) such that the treatment varies discontinuously when the running variable crosses some cutoff (also called threshold) value $\bar{r}$. By utilizing this discontinuity, the researcher has the power to identify and estimate the causal impact of interest.

Most theoretical studies of the RD design assume that the treatment is a binary intervention. However, in empirical settings, researchers may be interested in a continuous treatment that takes value inside an interval. Such examples include sleep time, air pollution level, and medical spending. The goal of this study is to provide methods for examining the causal effect of a continuous treatment variable in an RD setting.

It takes a few steps to extend the idea of RD design from a binary treatment to a continuous one. With a binary treatment, the sharp design refers to the case where the running variable completely determines the treatment. In particular, the treatment changes from $0$ to $1$ when the running variable crosses the cutoff. The fuzzy design refers to the case where the treatment probability jumps at the cutoff. The jump can be smaller and need not be from $0$ to $1$. The sharp and fuzzy designs of a binary treatment are demonstrated in Figure \ref{fig:binary-demo}.

\begin{figure}
    \centering
    \includegraphics[width=1\linewidth]{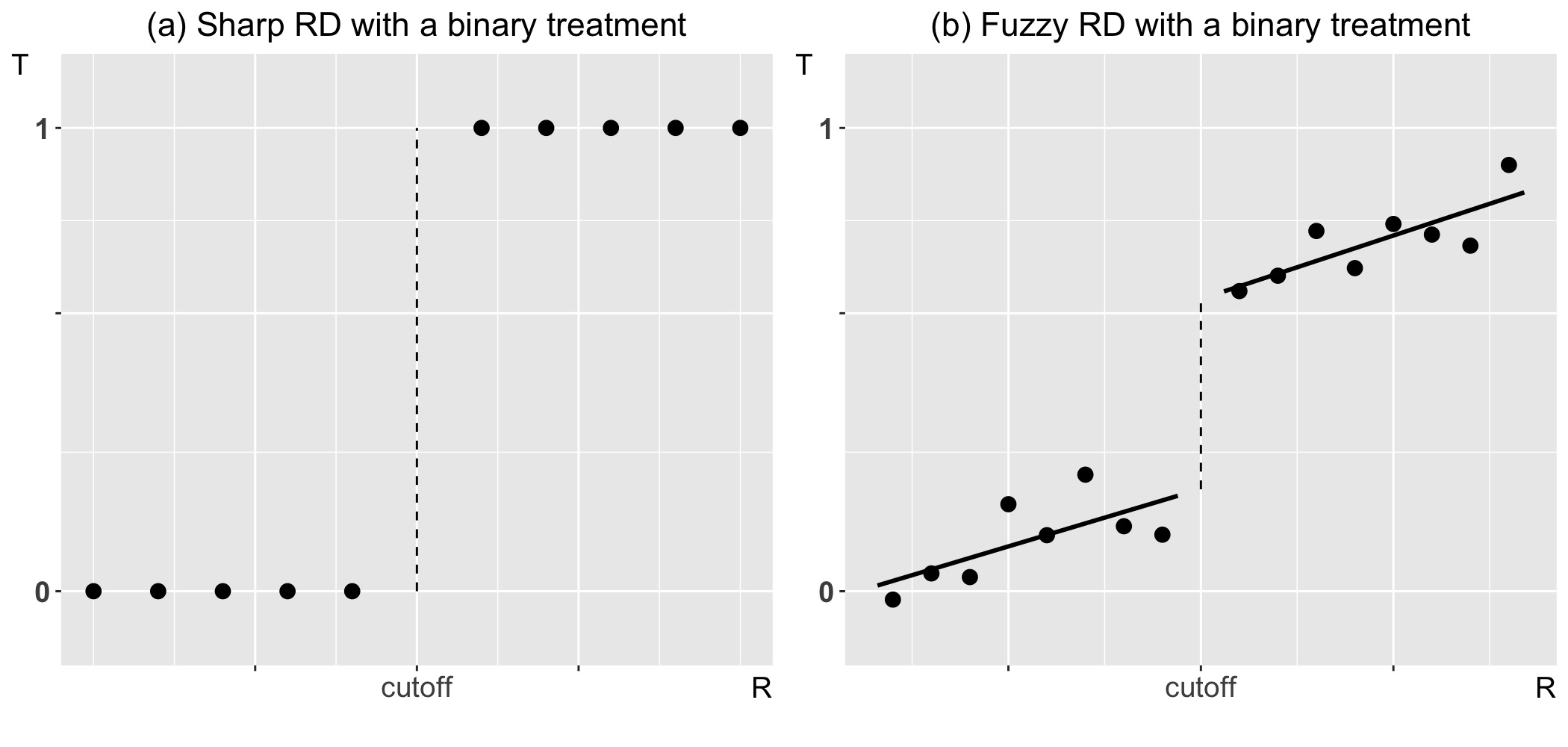}
    \caption{Demonstration of RD designs with a binary treatment.}
    \label{fig:binary-demo}
    \caption*{\footnotesize Graph (a) demonstrates the sharp RD design using a raw scatter plot. Graph (b) demonstrates the fuzzy RD design using a binscatter plot, where each dot represents the average treatment probability in the respective bin.}
\end{figure}

When the treatment variable is continuous, the representation of the RD becomes more complicated than the binary case. The reason is that the distribution of a binary variable can be completely summarized by the scalar treatment probability as in Figure \ref{fig:binary-demo}(b), while a continuous variable contains much more information. Specifically, we can consider quantile regressions of the treatment on the running variable at different quantile levels. Each quantile level would deliver a different regression model with a different discontinuity. Eventually, we would obtain an infinite number of regression discontinuities based on all quantile levels of the treatment. This infinite set of regression discontinuities can be represented as the entire variation between the conditional quantile function of the treatment from just below and just above the cutoff. Figure \ref{fig:continuous-demo} provides a demonstration.

\begin{figure}
    \centering
    \includegraphics[width=1\linewidth]{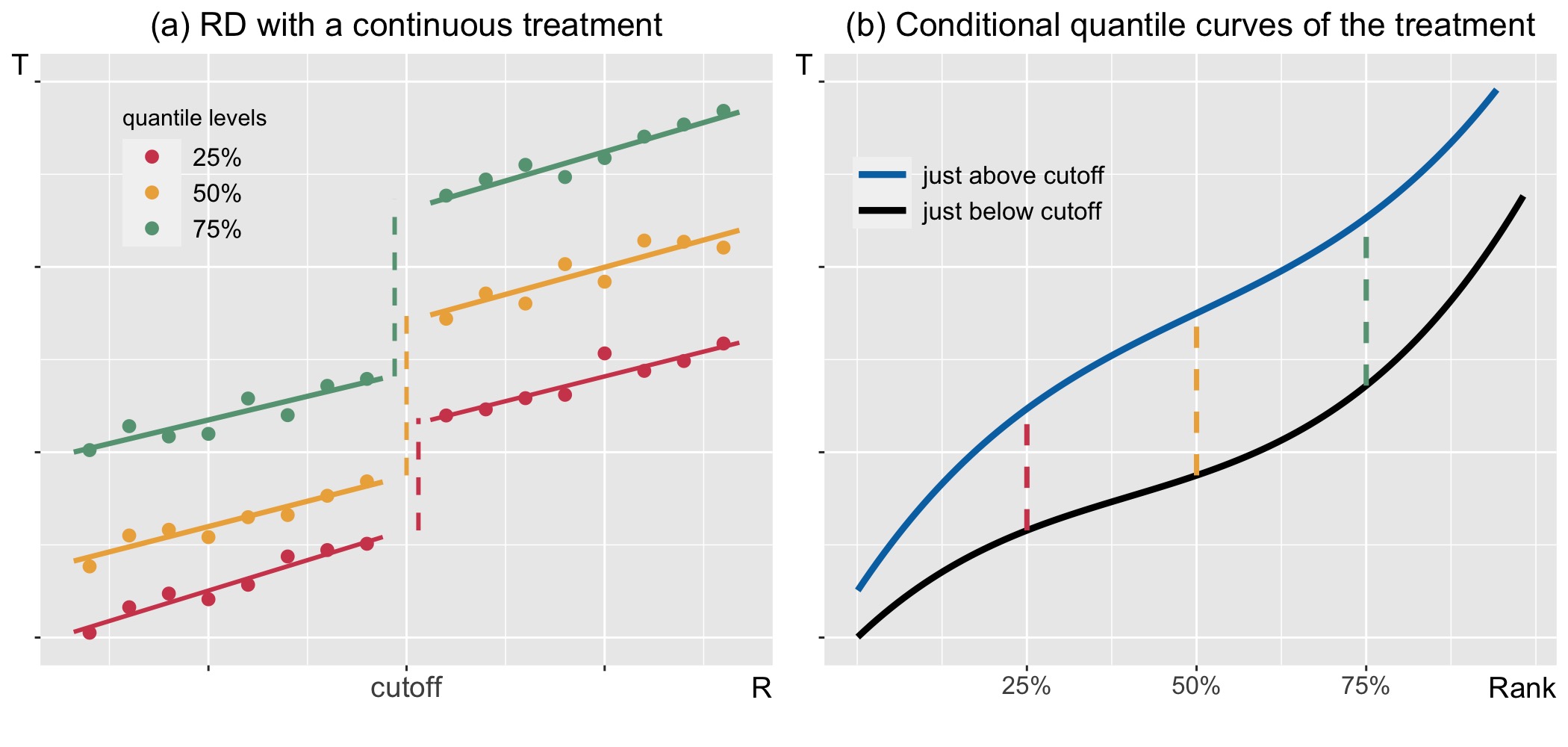}
    \caption{Demonstration of RD designs with a continuous treatment.}
    \label{fig:continuous-demo}
    \caption*{\footnotesize Graph (a) demonstrates the regression discontinuities of a continuous treatment variable at different quantile levels. The plot is a binscatter plot, where each dot represents the corresponding quantile treatment level in the respective bin. Different quantile regressions bring different discontinuities. Graph (b) plots the conditional quantile curve of the treatment from just below and just above the cutoff. The horizontal axis specifies which quantile level we are looking at. The entire difference between these two curves constitutes the content of RD of a continuous treatment variable. If there is no regression discontinuity, then the two quantile curves would completely overlap. Notice that we use the same color to denote the corresponding jumps between the two plots. }
\end{figure}

The exogenous variation of the treatment contained in the aforementioned set of regression discontinuities provides tremendous identification power on the causal effect of interest. To fully express the causal effect of the treatment $T$ on the outcome $Y$, we introduce the structural function 
\begin{align*}
    Y = g^*(T,R,\varepsilon),
\end{align*}
where $\varepsilon$ contains unobserved causal factors (for easy reference, $\varepsilon$ will be called the error term hereafter). The structural function $g^*$ specifies how the treatment $T$ determines the outcome $Y$ together with the running variable $R$ and error term $\varepsilon$. 

Consider an empirical example for concreteness, where we are interested in the causal impact of sleep time on health. Figure \ref{fig:continuous-empirical}(a) shows the histogram of sleep time based on the American Time Use Survey (ATUS) and demonstrates that sleep time is indeed a continuous treatment variable. The causal identification is based on exploiting the discontinuity in the timing of natural light at time zone boundaries. Individuals living on the late sunset side of the time zone boundary tend to go to bed at a later time, while in the morning, everyone gets up and goes to work at 8 am. This generates an exogenous variation in the sleep time across the time zone boundary.\footnote{This identification strategy is first proposed by \citep{Giuntella2019sunset} within a linear model. See the empirical application in Section \ref{sec:numeric} for more details.} Similar to the demonstration in Figure \ref{fig:continuous-demo}(b), we would expect the distribution of the sleep time for individuals living on the early sunset side to first-order stochastically dominate the distribution on the late sunset side. This relationship is supported by Figure \ref{fig:continuous-empirical}(b) based on nonparametric estimates of the conditional quantiles of sleep time. In this example, the running variable is the distance to the time zone boundary, and the cutoff $\bar{r}$ is at the time zone boundary. The error term $\varepsilon$ may contain unobserved eating habits that correlate with both health and sleep time. 

\begin{figure}
    \centering
    \includegraphics[width=1\linewidth]{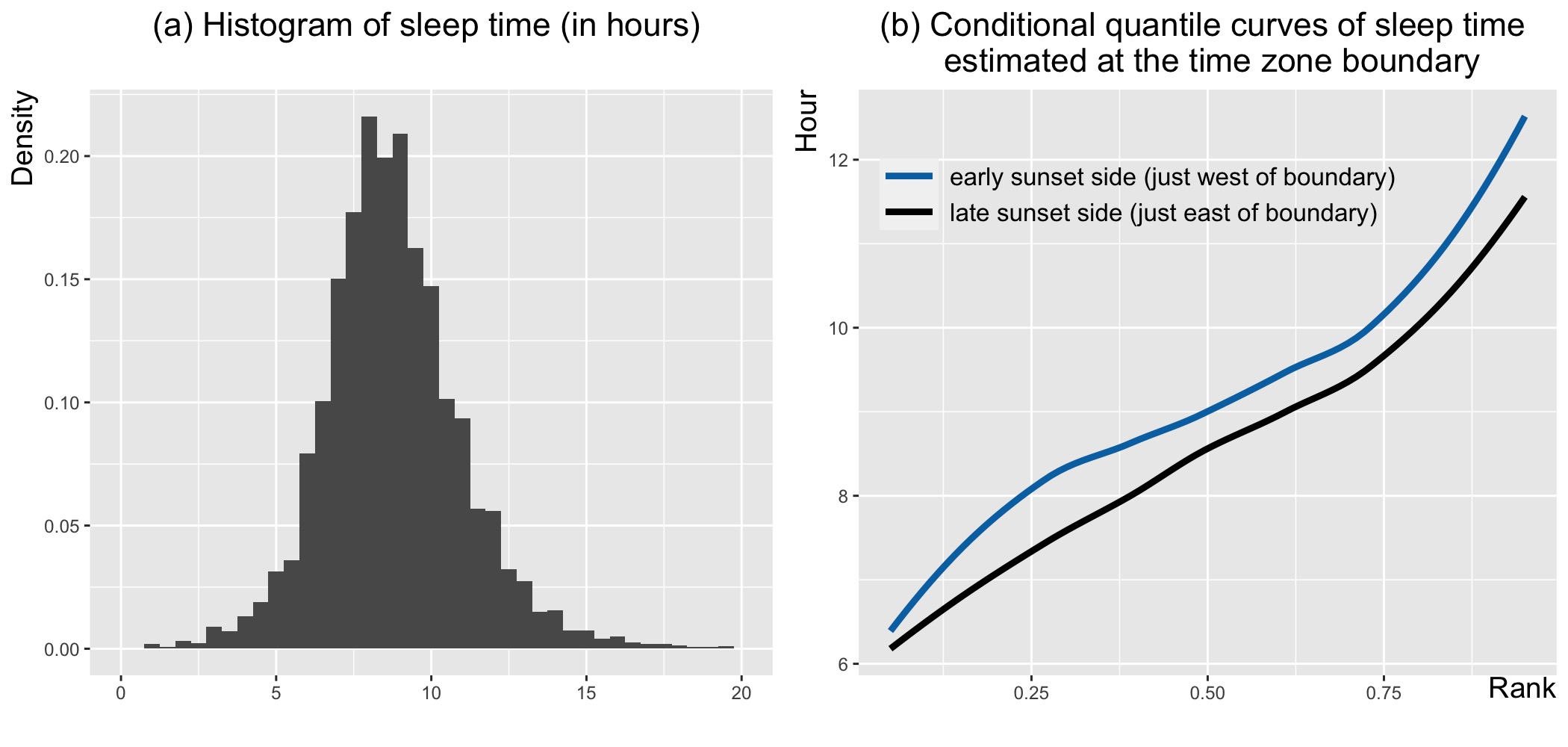}
    \caption{Empirical illustration of RD designs with a continuous treatment.}
    \label{fig:continuous-empirical}
    \caption*{\footnotesize Graph (a) shows the histogram of sleep time. Evidently, this variable is better modeled as continuous rather than discrete. Graph (b) shows the estimated conditional quantile curves of sleep time given that the geographical location is just west and east of the time zone boundary. The nonparametric estimator used here is the local constant quantile regression. The RD is clearly observed as the blue curve first-order stochastically dominates the black curve, a similar situation as demonstrated in Figure \ref{fig:continuous-demo}(b).}
\end{figure}

The goal of the RD design is to use the discontinuity to identify the structural function $g^*$ at the cutoff $\bar{r}$. When the treatment is binary, the information contained in the structural function can be reduced to a scalar treatment effect 
\begin{align*}
    g^*(1,\bar{r},\varepsilon) - g^*(0,\bar{r},\varepsilon),
\end{align*}
which is the difference in outcome when the treatment is manipulated from $0$ to $1$. However, when the treatment is continuous, the structural function is an infinite-dimensional object and is much harder to identify.

In practice, empirical studies often use the \emph{two stage least squares} (TSLS) method to estimate the following \emph{Wald ratio}:
\begin{align*}
    \text{Wald ratio} = \frac{ \lim_{r \uparrow \bar{r}} \mathbb{E}[Y|R=r] - \lim_{r \downarrow \bar{r}} \mathbb{E}[Y|R=r] }{\lim_{r \uparrow \bar{r}} \mathbb{E}[T|R=r] - \lim_{r \downarrow \bar{r}} \mathbb{E}[T|R=r]},
\end{align*}
There are two motivations behind this procedure. First, in the binary treatment case, the Wald ratio would identify the treatment effect.\footnote{As shown in \cite{hahn2001identification}, the Wald ratio identifies the average treatment effect in the sharp design and the local average treatment effect (for the compliers) in the fuzzy design.} Second, in the continuous treatment case, if the structural function is linear and separable in the treatment, that is, if the structural function can be decomposed as
\begin{align*}
    g^*(T,R,\varepsilon) = \beta T + \tilde{g}(R,\varepsilon),
\end{align*}
then the Wald ratio would identify the slope coefficient $\beta$ of the treatment.\footnote{Such a linear specification of RD design with a continuous treatment can be found in Section 3.4.2 of \cite{lee2010regression}.} However, the Wald ratio cannot identify the structural function in general because the structural function is infinite-dimensional while the Wald ratio is one-dimensional. Any attempt to condense the structural function into a scalar bears the risk of dampening the causal interpretation of the model.

The preceding discussion shows that the general identification of the structural function in RD designs remains an unsolved issue. It is desirable to know whether the structural function (at the cutoff) can be identified without the aforementioned linearity and separability conditions. This issue is a practical concern. For instance, in the time zone example, there are reasons for one to believe that the structural function is nonlinear and nonseparable in the treatment.\footnote{The nonlinearity can be due to the fact that both undersleeping and oversleeping are harmful to health. The nonseparability can be due to the effect heterogeneity caused by unobserved eating habits, which affect both sleep time and health.} The optimal sleep time can only be determined after the identification of the nonlinear structural function. From the theoretical perspective, it is wise to achieve identification in the nonparametric sense and avoid functional form restrictions such as linearity and separability that do not have economic theory foundations. As an advantage, the more general specification allows the treatment effect to be heterogeneous across different levels of the treatment and outcome.

The current study aims to precisely tackle the identification and estimation of the possibly nonlinear and nonseparable structural function. The nonparametric identification result is established based on shape restrictions, including monotonicity and smoothness conditions. The idea behind the identification result is that we are using the infinite set of regression discontinuities in Figure \ref{fig:continuous-demo}(b) to identify the infinite-dimensional structural function. The monotonicity condition restricts the structural function $g^*$ to be strictly increasing in the error term $\varepsilon$. This condition requires the error term to be one-dimensional, which is the potential restriction of the model. However, this condition is common in the nonparametric identification literature \citep[e.g.,][]{matzkin2003nonparametric} and is satisfied by most, if not all, parametric models used in practice. The smoothness and other regularity conditions imposed in this paper are common in the RD literature.

A semiparametric estimation procedure is developed based on the nonparametric identification result. The structural function is parametrized while nonlinearity and nonseparability are maintained. One such parametrization could be
\begin{align*}
    g_\gamma(T,\bar{r},\varepsilon) = \gamma_1 T + \gamma_2 T^2 + \gamma_3 T\varepsilon + \varepsilon.
\end{align*}
The relationship between the treatment and the running variable is left to be nonparametric. Under appropriate conditions, the semiparametric estimator of the structural parameter $\gamma=(\gamma_1,\gamma_2,\gamma_3)$ is shown to be consistent and asymptotically normal. As an interesting finding, the convergence rate of the semiparametric estimator, $n^{-2/5}$, is the same as in the binary treatment case. There is no loss in terms of convergence rate when extending the RD design from the binary treatment case to the continuous case. The faster convergence rate is due to the integral smoothing in the estimation of the criterion function constructed from the identification equation. To understand this phenomenon, one can consider the analogy in regular semiparametric estimation theory, where the first step is nonparametric while the second step recovers the parametric rate.

The rest of the paper is organized as follows. The remaining part of this section discusses the literature. Section \ref{sec:identification} introduces the RD model with a continuous treatment and presents the nonparametric identification result. Section \ref{sec:estimation} proposes the semiparametric estimation procedure and derives its asymptotic properties. Section \ref{sec:numeric} presents the empirical application and simulation studies. The technical proofs for the identification and estimation results are collected in Appendices \ref{sec:proof-id} and \ref{sec:proof-est}, respectively.

\subsection{Relation to the literature}

The RD method is first introduced by \citet{thistlethwaite1960regression} into the literature. \citet{hahn2001identification} establish the theoretical foundation of RD designs by using the potential outcome framework and show that the RD Wald ratio can be interpreted as the \emph{local average treatment effect} (LATE) for compliers local to the cutoff. Early reviews of the RD design can be found in \citet{IMBENS2008regression} and \citet{lee2010regression}. For more recent reviews, see \citet{cattaneo2017regression} and \citet{cattaneo2021regression}.

There are many empirical papers that study the causal effect of a continuous treatment in an RD design, some of which are given in Table \ref{tb:rd-studies}. As explained earlier, these studies apply the TSLS method to estimate the Wald ratio. Hence, there is room for potential improvement in these settings by using the semiparametric estimator developed in the current study.

\begin{table}[!hbtp]
    \small
	\centering
    \def\arraystretch{1.5}
	\begin{tabularx}{\textwidth}{|l|l|l|l|X|}
	\hline
	Context & Outcome $Y$ & Treatment $T$ & Running $R$ & Studies \\ \hline
	\makecell[l]{US time \\ zone system} & Health status & Sleep time & \makecell[l]{Distance to \\ time zone \\ boundary} & \cite{Giuntella2019sunset} \\ \hline
	\makecell[l]{Heating policy \\ in China} & Life expectancy & Air pollution & \makecell[l]{Distance to \\ Huai River} & \makecell[l]{\cite{Chen2013evidence} \\ \cite{Ebenstein2017new}} \\ \hline 
	\makecell[l]{Minimum capital \\ requirement} & \makecell[l]{Bank failue} & Capital & Town size & \cite{dong2021regression} \\ \hline
	\makecell[l]{Very low birth \\ weight policy} & \makecell[l]{Infant \\ mortality} & \makecell[l]{Medical \\ spending} & Birth weight & \makecell[l]{\cite{almond2010estimating} \\ \cite{barreca2011saving}} \\ \hline
	\makecell[l]{Tax distribution \\ in Brazil} & \makecell[l]{Electoral chance \\ of the incumbent} & \makecell[l]{Local \\ government \\ spending} & \makecell[l]{Municipality \\ population} & \cite{litschig2010government} \\ \hline
	\makecell[l]{Child-related \\ tax benefits} & \makecell[l]{Personal \\ achievements }& \makecell[l]{ Family \\ income} & \makecell[l]{Birthdate} &  \makecell[l]{\cite{barr2021effect} \\ \cite{cole2021effect}} \\ \hline
	\end{tabularx}
	\caption{Selected empirical RD studies with a continuous treatment variable.}
	\label{tb:rd-studies}
\end{table}

The theoretical literature on RD designs focuses on the case of a binary treatment variable. The one exception is the recent paper by \citet{dong2021regression}, which studies RD designs specifically with a continuous treatment variable. Under simple conditions, they propose a way to identify and estimate the \emph{Quantile specific LATE}. This parameter bears a causal interpretation as it is a weighted average of the derivative of the structural function \citep[p. 4]{dong2021regression}. It can also be understood as the treatment effect given a particular quantile of the treatment. Their results are established under conditions weaker than the ones in our paper. In particular, they do not assume the monotonicity condition of the structural function. In certain situations, however, the policy design process may require information beyond the weighted average of the structural function. The current paper takes a different approach and aims to identify the structural function directly.

It has become common in the literature to identify a certain weighted average of the structural function as the causal estimand. One of the first examples is the 2SLS estimation with a multivalued treatment \citep{angrist-imbens1995}. The reason for this trend is twofold: the direct identification of the structural function is difficult, and the researchers want the assumptions they make to be minimal. However, the weighted average only provides summary information on the structural function, which is not sufficient in optimal policy designs. In this paper, we make an effort to identify the structural function itself at the expense of making stronger assumptions. In the empirical application in Section \ref{sec:numeric}, we show that the estimated nonlinear structural function can help determine the optimal sleep time while the TSLS estimates cannot.

The identification in the RD design is related to that in the \emph{instrumental variables} (IV) models of triangular systems. The control function approach described in \cite{imbens2009identification} states that the variation in the treatment becomes exogenous after conditioning on the control function. As explained later, a similar phenomenon is also observed in the RD model with a continuous treatment. It explains the intuition behind the nonparametric identification equation. Another relevant literature is the one that studies instruments with small support \citep{torgovitsky2015identification,D2015id,TORGOVITSKY2017minimum}. These papers examine a model with a discrete instrumental variable and a continuous treatment variable. Since RD can be interpreted as a local IV approach, our framework is related to the large body of this IV literature.

That said, this paper is not a straightforward extension of the results from the IV literature. The difference between the RD design and the IV approach includes the following. First, the identification in the IV model relies on the (conditional) independence of the IV with the error term, while the identification in RD designs is based on the discontinuity and does not depend on any independence assumption. This is one of the reasons that the RD method is considered to be more credible than IV for causal inference. Second, in an RD design, the running variable directly affects both the outcome and the treatment and hence does not satisfy the exclusion restriction typically required in IV models. This inclusion also gives rise to the unique issue of extrapolation away from the cutoff. Third, the estimation procedure in the RD design focuses on the local neighborhood of the cutoff. It is theoretically more challenging to derive the asymptotic properties of the estimator.

The problem studied by this paper is also related to the broad literature on the nonparametric identification of structural functions. Relevant papers include \cite{matzkin2003nonparametric} and \cite{hoderlein2007identification,hoderlein2009identification}. The identification there relies on the exogeneity of the treatment, which is not required in RD designs.

\section{RD Design with a continuous treatment} \label{sec:identification}

This section describes the RD model with a continuous treatment, explains the assumptions of the model, and discusses the nonparametric identification of the structural function local to the cutoff. 

\subsection{The model}

We study the following causal equation:
\begin{align} \label{eqn:outcome-equation}
    Y & = g^*(T,R,\varepsilon),
\end{align}
where $Y$ is the outcome of interest, $T$ is the treatment, and $R$ is the running variable. The scalar variable $\varepsilon$ represents unobserved causal factors in the outcome equation. We assume all the random variables are absolutely continuous. The function $g^*$ is the unknown true structural function. 

The running variable $R$ partly determines the treatment $T$ by the following treatment choice function:
\begin{align} \label{eqn:first-stage}
    T = 
    \begin{cases}
        m_0(R,U_0), R < \bar{r}, \\
        m_1(R,U_1), R \geq \bar{r},
    \end{cases}
\end{align}
where $\bar{r}$ is the cutoff value, and $U_0$ and $U_1$ are scalar variables,  representing other factors that are not observable to an econometrician. For easy reference, they will be referred to as the error terms hereafter. The important feature of the RD design is that the treatment varies discontinuously when the running variable crosses the cutoff $\bar{r}$. The functions $m_0$ and $m_1$ represent respectively the treatment choice mechanism when $R$ is below and above the cutoff. 

It is important to point out that the variables $U_0,U_1,T$ and $R$ are allowed to be correlated with the error term $\varepsilon$. If we assume $\varepsilon$ to be independent of $(T,R)$, then we can follow \cite{matzkin2003nonparametric} or \cite{hoderlein2007identification} to identify the structural function. If we assume $\varepsilon \perp R$ and $R$ is excluded from $m_0,m_1$ and $g$, then we can follow \cite{torgovitsky2015identification} to identify the structural function by treating the binary variable $\mathbf{1}\{R \geq \bar{r}\}$ as the instrument.

We make the following assumptions on the model imposed by (\ref{eqn:outcome-equation}) - (\ref{eqn:first-stage}). Let $\mathcal{G}$ be the set of candidate structural functions such that the true $g^*$ is contained in $\mathcal{G}$. That is, $\mathcal{G}$ is the infinite-dimensional parameter space where the structural function belongs to. Denote the conditional distribution function by $F_{\cdot|\cdot}(\cdot | \cdot)$, the conditional density function by $f_{\cdot|\cdot}(\cdot | \cdot)$, and the conditional quantile function by $F^{-1}_{\cdot|\cdot}(\cdot | \cdot)$. 


\begin{assumption} [Dual Monotonicity] \label{ass:dual-monotonicity} 
    \ \begin{enumerate}[label = (\roman*)] 
        \item Every $g \in \mathcal{G}$ satisfies that for each given $T=t$ and $R=r$, $g$ is strictly increasing in $\varepsilon$.
        \item For each given $R=r$, $m_0$ is strictly increasing in $U_0$ and $m_1$ is strictly increasing in $U_1$.
    \end{enumerate}
\end{assumption}

\begin{assumption} [Smoothness] \label{ass:smoothness} 
    
    \ \begin{enumerate} [label = (\roman*)] 
        \item The functions $m_0,m_1$ and every $g \in \mathcal{G}$ are continuous on their respective domains.
        \item The conditional quantile functions $F^{-1}_{U_0 | R}(u | r)$ and $F^{-1}_{U_1 | R}(u | r)$ are strictly increasing in $u$ and continuous in $(u,r)$.
        \item The conditional distribution functions $F_{\varepsilon | U_0,R}(e | u , r)$ and $F_{\varepsilon | U_1,R}(e | u , r)$ are strictly increasing in $e$ and are continuous in $r$ at $\bar{r}$.

        \item The running variable $R$ is absolutely continuous, and its density is strictly positive around the cutoff $\bar{r}$.
    \end{enumerate}
\end{assumption}

\begin{assumption} [Rank Similarity] \label{ass:rank-similarity}
    $U_0 | (\varepsilon,R=\bar{r}^-) $ has the same distribution as $U_1 | (\varepsilon,R=\bar{r}^+) $. That is, 
    \begin{align*}
       \lim_{r \uparrow \bar{r}} f_{U_0|\varepsilon,R}(u | e,r) = \lim_{r \downarrow \bar{r}} f_{U_1|\varepsilon,R}(u | e,r).
    \end{align*}
\end{assumption}

Assumption \ref{ass:dual-monotonicity} defines a one-to-one mapping between $(Y,T)$ and $(\varepsilon,U_0,U_1)$ for a given value of $R$. Assumption \ref{ass:smoothness} states that except for the discontinuity introduced in (\ref{eqn:first-stage}), everything else is assumed to be reasonably smooth. Assumption \ref{ass:rank-similarity} is similar to Assumption 3 in \cite{dong2021regression}. It imposes the rank similarity condition \citep{chernozhukov2005endogeneity} on $(U_0,U_1)$. 

The treatment choice functions $(m_0,m_1)$ are not identified. Rather than trying to identify them, it is more convenient to consider a normalization to a quantile representation. By using the monotonicity of $m_0$ and $m_1$ in Assumption \ref{ass:dual-monotonicity}(ii), we define
\begin{align} \label{eqn:def-U}
    U = \mathbf{1}\{R < \bar{r}\}  F_{U_0|R}(U_0 | R) + \mathbf{1}\{R \geq \bar{r}\}  F_{U_1|R}(U_1 | R) = F_{T | R}(T | R),
\end{align}
as the conditional rank of $T$ given $R$.\footnote{The second equality in Equation (\ref{eqn:def-U}) is proved in Lemma \ref{lm:quantile-representation}} Then the treatment choice model in (\ref{eqn:first-stage}) can be written as
\begin{align*}
    T = h(R,U) = 
    \begin{cases}
        h_0(R,U), R < \bar{r}, \\
        h_1(R,U), R \geq \bar{r},
    \end{cases}
\end{align*}
where 
\begin{align*}
    h_0(r,u) = m_0\big(r,F_{U_0|R}^{-1}(u | r)\big) \text{ and } h_1(r,u) = m_1\big(r,F_{U_1|R}^{-1}(u | r)\big) .
\end{align*}
By using $[r_0,r_1]$ to denote the support of $R$, we can write the domains of $h_0$ and $h_1$ respectively as $[r_0,\bar{r}] \times [0,1]$ and $[\bar{r},r_1] \times [0,1]$.

The following lemma shows that the function $h$ defined above is the conditional quantile function of $T$ given $R$, and the quantile representation is a valid normalization in the sense that it preserves the monotonicity and smoothness conditions. Consequently, the function $h$ (including both $h_0$ and $h_1$) and the rank $U = h^{-1}(R,T)$ are identified from the data, where $h^{-1}$ denotes the inverse of $h$ with respect to the second argument $U$. 

\begin{lemma} [Quantile Representation] \label{lm:quantile-representation}
    The following statements hold under Assumptions \ref{ass:dual-monotonicity} - \ref{ass:rank-similarity}:
    \begin{enumerate} [label = (\roman*)]
        \item $U \perp R$, $U | R \sim \text{Unif}[0,1]$, and $\mathbb{P}(T \leq h(R,u) | R) = u, u \in [0,1]$.
        \item For each $R=r$, $h_0$ and $h_1$ are strictly increasing in $U$.
        \item The functions $h_0$ and $h_1$ are continuous.
        \item The conditional distribution function $F_{\varepsilon|U,R}(e | u,r)$ is strictly increasing in $e$ and is continuous in $r$ at $\bar{r}$, that is,
        \begin{align*}
            \lim_{r \uparrow \bar{r}}F_{\varepsilon|U,R}(e | u,r) = \lim_{r \downarrow \bar{r}}F_{\varepsilon|U,R}(e | u,r), \text{ for every } (e,u).
        \end{align*}
    \end{enumerate}
\end{lemma}

After the normalization, $U$ is independent of $R$ but $U$ and $\epsilon$ are possibly correlated even after conditioning on $R$.
Let $F_{Y | T, R}$ be the conditional distribution function of $Y$ given $T$ and $R$. We define
\begin{align*}
    F^-_{Y | T ,R}(y | t, r) & = 
    \begin{cases}
        F_{Y | T ,R}(y | t, r), & \text{ if } r < \bar{r}, \\
        \lim_{r \uparrow \bar{r}} F_{Y | T ,R}(y | t, r), & \text{ if } r = \bar{r}.
    \end{cases} \\
    F^+_{Y | T ,R}(y | t, r) & = 
    \begin{cases}
        F_{Y | T ,R}(y | t, r), & \text{ if } r > \bar{r}, \\
        \lim_{r \downarrow \bar{r}} F_{Y | T ,R}(y | t, r), & \text{ if } r = \bar{r}.
    \end{cases} 
\end{align*}
The above left and right limits exist in view of Assumptions \ref{ass:dual-monotonicity} and \ref{ass:smoothness}. The following assumption states that the support of the unobserved $\varepsilon$ does not vary with $U$ or $R$. This invariance of the support is not strong since it still allows $\varepsilon$ to be correlated with $U$ or $R$ in any way. 

\begin{assumption} [Support Invariance] \label{ass:support-invariance}
     Supp$(\varepsilon | U=u, R=r)$ does not depend on $u$ or $r$ in the neighborhood of $\bar{r}$. This common support is denoted by $\mathcal{E}$.
\end{assumption}

\subsection{Nonparametric identification}

We derive an important implication of the model (\ref{eqn:outcome-equation}) - (\ref{eqn:first-stage}). This implication is the key to identification and estimation.
A function $g \in \mathcal{G}$ is said to satisfy Condition (\ref{eqn:refutable-implication}) if for every $e \in \mathcal{E}$ and $ u \in [0,1]$,
\begin{align}  \label{eqn:refutable-implication}
    F^-_{Y | T,R} (g(h_0(\bar{r},u),\bar{r},e) | h_0(\bar{r},u),\bar{r}) = F^+_{Y | T,R} (g(h_1(\bar{r},u),\bar{r},e) | h_1(\bar{r},u),\bar{r}).
\end{align}
If the function $g$ in Condition (\ref{eqn:refutable-implication}) is equal to the true $g^*$, then the left-hand side of (\ref{eqn:refutable-implication}) is equal to the conditional distribution of $\varepsilon$ given $U$ evaluated from the left side of the cutoff $\bar{r}$. Symmetrically, the right-hand side of (\ref{eqn:refutable-implication}) is equal to the conditional distribution of $\varepsilon $ given $U$ evaluated from the right side of the cutoff $\bar{r}$. Then the equality holds by the continuity of $F_{\varepsilon|U,R}(e|u,\cdot)$ stated in Lemma \ref{lm:quantile-representation}(iv). We summarize this result in the following lemma.

\begin{lemma} [Local Control Function] \label{lm:local-control}
    Under Assumptions \ref{ass:dual-monotonicity} - \ref{ass:rank-similarity}, $g^*$ satisfies Condition (\ref{eqn:refutable-implication}) with both sides of the equation equal to $F_{\varepsilon|U,R}(e| u,\bar{r})$.

\end{lemma}

When $g = g^*$, Condition (\ref{eqn:refutable-implication}) can be written as 
\begin{align*}
    \lim_{r \uparrow \bar{r}} \mathbb{P}(\varepsilon | T=h_0(r,u),R=r) = \lim_{r \downarrow \bar{r}} \mathbb{P}(\varepsilon | T=h_1(r,u),R=r)
\end{align*}
This leads to another interpretation of Lemma \ref{lm:local-control}: $U$ can serve as a control function local to the cutoff. After fixing the value of $U$, the variation in the treatment $T$ becomes locally exogenous. This is because given $U$ and $R$, the treatment $T$ becomes deterministic. The only variation left in $T$ around the cutoff is due to the discontinuity in the treatment choice function.
Lemma \ref{lm:local-control} is essentially a version of Lemma 1(i) in \citet{dong2021regression}. From the IV perspective, Lemma \ref{lm:local-control} corresponds to Theorem 1 in \cite{imbens2009identification}. It is also similar to Theorem 1 in \citet{torgovitsky2015identification} in that it provides a (necessary) characterization of the identified set of the structural function.\footnote{The identified set can be defined as the subset of $\mathcal{G}$ that contains the functions $g$ that can generate the observed distribution of $(Y,T,R)$. However, it is rather a detour to formally define such a set because in Section \ref{sec:estimation} we directly use Condition (\ref{eqn:refutable-implication}) for estimation.} 

For any $g \in \mathcal{G}$, Lemma \ref{lm:local-control} can be used to verify whether $g = g^*$. In particular, if 
\begin{align*} 
    F^-_{Y | U,R} (g(h_0(\bar{r},u),\bar{r},e) | u,\bar{r}) \ne F^+_{Y | U,R} (g(h_1(\bar{r},u),\bar{r},e) | u,\bar{r}), \text{ for some } e \text{ and } u,
\end{align*}
then $g$ can not be the true structural function. We further introduce some regularity conditions below.

\begin{assumption} \label{ass:strong-discontinuity} 
   \ \begin{enumerate} [label = (\roman*)]
        \item (Fuzzy RD). The support of $T | R$ from just below and above the cutoff are intervals denoted respectively by $\text{Supp}(h_0(\bar{r},U)) = [t_0',t_0''] \text{ and } \text{Supp}(h_1(\bar{r},U)) = [t_1',t_1''].$
        The two supports are overlapping: $[t_0',t_0''] \cap [t_1',t_1''] \ne \emptyset$.\footnote{Infinite intervals are also allowed. For example, $\text{Supp}(h_0(\bar{r},U))$ can be $(-\infty,t_0'']$, $[t_0',\infty)$, or $\mathbb{R}$. We use the notation $[t_0',t_0'']$ to represent all these cases. }
        \item (Strong Discontinuity). Local to the cutoff, the functions $h_0$ and $h_1$ intersects and only intersects finitely many times. That is, the following set is nonempty and finite:
        \begin{align*}
            \{ h_0(\bar{r},u): h_0(\bar{r},u) = h_1(\bar{r},u) \in [t_0',t_0''] \cap [t_1',t_1''], u \in [0,1] \}.
        \end{align*}
        
    \end{enumerate}
\end{assumption}

Assumption \ref{ass:strong-discontinuity} imposes restrictions on the nature of the discontinuity.
Assumption \ref{ass:strong-discontinuity}(i) requires that the RD design is fuzzy in that there are treatment levels that are taken both below and above the cutoff. Assumption \ref{ass:strong-discontinuity}(ii) imposes restrictions on the strength of the discontinuity. It requires that the conditional quantile functions $h_0(\bar{r},\cdot)$ and $h_1(\bar{r},\cdot)$ only intersects finitely many times. The two curves can intersect but not overlap. If the two functions $h_0(\bar{r},\cdot)$ and $h_1(\bar{r},\cdot)$ overlaps on some interval, then the structural function is not identified on that interval because there is no exogenous variation in the treatment inside that interval.\footnote{In that case, it is possible to partially identify the structural function. } In the extreme case where the two curves completely overlap, there is no discontinuity. 

With the above assumptions, we present the main identification result of the paper. The following theorem shows that Condition (\ref{eqn:refutable-implication}) identifies the true structural function up to a monotone transformation of the error term.

\begin{theorem} [Nonparametric Identification] \label{thm:id-rd-cutoff}
    Let Assumptions \ref{ass:dual-monotonicity} - \ref{ass:strong-discontinuity} hold. If $g \in \mathcal{G}$ satisfies Condition (\ref{eqn:refutable-implication}), then there exists a continuous and strictly increasing function $\lambda^g$ such that for every $t \in [t_0',t_0''] \cup [t_1',t_1''], e \in \mathcal{E},$ and $u \in [0,1]$, $g^*(t,\bar{r},e) = g(t,\bar{r},\lambda^g(e))$.

\end{theorem}

\begin{remark}
    Theorem \ref{thm:id-rd-cutoff} is proved based on the sequencing approaching developed in the proof of Theorem 2 in \cite{torgovitsky2015identification}. 
\end{remark}

Theorem \ref{thm:id-rd-cutoff} is the best one can achieve in terms of identifying the nonseparable structural function because the error term is unobserved. Any $g \in \mathcal{G}$ that satisfies Condition (\ref{eqn:refutable-implication}) is equally good as the true $g^*$. The only difference is that the error term is rescaled by the monotone transformation $\lambda^g$. Therefore, such a function $g$ can also be seen as a ``version'' of $g^*$.

Inspecting the conditions of Theorem \ref{thm:id-rd-cutoff}, we can see that no independence assumption is needed. This is why the RD design is often considered a more credible approach than instrumental variables for conducting causal inference. However, in many studies of RD designs, a local independence assumption is imposed, explicitly making the running variable exogenous around the cutoff. For example, Assumption A3(i) in \cite{hahn2001identification} requires $(\varepsilon,U)$ to be jointly independent of $R$ conditioning on $R$ near $\bar{r}$.
In the binary treatment case, \citet{dong2018alternative} shows that this local independence condition is not needed to achieve identification. 

Based on the observation made in Lemma \ref{lm:local-control}, we can recover the conditional distribution of $\varepsilon$ given $U$ and $R = \bar{r}$. For any $g \in \mathcal{G}$, if $g$ is the true structural function, then the corresponding conditional distribution of $\varepsilon$ is
\begin{align} \label{eqn:def-F-varepsilon-g}
    F_{\varepsilon | U,R}^g(e | u,\bar{r}) = F^-_{Y | T,R} (g(h_0(\bar{r},u),\bar{r},e) | h_0(\bar{r},u),\bar{r}) = F^+_{Y | T,R} (g(h_1(\bar{r},u),\bar{r},e) | h_1(\bar{r},u),\bar{r}).
\end{align}
In fact, the above conditional distribution $F_{\varepsilon | U,R}^g$ is a transformed version of the true conditional distribution $F_{\varepsilon | U,R}$, where the transformation is the $\lambda^g$ defined in Theorem \ref{thm:id-rd-cutoff}. This means that the conditional distribution of $\varepsilon | U,R=\bar{r}$ is identified up to the same monotone transformation as the structural function. To eliminate such inconvenience caused by the error term, we can integrate out $\varepsilon$ and obtain a unique \emph{conditional average structural function} (CASF):
\begin{align*}
    \beta^*(t) = \mathbb{E} [g^*(t,r,\varepsilon) | R=\bar{r}],
\end{align*}
where the expectation is taken with respect to the true conditional distribution of $\varepsilon$ given $ R=\bar{r}$. The following corollary summarizes the above discussion. 
\begin{corollary} [CASF] \label{cor:CASF-identification}
    Let Assumptions \ref{ass:dual-monotonicity} - \ref{ass:strong-discontinuity} hold. For any $g^{\circ} \in \mathcal{G}$ that satisfies Condition (\ref{eqn:refutable-implication}), let $\lambda^g$ be the transformation defined in Theorem \ref{thm:id-rd-cutoff}. The following two statements hold true.\footnote{Since the conditional distribution of $\varepsilon$ given $U,R=\bar{r}$ is identified. We can also identify other structural parameters, including the conditional quantile structural function.}
    \begin{enumerate} [label = (\roman*)]
        \item For every $e \in \mathcal{E}$ and $u \in [0,1]$, $F^{g^{\circ}}_{\varepsilon | U,R}(\lambda^{g^{\circ}}(e) | u,\bar{r}) = F_{\varepsilon | U,R}(e | u,\bar{r}).$
        \item For any $t \in [t_0',t_0''] \cup [t_1',t_1'']$, the CASF $\beta^*(t,\bar{r})$ is uniquely identified as
        \begin{align*}
            \beta^*(t) = \int {g^{\circ}}(t,\bar{r},e) dF^{g^{\circ}}_{\varepsilon | U,R}(e | u,\bar{r}) du.
        \end{align*}
    \end{enumerate}
\end{corollary}

The CASF $\beta^*(t)$ gives the average outcome the policy-maker can achieve when the treatment level for individuals with characteristic $R = \bar{r}$ is set to $t$. It is worth noting the difference between the CASF and the \emph{local average structural function} (LASF) commonly seen in the LATE literature. The LASF represents the average outcome for the so-called compliers, an unobservable subpopulation. Therefore, the policy-maker cannot assign treatment to the compliers even when the LASF is identified. On the other hand, the identified CASF can directly guide the treatment assignment to the subpopulation with $R = \bar{r}$. The derivative of the CASF is not the causal effect specific to any subpopulation. Following the spirit of, for example, \cite{heckman2001policy}, we may call CASF a policy-relevant parameter.

\section{Semiparametric estimation} \label{sec:estimation}

    In this section, we consider parametrizations of the structural function that maintain nonlinearity and nonseparability. We propose a semiparametric estimation procedure and derive its large-sample properties. The estimator is semiparametric because the structural function is parametrically specified, while the treatment choice model is left nonparametrically specified.
    
    We do not consider a fully nonparametric estimator since such a procedure can be too data-demanding for practical use, which is especially true for the RD design since the estimation is in the local neighborhood of the cutoff.\footnote{From the theoretical perspective, it can be challenging to construct a fully nonparametric estimator. If we follow the sieve approach, for example, we would need to consider a basis of functions that are strictly increasing in one of the arguments to accommodate the monotonicity of the structural function, which is a non-trivial task.}


\subsection{Construction of the estimator}

    Consider the following parametrization of $\mathcal{G}$ local to the cutoff.

    \begin{assumption} [Semiparametric Specification]
    There is a one-to-one mapping from the class $\{ (t,e) \mapsto g(t,\bar{r},e): g \in \mathcal{G} \}$ of functions to a finite-dimensional parameter space $\Gamma \subset \mathbb{R}^{d_\Gamma}$. We write such parametrization as $\{g_\gamma(\cdot,\bar{r},\cdot): \gamma \in \Gamma\}$.
    Assume this parametric model is correctly specified, that is, there exists $\gamma^* \in \Gamma$ such that $g_{\gamma^*}(\cdot,\bar{r},\cdot) = g^*(\cdot,\bar{r},\cdot)$.
    \end{assumption}

    \begin{assumption} [Normalization of $\Gamma$] \label{ass:normalization-Gamma}
        For any $\gamma,\gamma' \in \Gamma$, if there exists a transformation $\lambda$ such that
        \begin{align*}
            g_\gamma(\cdot,\bar{r},\cdot) = g_{\gamma'}(\cdot,\bar{r},\lambda(\cdot)),
        \end{align*}
        then $\gamma = \gamma'$ and $\lambda$ is the identity transformation.
    \end{assumption}
    Assumption \ref{ass:normalization-Gamma} is a normalization condition that fixes the scale of the error term $\varepsilon$. An example is provided below to illustrate the parametrization of the structural function. 
    One way to achieve such normalization is to have some treatment value $\tilde{t}$ such that $g_\gamma(\tilde{t},\bar{r},e) = e, \text{ for all } \gamma \in \Gamma.$ 
    
    \begin{example} \label{eg:1}
        Let $\tilde{t} \in [t_0',t_0''] \cap [t_1',t_1'']$. We can specify the structural function by
    \begin{align*}
        g_\gamma(T,\bar{r},\varepsilon) = \gamma_1 (T-\tilde{t}) + \gamma_2 (T-\tilde{t})^2 + \gamma_3 (T-\tilde{t})\varepsilon + \varepsilon.
    \end{align*}
    The parameter $\gamma = (\gamma_1,\gamma_2,\gamma_3)$ is three-dimensional. The function $g_\gamma(T,\bar{r},\varepsilon)$ is strictly increasing in $\varepsilon$ for all $\gamma$ satisfying $\mathbb{P}(\gamma_3(T-\tilde{t})+1 > 0)=1$. The parametrization satisfies Assumption \ref{ass:normalization-Gamma} because, by construction, $g_{\gamma}(\tilde{t},\bar{r},e) = e$ for all values of $\gamma$ and $e$.\footnote{This normalization strategy is presented in Equation (2.5) of \citet{matzkin2003nonparametric}.} The model is quadratic in $T$ and nonseparable between $T$ and $\varepsilon$. The effect of $T$ on $Y$ is allowed to be nonlinear and contain unobserved heterogeneity. The distribution of $\varepsilon$ is not parametrized and thus can be very general.
    \end{example}

    The true parameter $\gamma^*$ in the normalized semiparametric model can be identified as follows. We use $h^* = (h_0^*,h_1^*)$ to signify the true conditional quantile functions and $h = (h_0,h_1)$ a generic pair of conditional quantile functions. Let $w(e,u)$ be a weighting function defined on $\mathbb{R} \times [0,1]$.
    Define the criterion function as 
    \begin{align*}
        \left\lVert D_{\gamma,h} \right\rVert_{w} = \left( \int_0^1 \int_{\mathbb{R}} | D_{\gamma,h}(e,u)|^2 w(e,u) d edu \right)^{1/2} ,
    \end{align*}
    where $D_{\gamma,h}(e,u)$ is defined to be
    \begin{align} \label{eqn:D-gammah-def}
         \int_0^u \left( F^-_{Y|T,R}( g_\gamma(h_0(\bar{r},v),\bar{r},e) | h_0(\bar{r},v),\bar{r} ) - F^+_{Y|T,R}( g_\gamma(h_1(\bar{r},v),\bar{r},e) | h_1(\bar{r},v),\bar{r} ) \right) dv.
    \end{align}
    This criterion function is based on Equation (\ref{eqn:refutable-implication}), which by Lemma \ref{lm:local-control} is a necessary characterization of the identified set. We take an integral form of Condition (\ref{eqn:refutable-implication}) because it gives a faster convergence rate of the resulting estimator.
    
    \begin{assumption} [Weighting function $w$] \label{ass:w}
        The function $w$ is nonnegative, integrates to one, and is bounded on $\mathbb{R} \times [0,1]$. The support of $w$ contains $\mathcal{E} \times [0,1]$. 
    \end{assumption}
    In practice, we may use a weighting function $w$ that is supported on the entire domain $\mathbb{R} \times [0,1]$ since $\mathcal{E}$ is unknown.
    The following corollary provides the semiparametric identification result, which is based on the nonparametric identification result in Section \ref{sec:identification}. It shows that the criterion function, when evaluated at the true nuisance parameter value $h^*$, is uniquely minimized by the true $\gamma^*$.

    \begin{corollary} [Semiparametric Identification] \label{cor:gamma-star-identification}
        Let Assumptions \ref{ass:dual-monotonicity} - \ref{ass:w} hold.
        For any $\gamma \in \Gamma$ such that $\gamma \ne \gamma^*$, we have $\norm{ D_{\gamma,h^*} }_w > \norm{ D_{\gamma^*,h^*} }_w = 0.$
        In other words, $\gamma^*$ is the unique minimizer of $\norm{ D_{\gamma^*,h^*} }_w$.
    \end{corollary}

    Assume there is an \emph{independent and identically distributed} (iid) sample $(Y_i,T_i,R_i)_{i=1}^n$ available. We propose an estimation procedure based on the above semiparametric identification result. The idea is that we first estimate the nonparametric components $(h_0,h_1)$ and $(F^-_{Y|T,R},F^+_{Y|T,R})$ that appear in the criterion function. Then we construct an empirical version of the criterion function and take its minimizer to be the estimator.
    
    The estimation procedure of $\gamma$ is more specifically divided into three steps. The first step is to estimate the conditional quantile functions $h_0$ and $h_1$. 
    The second step uses \emph{local linear regression} (LLR) to estimate the conditional distributions $F^-_{Y | T,R}$ and $F^+_{Y | T,R}$. It is standard to use local polynomials in the estimation of RD designs \citep{porter2003estimation,sun2005adaptive}. The difference is that classical RD methods use local polynomial to estimate the conditional expectation function of $Y$ given $R$ while we estimate the conditional distribution of $Y$.\footnote{Local linear estimation of the conditional distribution function can be found in \cite{hansen2004nonparametric,xie2021uniform}.}  
    The third step constructs an estimate of the criterion function by replacing the nonparametric nuisance parameters in (\ref{eqn:D-gammah-def}) by their estimated counterparts and then finds the estimate of $\gamma^*$ by minimizing the estimated criterion function. We describe the detail of the estimation procedure as follows. Denote $\mathcal{Y}$ as the range of the outcome $Y$. 

    \begin{itemize}
        \item \textbf{STEP 1.} 
        Choose estimators $\hat{h}_0(\bar{r},\cdot)$ and $\hat{h}_1(\bar{r},\cdot)$ of the corresponding conditional quantile processes, $h_0(\bar{r},\cdot)$ and $h_1(\bar{r},\cdot)$. Specific constructions are provided in Section \ref{ssec:quantile-est}.

        \item \textbf{STEP 2.} Choose two bandwidth sequences $b_1 = b_{1n}$ and $b_2 = b_{2n}$ and three kernel functions $k_Y$, $k_T$, and $k_R$. Define $K_Y(y) = \int_{-\infty}^y k_Y(\tilde{y}) d \tilde{y} $. For each $y \in \mathcal{Y}$ and $t \in [t_0',t_0'']$, solve the following minimization problem:
        \begin{align*} 
            \min_{a^-,a_{T}^-,a_{R}^-} \sum_{i:R_i < \bar{r}} & \left( K_Y\left( \frac{y - Y_i}{b_2} \right) - a^- - a_{T}^-(T_i-t) - a_{R}^-(R_i - \bar{r}) \right)^2  \\
            & \times k_T\left( \frac{T_i - t}{b_1} \right) k_R\left( \frac{R_i - \bar{r}}{b_1} \right).
        \end{align*}
        The minimizer $\hat{a}^-$ is the estimate $\hat{F}^-_{Y | T,R}(y | t,\bar{r})$. For each $y \in \mathcal{Y}$ and $t \in [t_1',t_1'']$, solve the following minimization problem:
        \begin{align*}
            \min_{a^+,a_{T}^+,a_{R}^+} \sum_{i:R_i \geq \bar{r}} & \left( K_Y\left( \frac{y - Y_i}{b_2} \right) - a^+ - a_{T}^+(T_i-t) - a_{R}^+(R_i - \bar{r}) \right)^2 \\
            & \times k_T\left( \frac{T_i - t}{b_1} \right) k_R\left( \frac{R_i - \bar{r}}{b_1} \right).
        \end{align*}
        The minimizer $\hat{a}^+$ is the estimate $\hat{F}^+_{Y | T,R}(y | t,\bar{r})$. 

        \item \textbf{STEP 3.} Construct the empirical version of the criterion function:
        \begin{align*}
            \left\lVert \hat{D}_{\gamma,\hat{h}} \right\rVert_{w} = \left(  \int_0^1 \int_{\mathbb{R}} | \hat{D}_{\gamma,\hat{h}}(e,u)|^2 w(e,u) dedu \right)^{1/2} ,
        \end{align*}
        where $\hat{D}_{\gamma,\hat{h}}(e,u) $ is defined to be
        \begin{align*}
             \int_0^u \left( \hat{F}^-_{Y|T,R}( g_\gamma(\hat{h}_0(\bar{r},v),\bar{r},e) | \hat{h}_0(\bar{r},v),\bar{r} ) - \hat{F}^+_{Y|T,R}( g_\gamma(\hat{h}_1(\bar{r},v),\bar{r},e) | \hat{h}_1(\bar{r},v),\bar{r} ) \right) dv.
        \end{align*}
        The estimator $\hat{\gamma}$ is any parameter value in $\Gamma$ that satisfies
        \begin{align} \label{eqn:gamma-hat-def}
            \big\lVert \hat{D}_{\hat{\gamma},h} \big\rVert_{w} \leq \inf_{\gamma \in \Gamma} \big\lVert \hat{D}_{\gamma,h} \big\rVert_{w} + O_p\left(\alpha_n\right),
        \end{align}
        where $\alpha_n \rightarrow 0$ is specified by Equation (\ref{eqn:alpha-n}) in Appendix \ref{sec:proof-est}.
    \end{itemize}

    \subsection{Asymptotic normality}

    More regularity assumptions are imposed for the estimator to enjoy desirable statistical properties. Define
    \begin{alignat*}{3}
        f^-_{T | R}(t | r) & = 
        \begin{cases}
            f_{T | R}(t | r), & \text{ if } r < \bar{r}, \\
            \lim_{r \uparrow \bar{r}} f_{T | R}(t | r), & \text{ if } r = \bar{r}.
        \end{cases} \quad
        f^+_{T | R}(t | r) & = 
        \begin{cases}
            f_{T | R}(t | r), & \text{ if } r > \bar{r}, \\
            \lim_{r \downarrow \bar{r}} f_{T | R}(t | r), & \text{ if } r = \bar{r}.
        \end{cases}
    \end{alignat*}
    The above left and right limits exist in view of Assumptions \ref{ass:dual-monotonicity} and \ref{ass:smoothness}.

    \begin{assumption} [Distributions of $Y,T,$ and $R$] \label{ass:smoothness-distribution}
        \ \begin{enumerate} [label = (\roman*)]
            \item The support of $T$ does not vary with $R$ except when crossing the cutoff $\bar{r}$, i.e., Supp$(T | R=r) = [t_0',t_0'']$ for $r < \bar{r}$ and Supp$(T | R=r) = [t_1',t_1'']$ for $r > \bar{r}$. The density functions $f^-_{T,R}$ and $f^+_{T,R}$ are bound away from zero. 
            
            \item The density functions $f^-_{T,R}$ and $f^+_{T,R}$ are twice continuously differentiable, and $\frac{\partial^2}{\partial t^2} f^-_{T,R}(t,\bar{r})$ and $\frac{\partial^2}{\partial t^2} f^+_{T,R}(t,\bar{r})$ are Lipschitz continuous with respect to $t$. 
            \item The support of $Y$, $\mathcal{Y}$, is compact. The conditional distribution functions $F^-_{Y | T,R}$ and $F^+_{Y | T,R}$ are three-times continuously differentiable over $\mathcal{Y} \times [t_0',t_0''] \times [r_0,\bar{r}]$ and $\mathcal{Y} \times [t_1',t_1''] \times [\bar{r},r_1]$, respectively. 
        \end{enumerate}
    \end{assumption}

    \begin{assumption} [Complexity of the Parametric Model] \label{ass:complexity-parametric-model}
        The parametrization $\{g_\gamma(\cdot,\bar{r},\cdot):\gamma \in \Gamma\}$ satisfies the following conditions.
        \begin{enumerate} [label = (\roman*)]
            \item The parameter space $\Gamma$ is compact.
            \item The class of functions $\{ T \mapsto g_\gamma(T+v,\bar{r},e): \gamma \in \Gamma, v \in (-1,1), e \in \mathcal{E} \}$ is finite-dimensional.
            \item The function $g_\gamma(t,\bar{r},e)$ is twice continuously differentiable over $\gamma \in \Gamma$, $t \in [t_0',t_0''] \cup [t_1',t_1'']$, and $ e \in \mathcal{E}$.
            \item The gradient $\nabla_\gamma D_{\gamma^*,h^*}(e,u)$ is a vector of linearly independent functions of $(e,u)$.

        \end{enumerate}
    \end{assumption}

    \begin{assumption} [Kernels] \label{ass:kernels}
        \ \begin{enumerate} [ label = (\roman*)]
            \item  The kernel functions $k_T$ and $k_R$ are (1) supported on $[-1,1]$, (2) strictly greater than zero in the interior of the support, (3) of bounded variation, (4) continuously differentiable on $\mathbb{R}$. 
            \item The kernel function $k_Y$ is (1) nonnegative and (2) integrable on $\mathbb{R}$ with $\int k_Y(y) dy =1$ and satisfies (3) $\int y k_Y(y) dy = 0$. 
        \end{enumerate}

    \end{assumption}

    \begin{assumption} [Bandwidth] \label{ass:bandwidth} The bandwidth $b_1$ and $b_2$ satisfy the following conditions:
        \begin{enumerate} [label = (\roman*)]
            \item $b_1 \asymp b_2$.\footnote{The notation $b_1 \asymp b_2$ means that there exists $C>1$ such that $b_1/b_2 \in [1/C,C]$.}
            \item $(n \log n) b_1^6 = o(1)$.
            \item $n b_1^{\frac{13}{3} + \epsilon} \rightarrow \infty$, for some sufficiently small $\epsilon>0$.
        \end{enumerate}

    \end{assumption}

    \begin{assumption} [First-step Conditional Quantile Estimators] \label{ass:h-tilde}
        The estimators $\hat{h}_0$ and $\hat{h}_1$ satisfy the following conditions.
        \begin{enumerate} [label = (\roman*)]
            \item Monotonicity and smoothness: for every $n$ sufficiently large, there exist $C > 0$ and deterministic and finite partitions $\mathcal{P}_0^n$ and $\mathcal{P}_1^n$ on $(0,1)$ such that
            \begin{align*}
                \mathbb{P}\left(\hat{h}_0(\bar{r},\cdot) \notin \mathcal{H}_0(\mathcal{P}_0^n) \right), \mathbb{P}\left(\hat{h}_1(\bar{r},\cdot) \notin \mathcal{H}_1(\mathcal{P}_1^n) \right) = O \big( \sqrt{b_1}\big),
            \end{align*}
            where
            \begin{align*}
                \mathcal{H}_0(\mathcal{P}_0^n) = &\{ \text{ function } h \text{ from } [0,1] \text{ into } [t_0',t_0'']: \text{on each element of $\mathcal{P}_0^n$}, h \text{ is strictly} \\
                & \quad \text{increasing, its inverse $h^{-1}$ is three-times continuously differentiable, } \\ 
                & \quad \text{and $(h^{-1})^{(3)}$ is Lipschitz continuous } \}, 
            \end{align*}
            and $\mathcal{H}_1(\mathcal{P}_1^n)$ is defined analogously by replacing $\mathcal{P}_0^n$ with $\mathcal{P}_1^n$.
            \item Uniform Bahadur representation:
            \begin{align*}
                \hat{h}_0(\bar{r},u) - h_0^*(\bar{r},u) & = b_1^2 \nu_0(u) + O_p ( b_1^3)\\ 
                & \quad + \frac{1}{nb_1} \sum_{i=1}^n q_{0}(T_i,R_i;u) k_{Q,0}\left( \frac{R_i - \bar{r}}{b_1} \right) \mathbf{1}\{R_i < \bar{r}\} + o_p \big( 1/\sqrt{nb_1} \big),  \\
                \hat{h}_1(\bar{r},u) - h_1^*(\bar{r},u) & = b_1^2 \nu_1(u) + O_p ( b_1^3)\\
                & \quad + \frac{1}{nb_1} \sum_{i=1}^n q_{1}(T_i,R_i;u) k_{Q,1}\left( \frac{R_i - \bar{r}}{b_1} \right) \mathbf{1}\{R_i \geq \bar{r}\} + o_p \big(1/\sqrt{nb_1} \big), 
            \end{align*}
            uniformly over $u \in (0,1)$. The functions $\nu_0$ and $\nu_1$ are bounded. The functions $q_0$ and $q_1$ are (1) bounded, (2) centered, that is, $\mathbb{E}[q_{0}(T,R;u)|T,R] = \mathbb{E}[q_{1}(T,R;u)|T,R] = 0$, and (3) does not vary with $n$. The functions $k_{Q,0}$ and $k_{Q,1}$ are bounded.
            \item Uniform convergence rate:
            \begin{align*}
                \lVert \hat{h} - h^* \rVert_\infty & = \sup_{ u \in (0,1)} |\hat{h}_0(\bar{r},u) - h_0^*(\bar{r},u)| \vee |\hat{h}_1(\bar{r},u) - h_1^*(\bar{r},u)| \\
                & = O_p \left( \sqrt{\log n / (n b_1)} + b_1^2 \right).
            \end{align*}
            
        \end{enumerate}

    \end{assumption}

    A brief discussion of the assumptions is in order. Assumption \ref{ass:smoothness-distribution} imposes smoothness restrictions on the joint distribution of $(Y,T,R)$. In the previous section, the identification result only requires continuity of the relevant functions. For estimation, we need higher-order smoothness regarding the distribution functions. Assumption \ref{ass:complexity-parametric-model} imposes restrictions on the parametric model of the structural function. Part (ii) restricts the complexity of the model. Part (iii) imposes high-order smoothness on the structural function. Part (iv) is similar to Assumption D4 in \cite{TORGOVITSKY2017minimum} and requires that $\nabla_\gamma D_{\gamma^*,h^*}$ to carry information about each component of the parameter. 
 
 Assumption \ref{ass:kernels} imposes restrictions on the kernel functions $k_T$, $k_R$, and $k_Y$. The differentiability is needed to prove a stochastic equicontinuity condition. Assumption \ref{ass:bandwidth} restricts that $b_1$ and $b_2$ are of the same asymptotic order, which is slightly faster than $n^{-1/6}$ and slightly slower than $n^{-3/13}$. This assumption is not restrictive and allows for the asymptotic mean squared error (AMSE) optimal bandwidth as well as undersmoothing.

 Assumption \ref{ass:h-tilde} imposes high-level restrictions on the first-stage nonparametric conditional quantile estimators. Part (i) assumes that the quantile estimators are piece-wise monotonic and smooth with a high probability. Part (ii) and (iii) give the uniform Bahadur representation and the uniform convergence rate, which are fairly standard in the quantile estimation literature. In Section \ref{ssec:quantile-est}, we discuss a specific nonparametric quantile estimator that satisfies Assumption \ref{ass:h-tilde}.

    \begin{theorem} [Asymptotic Distribution of the Semiparametric Estimator] \label{thm:estimation}
        Let Assumptions \ref{ass:dual-monotonicity} - \ref{ass:h-tilde} hold. Then $\norm{\hat{\gamma} - \gamma^*}_2 = O_p(b_1^2 + 1/\sqrt{nb_1})$ and
        \begin{align*}
            \big( \sqrt{nb_1} (\Sigma_- + \Sigma_+)^{-1/2} \big) (\Delta (\hat{\gamma} - \gamma^*) - b_1^2( B_- - B_+)) \overset{d}{\rightarrow} N(0,\bm{I}_{d_\Gamma}),
        \end{align*}
        where $I_{d_\Gamma}$ is the $d_\Gamma$-dimensional identity matrix.
        The exact forms of $\Delta$, $B_-$, $B_+$, $\Sigma_-$ and $\Sigma_+$ are given in Equations (\ref{eqn:nabla-D}), (\ref{eqn:bias-}), (\ref{eqn:bias+}), (\ref{eqn:sigma-}) and (\ref{eqn:sigma+}) in Appendix \ref{sec:proof-est}, respectively.

    \end{theorem}

    \begin{remark}
        The convergence rate of $\hat{\gamma}$ is $b_1^2 + 1/\sqrt{nb_1}$, which is equal to $n^{-2/5}$ when $b_1 \asymp n^{-1/5}$. This rate is the same as the one obtained in the classical RD design with a binary treatment variable \citep{hahn2001identification}. Having a continuous treatment does not slow down the convergence rate of the estimator in this case. This is due to the integral smoothing in the definition of the criterion function $D_{\gamma,h}$ in (\ref{eqn:D-gammah-def}).
    \end{remark}

    \begin{remark}
        The proof of Theorem \ref{thm:estimation} follows the general steps of proving asymptotic normality of semiparametric estimators as in, for example, \cite{TORGOVITSKY2017minimum} and \cite{chen2003estimation}. The main difficulty is that the usual stochastic equicontinuity condition is not sharp because the criterion function is nonparametrically estimated. To overcome this issue, we first use the empirical process theory to derive a uniform convergence rate for the estimated criterion function, which gives an initial bound on the convergence rate of $\hat{\gamma}$. A sharper stochastic equicontinuity result is then derived based on this initial bound together with more applications of the empirical process theory. This sharper stochastic equicontinuity result helps demonstrate that the usual linearization of the criterion function is valid. See the proof in Appendix \ref{sec:proof-est} for details.
    \end{remark}

    Once the asymptotic normal distribution of the estimator $\hat{\gamma}$ is established, we can conduct inference for $\gamma^*$. Theorem \ref{thm:estimation} together with the undersmoothing condition that $nb_1^5 = o(1)$ gives that
    \begin{align*}
        \sqrt{nb_1} (\hat{\gamma} - \gamma^*)  \overset{d}{\rightarrow} N(0,\Delta^{-1}(\Sigma_- + \Sigma_+)\Delta^{-1}).
    \end{align*}
    A linear null hypothesis regarding $\gamma$ can be written as $H \gamma = \eta$, where $\eta \in \mathbb{R}^{d_\eta}$ and $H$ is a $d_{\eta} \times d_\gamma$ full-rank matrix. Consider the  test statistic
    \begin{align*}
        nb_1 (\hat{\gamma} - \gamma^*)' \left( H \hat{\Delta}^{-1}(\hat{\Sigma}_- + \hat{\Sigma}_+)\hat{\Delta}^{-1} H' \right)^{-1} (\hat{\gamma} - \gamma^*),
    \end{align*}
    where $\hat{\Delta}$, $\hat{\Sigma}_-$, and $\hat{\Sigma}_+$ are consistent estimators of $\Delta$, $\Sigma_-$, and $\Sigma_+$, respectively. 
    By Slutsky's theorem, the above test statistic converges in distribution to the $\chi^2$ distribution with ${d_\eta}$ degrees of freedom. In Appendix \ref{sec:proof-est}, we discuss how to construct consistent estimators for $\Delta$, $\Sigma_-$, and $\Sigma_+$.

    \subsection{First-step nonparametric quantile estimators} \label{ssec:quantile-est}

    This section discusses how to construct nonparametric conditional quantile estimators that satisfy Assumption \ref{ass:h-tilde}. Consider the following two-step estimation procedure introduced by \cite{QU2015nonparametric}. Define $\rho_u(t) = t(u - \mathbf{1}\{t < 0\})$. 

    \begin{itemize}
        \item \textbf{STEP 1.} Choose a bandwidth sequence $b_3 = b_{3n} = o(1)$ and a kernel function $k_{\textit{FS}}$. Partition the unit interval $(0,1)$ into a grid of equally spaced points $\{u_1,\cdots,u_{J_n}\}$, where $J_n/(nb_3)^{1/4} \rightarrow \infty$. Solve the following optimization problem:
        \begin{align} \label{eqn:first-step-quantile-est}
            \min_{\{\mathrm{h}_j,\mathrm{h}'_j\}_{j=1}^{J_n}} \sum_{j=1}^{J_n} \sum_{i=1}^n \rho_{u_j} \left( T_i - \mathrm{h}_j - \mathrm{h}_j' (R_i - \bar{r}) \right) k_{\textit{FS}} \left( \frac{R_i - \bar{r}}{b_3} \right) \mathbf{1}\{ R_i < \bar{r} \}.
        \end{align}
        Denote the minimizers by $(\hat{h}_0(\bar{r},u_1),\cdots,\hat{h}_0(\bar{r},u_{J_n}))$.

        \item \textbf{STEP 2.} Let $u_0 = 0$ and $u_{J_{n+1}}=1$. Let $\hat{h}_0(\bar{r},u_0) = \min_{i:R_i < \bar{r}} T_i$ and $\hat{h}_0(\bar{r},u_{J_n + 1}) = \max_{i:R_i < \bar{r}} T_i$. Linearly interpolate between the estimates to obtain an estimate for the entire quantile process. That is, for any $u \in (u_j,u_{j+1})$, define
        \begin{align*}
            \hat{h}_0(\bar{r},u) = \frac{u_{j+1} - u}{u_{j+1} - u_j} \hat{h}_0(\bar{r},u_j) + \frac{u - u_j}{u_{j+1} - u_j} \hat{h}_0(\bar{r},u_{j+1}).
        \end{align*}
    \end{itemize}
    The estimator $\hat{h}_1(\bar{r},\cdot)$ can be analogously defined by using the data with $R_i \geq \bar{r}$.\footnote{There are three estimators of conditional quantile process in \cite{QU2015nonparametric}. The estimator explained here is their second one, denoted by $\hat{\alpha}^*$ in that paper. Their third estimator imposes a monotonicity constrain to the minimization problem (\ref{eqn:first-step-quantile-est}).}
    We can verify Assumption \ref{ass:h-tilde} for the estimator constructed above. Denote 
    \begin{align*}
        \Omega_{Q,0} & = \int (1,x) (1,x)' \mathbf{1}\{x < 0\} k_{\textit{FS}}(x) dx, \\
        \Omega_{Q,1} & = \int (1,x) (1,x)' \mathbf{1}\{x \geq 0\} k_{\textit{FS}}(x) dx.
    \end{align*}
    \begin{proposition} \label{prop:fs-quantile}
        Let Assumptions \ref{ass:smoothness}(iv) and \ref{ass:smoothness-distribution}(i)-(ii) hold. Assume that the third-order derivatives $\frac{\partial^3}{\partial R^3} h_0^*(r,u)$ and $\frac{\partial^3}{\partial R^3} h_1^*(r,u)$ are Lipschitz continuous respectively on $[r_0,\bar{r}] \times [0,1]$ and $[\bar{r},r_1] \times [0,1]$. Assume that the bandwidth $b_3 = cb_1$ for some constant $c>0$. Assume that the kernel $k_{\textit{FS}}$ is nonnegative, of bounded variation, compactly supported, having finite first-order derivatives and satisfying 
        \begin{align*}
            \int k_{\textit{FS}}(x) dx= 1, \int x k_{\textit{FS}}(x) dx= 0, \int x^2 k_{\textit{FS}}(x) dx < \infty.
        \end{align*}
        Then the estimators $\hat{h}_0(\bar{r},\cdot)$ and $\hat{h}_1(\bar{r},\cdot)$ described above satisfy Assumption \ref{ass:h-tilde}.
        The specific forms of $\nu_0$ and $\nu_1$ are 
        \begin{align*}
            \nu_0(u) & = \frac{c^2}{2} \frac{\partial^2}{\partial r^2} h_0^*(\bar{r},u) \iota' \Omega_{Q,0}^{-1} \int x^2  (1,x)' k_{\textit{FS}}(x) \mathbf{1}\{x < 0\} dx, \\
            \nu_1(u) & = \frac{c^2}{2} \frac{\partial^2}{\partial r^2} h_1^*(\bar{r},u) \iota' \Omega_{Q,1}^{-1} \int x^2  (1,x)' k_{\textit{FS}}(x) \mathbf{1}\{x \geq 0\} dx.
        \end{align*}
        The specific forms of $k_{Q,0}$ and $k_{Q,1}$ are 
        \begin{align*}
            k_{Q,j}(x) = \iota' \Omega_{Q,j}^{-1} (1,x/c)' k_{\textit{FS}}(x/c)/c, j =0,1.
        \end{align*}
        The functions $q_0$ and $q_1$ are
        \begin{align*}
            q_0(T,R;u) = (u - \mathbf{1}\{ T \leq h_0^*(\bar{r},u) \}) / (f_R(\bar{r}) f^-_{T | R}(h_0^*(\bar{r},u) | \bar{r}) ), \\
            q_1(T,R;u) = (u - \mathbf{1}\{ T \leq h_1^*(\bar{r},u) \}) / (f_R(\bar{r}) f^+_{T | R}(h_1^*(\bar{r},u) | \bar{r}) ),
        \end{align*}
        which take the form of an influence function for quantiles. 
    \end{proposition}

    Other quantile estimation methods are also available. For example, one can consider the generic framework proposed by \cite{chernozhukov2010quantile} for rearrangement. In particular, they show that the rearrangement of a preliminary estimated quantile process delivers a monotonic estimator that preserves the asymptotic properties. This result gives a different way to generate estimators that satisfy Assumption \ref{ass:h-tilde}. We can start with an estimator with desired asymptotic properties that give rise to Assumption \ref{ass:h-tilde}(ii) and (iii), and then apply the rearrangement procedure. The resulting estimator would be monotonic on the entire domain, and partitioning is unnecessary.


    \section{Numerical results} \label{sec:numeric}

    This section presents the empirical application and the simulation studies. The empirical study shows that the semiparametric estimator is considerably better than the simple TSLS estimator in discovering quantitative information regarding the structural function. The simulation studies show that the semiparametric procedure can accurately estimate the parameters with a moderate sample size.\footnote{Replication files for the empirical and simulation studies are available from the author upon request.}

    \subsection{Empirical application}

    In the empirical study, we examine the causal effect of sleep time on health status by exploiting the discontinuity in the timing of natural light at time zone boundaries. The unit of observation is the individual in the American Time Use Survey (ATUS), the outcome $Y$ is the individual's health status measured by the body-mass index (BMI),\footnote{BMI is a person's weight in kilograms divided by the square of height in meters. The Centers for Disease Control and Prevention define overweight as BMI > 25 and obesity as BMI > 30.} the treatment $T$ is the sleep time, and the running variable $R$ is the longitudinal distance to the nearest time zone boundary, with cutoff $\bar{r} = 0$ denoting the time zone border. As explained in the introduction, the identification is based on the exogenous variation in the sleep time around the time zone boundary. This exogenous variation is due to the difference in the timing of natural light on each side of the time zone boundary.

 Many studies in the medical literature examine the effect of sleep time on overweight issues. See \cite{beccuti2011sleep} and the references therein. These studies typically use survey or laboratory data. This problem is first studied by using the RD design in \cite{Giuntella2019sunset}.\footnote{\cite{Giuntella2019sunset} study many health and economics-related issues. Here we only mention the relevant ones.} The relevant outcome variable they use is a binary indicator of the obesity (or overweight) status indicating whether the BMI is above some threshold. They use the TSLS procedure to estimate a linear structural function. We consider two improvements based on their work. First, we directly use BMI as the outcome variable, providing a more quantitative measure of the health status. Second, we use the proposed semiparametric estimator to estimate a nonlinear structural function. Previous medical studies have provided evidence of the nonlinearity of the structural function. For example, \cite{hairston2010sleep} show that both undersleeping and oversleeping lead to an increase in BMI while sleeping around 8 hours leads to a more healthy BMI level.

 The data for this empirical application is collected from IPUMS CPS \citep{ipums-cps} and IPUMS ATUS \citep{ipums-atus} during the periods 2006 - 2008 and 2014 - 2016. By linking these datasets, we can locate the county where the individual lives and then use the county's centroid as the location of the individual. We focus on counties near the time zone boundary between the Eastern and Central time zone. The counties are divided into two regions based on their latitude. We estimate the model separately for each region.
 
 The estimated marginal effects of sleep on BMI from the semiparametric estimator and the TSLS estimator are shown in Figure \ref{fig:effects-empirical}. Several interesting findings are observed based on the semiparametric estimates. First, the marginal effects are increasing and increase from negative to positive. This lends some support to the previous argument that the structural function is nonlinear and neither sleeping too little nor too much is preferable. Second, we can determine the optimal (in terms of BMI) sleep time by finding the zero of the marginal effect curve. In both cases, the optimal sleep time is between 7 and 8 hours, which also aligns with the findings in previous medical studies. Third, the results from the two regions are similar, meaning that the variation across different latitudes is small. 
 
From Figure \ref{fig:effects-empirical}, we can also see that the TSLS estimates are not capable of demonstrating the above results. First, the TSLS procedure only provides a constant estimate of the marginal effect across all levels of sleep time. This means an extra hour of sleep would lead to the same effect on health regardless of the person's current sleep time, which is inappropriate in this setting. Moreover, we cannot estimate the optimal sleep time based on the linear structural function. Second, the magnitude of the TSLS estimates is small. This is because the TSLS provides a weighted average of the marginal effects across the entire range of sleep time. By averaging the negative and positive effects, the TSLS delivers an estimate attenuated toward zero, which is not informative for the researcher.\footnote{\cite{Giuntella2019sunset} find a more significant effect of sleep time on obesity. There are two possible reasons: they consider the binary indicator of obesity, and they include more control variables in the regression.}

\begin{figure}
    \centering
    \includegraphics[width=1\linewidth]{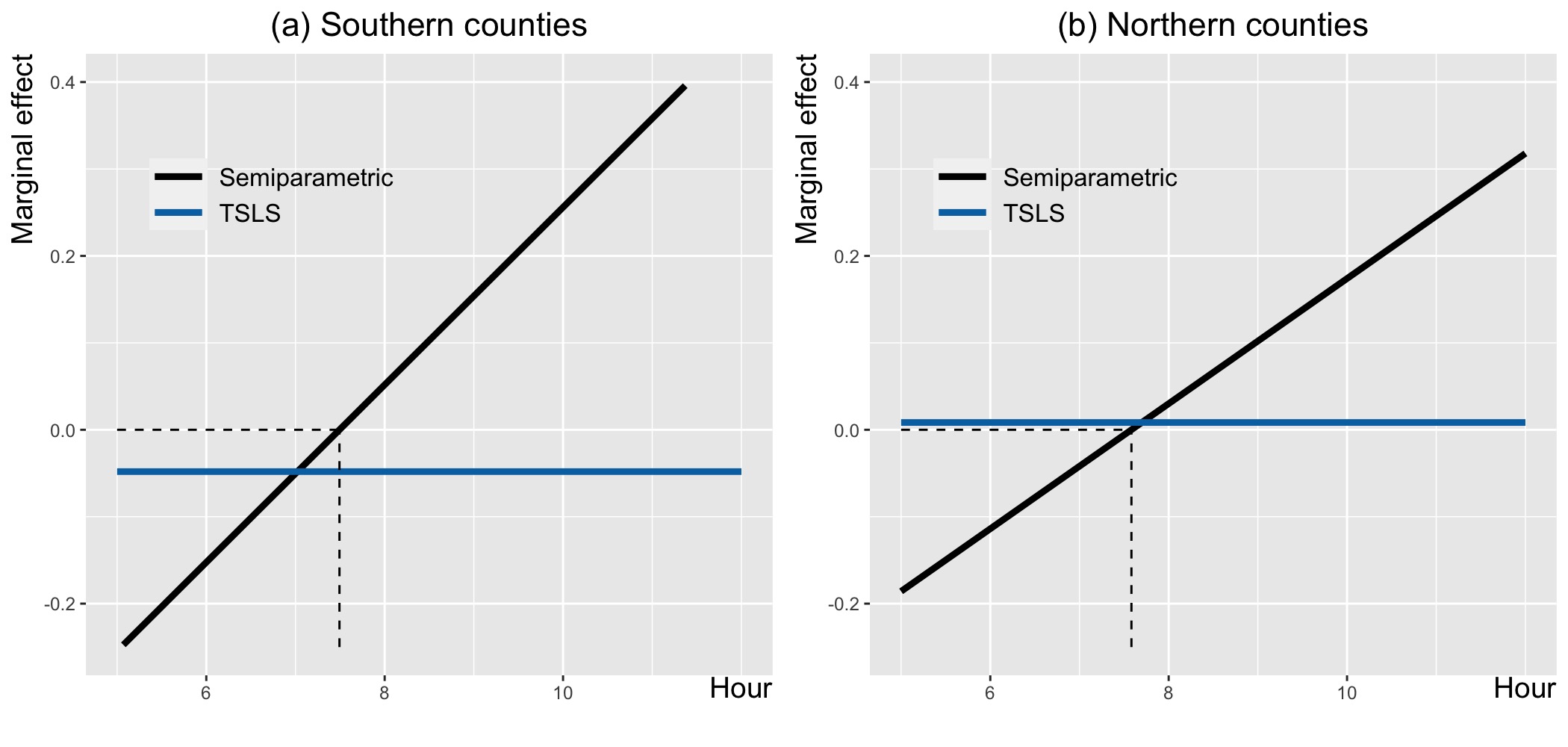}
    \caption{Estimated marginal effects of sleep time on BMI (kg/m$^2$).}
    \label{fig:effects-empirical}
    \caption*{\footnotesize The plots show the estimated marginal effects based on the semiparametric and the TSLS estimator. Graph (a) is computed based on counties with latitude < 37. Graph (a) is computed based on counties with latitude > 37. We can see that the marginal effects are increasing, indicating a nonlinear (U-shaped) structural function. The optimal sleep time computed as the zero of the marginal effect curve is between 7 and 8 hours. However, the marginal effect estimated from the TSLS procedure is constant across different sleep times. These estimates are small in magnitude and less informative for the researcher.}
\end{figure}

    \subsection{Simulations}
    We use simulation studies to investigate the performance of the proposed semiparametric method and compare it with the performance of the TSLS estimator.
    The data generating process (DGP) for these simulations was chosen to roughly approximate the ATUS data used in the empirical application. Let marginal distributions of $U$ and $\varepsilon$ are given by $F_U = \textit{Unif}(0,1)$ and $F_\varepsilon = \textit{Beta}(2,2)$, respectively. Two marginal distributions for $R$ are considered, $\textit{Unif}(0,1)$ and $N(0,1)$. The joint distribution of $(R,\varepsilon,U)$ be characterized by the Gaussian copula with correlation structure \textit{corr}$(R,U) = 0$, \textit{corr}$(\varepsilon,R) = \rho_R$, and \textit{corr}$(\varepsilon,U) = \rho_{U}$. The treatment choice model is given by $h^*_0(r,u) = r + 2\sin(\pi u /2)$ and $h^*_1(r,u) = r + 2u^3$. The structural function is given by 
    \begin{align*}
        g_{\gamma}(T,R,\varepsilon) = \gamma_1 (T-0.5) + \gamma_2 (T^2 - 0.5^2) + \gamma_3 (T - 0.5) \varepsilon + \varepsilon + R. 
    \end{align*}
    The true $\gamma^*$ is taken to be $(1,1,1)$.

    Further implementation details are described below. Construct a kernel function $k$ that is an even function given by 
    \begin{align*}
        k(x) = 
        \begin{cases}
            2x^3 - 3x^2 + 1, & \text{ if } x \in [0,1], \\
            0, & \text{ if } x > 1.
        \end{cases}
    \end{align*}
    We can verify that $k$ is continuously differentiable on the real line and compactly supported on $[0,1]$. Within the interior of its support, $k$ is strictly positive. We use this function $k$ to be the kernels $k_T$, $k_R$, $k_Y$, and $k_{\textit{FS}}$ in the estimation. The bandwidth is chosen to be $b_1 = b_2 = b_3 = 2n^{-1/5}$. The weighting function $w(e,u)$ is chosen to be constant in $u$ and equal to the standard normal density function with respect to $e$.

    Table \ref{tb:sim-semiparametric} contains the simulation results of the performance of the semiparametric estimator for different choices of the marginal distribution of $R$, the correlation parameters $(\rho_U,\rho_R)$ and the sample size. In each case, the number of replications is set at 500. We can see that the estimator performs well with a moderate sample size ($n=1000$). The marginal distribution of $R$ does not have a large impact on the performance. When the sample size is small, larger values of $\rho_R$ or $\rho_U$ can lead to poorer performance of the estimator. The plausible reason is that larger values of the correlation parameters would lead to more severe endogeneity issues in finite samples. When the sample size becomes large, the performance of the estimator does not vary significantly with the choices of $(\rho_R,\rho_U)$.

    \begin{table}[!htbp] 
        \centering
        \begin{tabular}{ccccccccccccc}
        \toprule
        Dist. $R$ & $\rho_U$ & $\rho_R$ & Param & \multicolumn{3}{c}{$n = 500$} & \multicolumn{3}{c}{$n=1000$} & \multicolumn{3}{c}{$n=1500$}\\ \cmidrule(lr){5-7} \cmidrule(lr){8-10} \cmidrule(lr){11-13}
        & & & & bias & sd & mse & bias & sd & mse & bias & sd & mse \\ \midrule 
        \multirow{12}{*}{$U(0,1)$} & \multirow{6}{*}{$.3$} & \multirow{3}{*}{$.3$} & $\gamma_1$ & -.257 & .206 & .108 & -.233 & .134 & .072 & -.219 & .119 & .062\\ 
        & & & $\gamma_2$ & .038 & .153 & .025 & .062 & .098 & .013 & .068 & .081 & .011 \\
        & & & $\gamma_3$ & .172 & .139 & .049 & .158 & .091 & .033 & .154 & .073 & .029 \\ \cmidrule(lr){4-13}
        & & \multirow{3}{*}{$.5$} & $\gamma_1$ & -.260 & .464 & .283 &-.231 & .134 & .071 & -.215 & .116 & .060  \\ 
        & & & $\gamma_2$ & .027 & .514 & .265 & .053 & .097 & .012 & .060 & .080 & .010  \\ 
        & & & $\gamma_3$ & .178 & .168 & .060 & .159 & .097 & .035 & .155 & .076 & .030 \\  \cmidrule(lr){4-13}
        & \multirow{6}{*}{$.5$} & \multirow{3}{*}{$.3$} & $\gamma_1$ & -.256 & .199 & .105 & -.225 & .129 & .067 & -.210 & .111 & .056 \\
        & & & $\gamma_2$ & .039 & .157 & .026 &  .061 & .098 & .013 & .068 & .078 & .011 \\
        & & & $\gamma_3$ & .189 & .140 & .055 & .171 & .086 & .037 & .163 & .070 & .031  \\ \cmidrule(lr){4-13}
        & & \multirow{3}{*}{$.5$} & $\gamma_1$ & -.248 & .456 & .269 & -.217 & .125 & .063 & -.202 & .104 & .052 \\
        & & & $\gamma_2$ & .015 & .512 & .262 & .044 & .095 & .011 & .052 & .075 & .008 \\
        & & & $\gamma_3$ & .222 & .419 & .225 & .178 & .100 & .042 & .171 & .076 & .035 \\  \cmidrule(lr){1-13}
        \multirow{12}{*}{$N(0,1)$} & \multirow{6}{*}{$.3$} & \multirow{3}{*}{$.3$} & $\gamma_1$ & -.224 & .479 & .280 & -.228 & .156 & .076 & -.216 & .136 & .065 \\
        & & & $\gamma_2$ & .006 & .527 & .278 & .059 & .114 & .016 & .067 & .093 & .013 \\
        & & & $\gamma_3$ & .152 & .170 & .052 & .152 & .105 & .034 & .148 & .085 & .029 \\ \cmidrule(lr){4-13}
        & & \multirow{3}{*}{$.5$} & $\gamma_1$ & -.254 & .362 & .196 & -.224 & .153 & .074 & -.212 & .124 & .060 \\
        & & & $\gamma_2$ & .030 &.268 & .073 & .048 & .112 & .015 & .058 & .091 & .012 \\
        & & & $\gamma_3$ & .167 & .178 & .060 & .152 & .113 & .036 & .149 & .086 & .030 \\ \cmidrule(lr){4-13}
        & \multirow{6}{*}{$.5$} & \multirow{3}{*}{$.3$} & $\gamma_1$ & -.243 & .233 & .113 & -.221 & .149 & .071 & -.209 & .127 & .060  \\
        & & & $\gamma_2$ & .027 & .185 & .035 & .059 & .113 & .016 & .069 & .090 & .013 \\
        & & & $\gamma_3$ & .176 & .172 & .061 & .163 & .102 & .037 & .157 & .082 & .031 \\ \cmidrule(lr){4-13}
        & & \multirow{3}{*}{$.5$} & $\gamma_1$ & -.230 & .500 & .303 & -.210 & .144 & .065 & -.198 & .118 & .053 \\
        & & & $\gamma_2$ & -.012 & .619 & .383 & .037 & .113 & .014 & .049 & .086 & .010 \\
        & & & $\gamma_3$ & .193 & .248 & .099 & .172 & .116 & .043 & .167 & .088 & .036 \\ \bottomrule
    \end{tabular}
    \caption{Performance of the semiparametric estimator.}
    \caption*{\footnotesize The structural function follows the three-parameter specification in Example \ref{eg:1}. The number of replications is $500$. The marginal distribution of $R$ is chosen to be the uniform distribution on $[0,1]$ or the standard normal distribution. The two correlation parameters $\rho_R = \text{corr}(\varepsilon,R)$ and $\rho_U = \text{corr}(\varepsilon,U)$ are chosen from $\{0.3,0.5\}$. The results demonstrate the following points. First, the semiparametric estimator performs well with a moderate sample size of 1000. Second, the performance of the estimator is not affected by The marginal distribution of $R$. Third, when the sample size is as large as 1000, the performance of the estimator does not vary significantly with the choices of $(\rho_R,\rho_U)$.}
    \label{tb:sim-semiparametric}
    \end{table}

    It is also of interest to compare the semiparametric estimator with the TSLS estimator. Directly comparing the the two estimators can be difficult since they are of different dimensions and converge to different limits. Instead, we can compare their performance on estimating the marginal effect. For the structural function $g_\gamma(t,\bar{r},e) = \gamma_1 t + \gamma_2 t^2 + \gamma_3 te + e$, the marginal effect of the treatment on the outcome is $\frac{\partial}{\partial t} g_\gamma(t,\bar{r},e) = \gamma_1 + 2\gamma_2 t + \gamma_3 e$, which takes on different values for different treatment and outcome levels. For a given treatment level $t$, we can use the semiparametric estimator to obtain an estimate $\hat{\gamma}_1 + 2\hat{\gamma}_2 t + \hat{\gamma}_3 e$ of the marginal effect. However, a TSLS procedure would deliver a scalar estimate that is a mixture of marginal effects across different treatment and outcome levels. Figure \ref{fig:nonlinear-sf} shows that with a nonlinear specification the semiparametric estimator outperforms the TSLS estimator. Figure \ref{fig:nonseparable-sf} shows similar findings with a fully nonlinear and nonseparable specification.
    
    Next, we compare the semiparametric estimator with the TSLS estimator when the structural function $g$ is linear. This is achieved by imposing $\gamma_2 = \gamma_3 = 0$. In this case, the TSLS estimator is consistent for the coefficient $\gamma_1$. However, the identification of the TSLS estimator is based solely on the difference between the two means. If the two distributions corresponding to $h_0$ and $h_1$ have the same mean, then the TSLS procedure suffers from weak identification issues. In contrast, the identification of the semiparametric estimator $\hat{\gamma}$ is based on the entire difference between $h_0$ and $h_1$. The semiparametric estimator continues to work even if the estimand of the TSLS estimator is weakly identified. For the simulation, we let $h_0$ be the quantile function of \textit{Beta}$(0.1,0.1)$, and $h_1$ be the quantile function of \textit{Beta}$(10,10)$. These two distributions are significantly different, but they have the same mean (0.5). Figure \ref*{fig:linear-sf} shows that in this case the semiparametric estimator outperforms the TSLS estimator even if the structural function is linear.
    

    \begin{figure}[!htbp] 
        \centering
        \includegraphics[width = 1\linewidth]{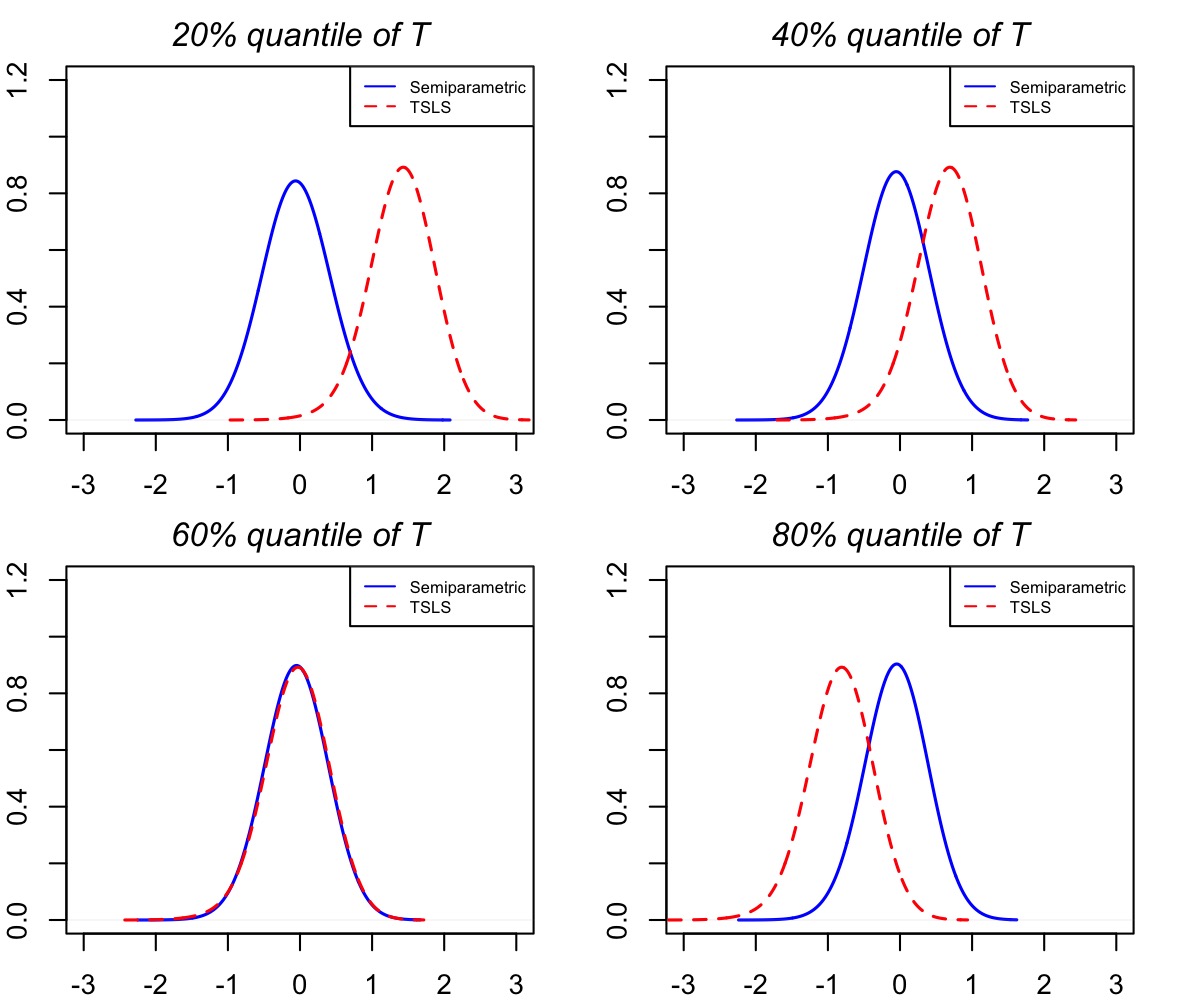}
        \caption{Marginal effects comparison with a nonlinear structural function.}
        \label{fig:nonlinear-sf}
        \caption*{\footnotesize Kernel density estimates of estimated marginal effect by the semiparametric and TSLS estimators (minus the true marginal effect) based on 500 replications. The sample size is 1000. The true structural function is specified to be $g(t,\bar{r},e)=t/2+t^2+e$, where the marginal effect is $1/2 + 2t$. The graphs show the estimation results of four quantile levels of the treatment: 20\%, 40\%, 60\%, and 80\%. The distribution of the semiparametric estimator is correctly centered while the TSLS estimator incurs a large bias. The TSLS estimator gives an approximately unbiased estimate of the marginal effect only around the 60\% quantile level.}
    \end{figure}

    \begin{figure}[!htbp] 
        \centering
        \includegraphics[width = 1\linewidth]{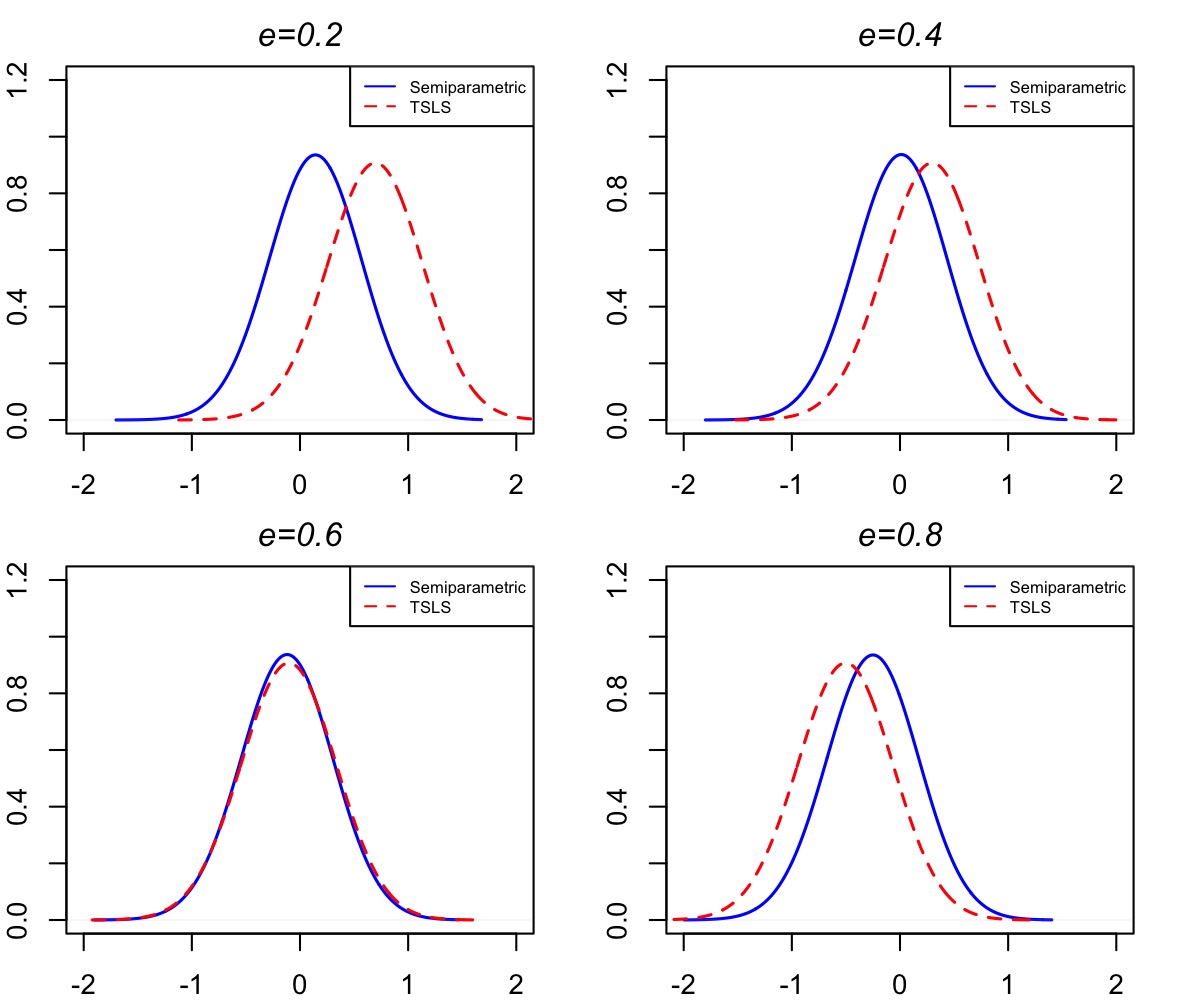}
        \caption{Marginal effects comparison with a nonseparable structural function.}
        \label{fig:nonseparable-sf}
        \caption*{\footnotesize Kernel density estimates of estimated marginal effect by the semiparametric and TSLS estimators (minus the true marginal effect) based on 500 replications. The sample size is 1000. The true structural function is specified to be $g(t,\bar{r},e)=t/2+t^2+2te+e$, where the marginal effect is $1/2 + 2t+2e$. The treatment level is specified to be the median. The graphs show the estimation results of four levels of $e$: 0.2, 0.4, 0.6, 0.8. The distribution of the semiparametric estimator is correctly centered while the TSLS estimator incurs a large bias. The TSLS estimator gives an unbiased estimate of the marginal effect only around $e=0.6$.}
    \end{figure}

    \begin{figure}[!htbp] 
        \centering
        \includegraphics[width = 1\linewidth]{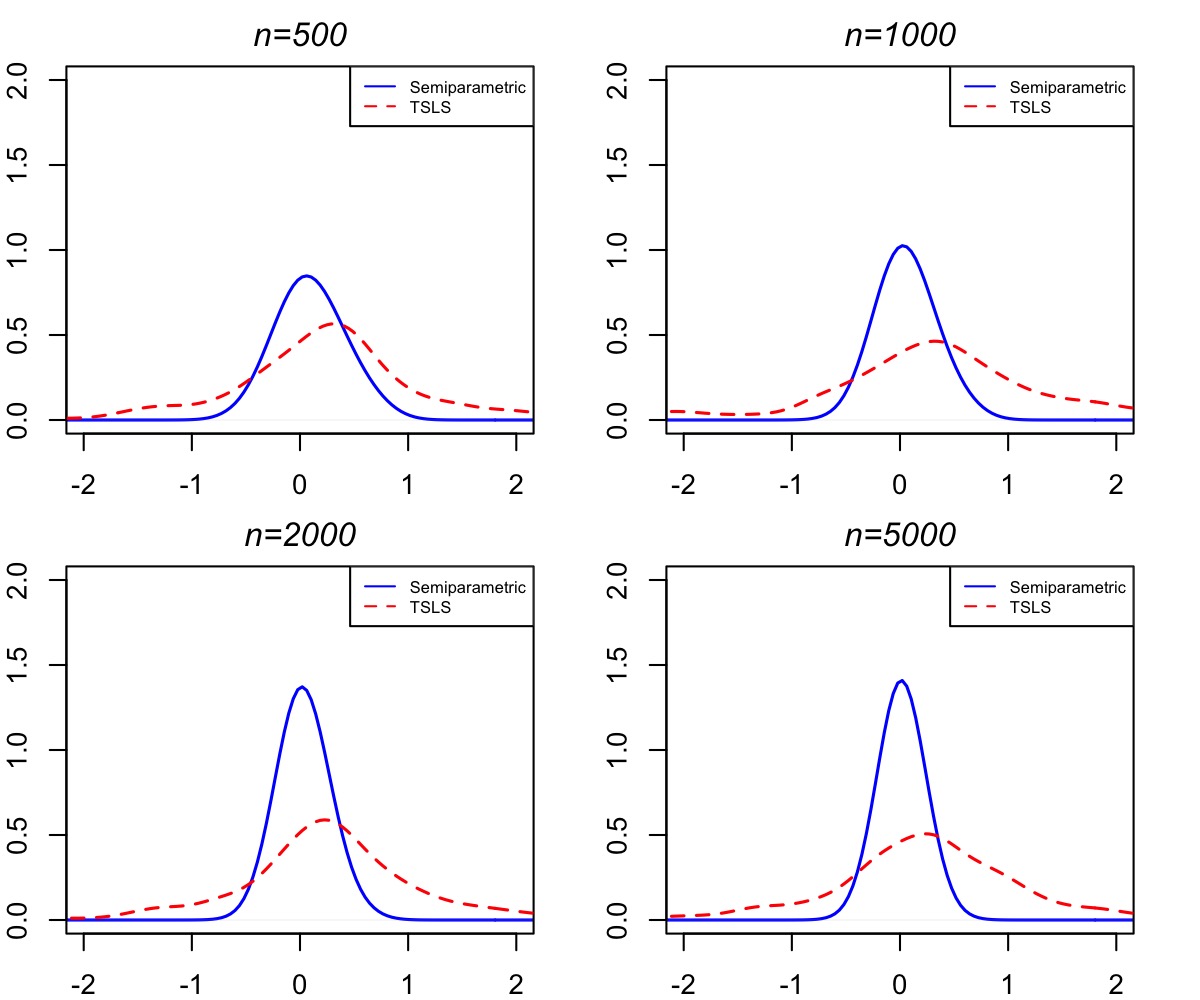}
        \caption{Semiparametric and TSLS estimators when the structural function is linear.}
        \label{fig:linear-sf}
        \caption*{\footnotesize Kernel density estimates of the semiparametric and TSLS estimators (minus the true $\gamma^*=1$) based on 500 replications. In the DGP, $h_0$ is equal to the quantile function of \textit{Beta}$(0.1,0.1)$, and $h_1$ is equal to the quantile function of \textit{Beta}$(10,10)$. These two distributions are significantly different, but they have the same mean (0.5). In this case, the semiparametric estimator outperforms the TSLS estimator because the latter is weakly identified. }
    \end{figure}
    
    \section{Conclusion}
    
    In this study, we have examined the identification and estimation of the structural function in an RD design with a continuous treatment variable. We have established the nonparametric identification result and proposed a semiparametric estimator for the possibly nonlinear and nonseparable structural function. The estimator is proven to be consistent and asymptotically normal. The empirical application and simulation studies demonstrate the advantage of the semiparametric estimator compared to the TSLS estimator. 
    
    There are two promising ways to extend the results in this paper in the future. First, we can consider extrapolating the identification result away from the cutoff. This can be done by identifying the derivative of the structural function with respect to the running variable at the cutoff as in \cite{dong2015identifying}. Second, we can apply the methodology developed in this paper to the regression kink design model studied by \cite{card2015RKD,dong2018jump} where the treatment choice function exhibits a kink instead of a discontinuity at the cutoff.

\appendix

\numberwithin{equation}{section}
\numberwithin{lemma}{section}
\numberwithin{definition}{section}

\section{Proof of identification results} \label{sec:proof-id}

In this section we prove the identification results including Lemma \ref{lm:quantile-representation}, Lemma \ref{lm:local-control}, Theorem \ref{thm:id-rd-cutoff}, Corollary \ref{cor:CASF-identification}, and Corollary \ref{cor:gamma-star-identification}.

\begin{proof} [Proof of Lemma \ref{lm:quantile-representation}]
    We first prove the second equality in Equation (\ref{eqn:def-U}). The conditional distribution function of $T$ given $R$ is 
    \begin{align*}
        F_{T|R}(t|r) & = \mathbb{P}(T \leq t | R=r) \mathbf{1}\{r < \bar{r}\} + \mathbb{P}(T \leq t | R=r) \mathbf{1}\{r \geq \bar{r}\} \\
        & = \mathbb{P}(U_0 \leq m_0^{-1}(r,t) | R=r) \mathbf{1}\{r < \bar{r}\} + \mathbb{P}(U_1 \leq m_1^{-1}(r,t) | R=r) \mathbf{1}\{r \geq \bar{r}\} \\
        & = F_{U_0|R}(m_0^{-1}(r,t)|r) \mathbf{1}\{r < \bar{r}\} + F_{U_1|R}(m_1^{-1}(r,t)|r) \mathbf{1}\{r \geq \bar{r}\},
    \end{align*}
    where the last line follows from the monotonicity of $m_0$ and $m_1$. Therefore, we have 
    \begin{align*}
        F_{T|R}(T|R) & = F_{U_0|R}(m_0^{-1}(R,T)|R) \mathbf{1}\{R < \bar{r}\} + F_{U_1|R}(m_1^{-1}(R,T)|R) \mathbf{1}\{R \geq \bar{r}\} \\
        & = F_{U_0|R}(U_0|R) \mathbf{1}\{R < \bar{r}\} + F_{U_1|R}(U_1|R) \mathbf{1}\{R \geq \bar{r}\} = U.
    \end{align*}
    For (i) of Lemma \ref{lm:quantile-representation}, take any $u \in [0,1]$ and $r < \bar{r}$, we have
    \begin{align*}
        \mathbb{P}(U \leq u | R=r) = \mathbb{P}(F_{U_0 | R}(U_0 | r) \leq u | R=r) = u.
    \end{align*}
    Similarly, we can show that $\mathbb{P}(U \leq u | R=r) = u$ when $r \geq \bar{r}$. Therefore, $U | R$ follows the uniform distribution. Then
    the second argument follows from the monotonicity of $h$ with respect to the second argument.
    The statements (ii) and (iii) are straightforward from the definition of $h_0$ and $h_1$ and the monotonicity and continuity assumptions. For (iv), notice that for $r < \bar{r}$,
    \begin{align*}
        F_{\varepsilon | U,R}(e | u,r) & = \mathbb{P}(\varepsilon \leq e | U=u, R=r) \\
        & = \mathbb{P}(\varepsilon \leq e | F_{U_0 | R}(U_0 | r)=u, R=r) \\
        & = \mathbb{P}(\varepsilon \leq e | U_0 = F_{U_0 | R}^{-1}(u | r), R=r) \\
        & = F_{\varepsilon | U_0,R}\big(e | F_{U_0 | R}^{-1}(u | r),r\big),
    \end{align*}
    where the third equality follows from the strict monotonicity of $F_{U_0 | R}^{-1}(u|r)$ in $u$ imposed in Assumption \ref{ass:smoothness}(ii). Similarly, we can show that for $r \geq \bar{r}$,
    \begin{align*}
        F_{\varepsilon | U,R}(e | u,r) & = \mathbb{P}(\varepsilon \leq e | U=u, R=r) = F_{\varepsilon | U_1,R}\big(e | F_{U_1 | R}^{-1}(u | r),r\big).
    \end{align*}
    Combining the two equations together, we have
    \begin{align*}
        F_{\varepsilon | U,R}(e | u,r) = 
        \begin{cases}
            F_{\varepsilon | U_0,R}\big(e | F_{U_0 | R}^{-1}(u | r),r\big), r < \bar{r}, \\
            F_{\varepsilon | U_1,R}\big(e | F_{U_1 | R}^{-1}(u | r),r\big), r \geq \bar{r}.
        \end{cases}
    \end{align*}
    By Assumption \ref{ass:smoothness}(iii), we know that $F_{\varepsilon | U,R}(e | u,r)$ is strictly increasing in the first argument $e$. By Bayes rule, the rank similarity condition in Assumption \ref{ass:rank-similarity} implies that $U_0 | R=\bar{r}^-$ has the same distribution as $U_1 | R=\bar{r}^+$, and $\varepsilon | U_0,R=\bar{r}^- $ has the same distribution as $ \varepsilon | U_1,R=\bar{r}^+$. Then we have
    \begin{align*}
        \lim_{r \uparrow \bar{r}} F_{\varepsilon | U,R}(e | u,r) =  F_{\varepsilon | U_0,R}\big(e | F_{U_0 | R}^{-1}(u | \bar{r}),\bar{r}\big) = F_{\varepsilon | U_1,R}\big(e | F_{U_1 | R}^{-1}(u | \bar{r}),\bar{r}\big) = \lim_{r \downarrow \bar{r}} F_{\varepsilon | U,R}(e | u,r).
    \end{align*}
\end{proof}

\begin{lemma} \label{lm:continuous-inverse-bivariate}
    Let $f: (x,y) \mapsto z$ be a real-valued bivariate function defined on a compact set in $\mathbb{R} \times \mathbb{R}^d$. Assume that $f$ is continuous on its entire domain and strictly increasing in the first argument $x$. Let $f^{-1}$ denote the inverse of $f$ with respect to the first argument. Then $f^{-1}$ is continuous on its domain and strictly increasing in the first argument.

    Therefore, under Assumptions \ref{ass:dual-monotonicity} and \ref{ass:smoothness}, the inverse of $g,h_0$ and $h_1$ (with respect to the last argument), which are respectively $g^{-1}, (h_0)^-$ and $(h_1)^-$, are all continuous and strictly increasing with respect to the last argument.
\end{lemma}

\begin{proof} [Proof of Lemma \ref{lm:continuous-inverse-bivariate}]
    Fix any $(z_0,y_0)$, we want to show that $f^{-1}(z_0,\cdot)$ is continuous at $(z_0,y_0)$. If not, then there exists $\delta>0$ and a sequence $\{(z_k,y_k)\}$ such that $\lVert (z_k,y_k) - (z_0,y_0) \rVert < 1/k$ but 
    \begin{align*}
        |f^{-1}(z_0,y_k) - f^{-1}(z_0,y_0)| > \delta.
    \end{align*}
    Denote $x_k = f^{-1}(z_k,y_k)$ and $x_0 = f^{-1}(z_0,y_0)$. Because the sequence $\{x_k\}$ lies in a compact set, it has a convergent subsequence. Without loss of generality, we assume $\{x_k\}$ itself is converging. Then $\lim x_k \ne x_0$. However, by the continuity of $f$,
    \begin{align*}
        f(\lim x_k,y_0) = f(\lim x_k,\lim y_k) = \lim f(x_k,y_k) = z_0 = f(x_0,y_0).
    \end{align*}
    This leads to a contradiction since $f(\cdot,y_0)$ is strictly increasing.

    To show that $f^{-1}$ is strictly increasing with respect to the first argument, take any $y$ and $z_1 > z_0$. If $f^{-1}(z_1,y) \leq f^{-1}(z_0,y)$, then $z_1 = f(f^{-1}(z_1,y),y) \leq f(f^{-1}(z_0,y),y) = z_0$, which leads to a contradiction.
\end{proof}

\begin{proof} [Proof of Lemma \ref{lm:local-control}]

    By the definition of $F^-_{Y|T,R}$ and the monotonicity and continuity of $g^*$ and $h_0$, 
    \begin{align*}
        \text{LHS of (\ref{eqn:refutable-implication}) with } g = g^* & = \lim_{r \uparrow \bar{r}} F_{Y | T,R} (g^*(h_0(r,u),r,e) | h_0(r,u),r) \\
        & = \lim_{r \uparrow \bar{r}} \mathbb{P}(Y \leq g^*(h_0(r,u),r,e) | T=h_0(r,u),R=r) \\
        & = \lim_{r \uparrow \bar{r}} \mathbb{P}(g^*(h_0(r,u),r,\varepsilon) \leq g^*(h_0(r,u),r,e) | T=h_0(r,u),R=r) \\
        & = \lim_{r \uparrow \bar{r}} \mathbb{P}(\varepsilon \leq e | U=u,R=r) \\
        & = F_{\varepsilon | U,R}(e | u,r)
    \end{align*}
    where the last line follows from the continuity of $F_{\varepsilon | U,R}(e | u,r)$ with respect to the last argument $r$ (Lemma \ref{lm:quantile-representation}). Similarly, we can show that the RHS of (\ref{eqn:refutable-implication}) is equal to $F_{\varepsilon | U,R}(e | u,\bar{r}).$
    Then the result follows.
\end{proof}

\begin{proof} [Proof of Theorem \ref{thm:id-rd-cutoff}]
    Denote 
    \begin{align*}
        \mathcal{T}^\times = \{ h_0(\bar{r},u): h_0(\bar{r},u) = h_1(\bar{r},u) \in [t_0',t_0''] \cap [t_1',t_1''], u \in [0,1] \}.
    \end{align*}
    By Assumption \ref{ass:strong-discontinuity}(ii), $\mathcal{T}^\times$ is nonempty and finite. Then $[t_0',t_0''] \cap [t_1',t_1'']$ is a closed interval with nonempty interior.\footnote{Notice that if $[t_0',t_0''] \cap [t_1',t_1'']$ is a singleton, then $\mathcal{T}^\times$ is empty.} Let 
    $$\inf ([t_0',t_0''] \cap [t_1',t_1'']) = t_1 \leq t_2 \leq \cdots \leq t_L = \sup ([t_0',t_0''] \cap [t_1',t_1''])$$
    denote the unique elements of $\mathcal{T}^\times \cup \{ \inf ([t_0',t_0''] \cap [t_1',t_1'']), \sup ([t_0',t_0''] \cap [t_1',t_1'']) \}.$

    Here is the strategy of the proof. For each $g \in \mathcal{G}$ that satisfies Equation (\ref{eqn:refutable-implication}), define
    \begin{align} \label{eqn:Ig-defn}
        \tilde{\lambda}^g(t,e) = g^{-1}(t,\bar{r},g^*(t,\bar{r},e)).
    \end{align} 
    The goal is to show that $\tilde{\lambda}^g$ is constant as a function of $t$ for every $e \in \mathcal{E}$.
    We proceed in five steps. Step 1 derives some useful properties of the function $\tilde{\lambda}^g$, including an important identity, Equation (\ref{eqn:lambda-g-invariant}). Step 2 shows that $\tilde{\lambda}^g$ is constant in $t$ on the interval $(t_1,t_2)$. Step 3 shows that $\tilde{\lambda}^g$ is constant in $t$ on the entire region $[t_1,t_L] = [t_0',t_0''] \cap [t_1',t_1'']$. Step 4 further expands this constancy to $[t_0',t_0''] \cup [t_1',t_1'']$. Step 5 concludes.

    \bigskip
    \noindent \textbf{Step 1.}
    Since $g^*$ and $g^{-1}$ are continuous and are strictly increasing in the last argument, $\tilde{\lambda}^g$ is also continuous and strictly increasing in the last argument. Also, $F^-_{Y|T,R}(\cdot | t,\bar{r})$ is strictly increasing since 
    \begin{align*}
        F^-_{Y | T,R}(y | t,\bar{r}) = F_{\varepsilon | U,R}((g^*)^{-1}(t,\bar{r},y) | (h_0)^{-1}(\bar{r},t),\bar{r}) 
    \end{align*}
    is strictly increasing in $y$ (Assumption \ref{ass:smoothness}(ii)).
    
    Notice that $g(t,\bar{r},\tilde{\lambda}^g(t,e)) = g^*(t,\bar{r},e).$
    Then for any $e \in \mathcal{E}$ and $u \in [0,1]$,
    \begin{align*}
        & F^-_{Y | T,R} (g(h_0(\bar{r},u),\bar{r},\tilde{\lambda}^g(h_0(\bar{r},u),e)) | h_0(\bar{r},u),\bar{r}) \\
        = &  F^-_{Y | T,R} (g^*(h_0(\bar{r},u),\bar{r},e) | h_0(\bar{r},u),\bar{r}) \\
        = &  F^+_{Y | T,R} (g^*(h_1(\bar{r},u),\bar{r},e) | h_1(\bar{r},u),\bar{r}) \\
        = &  F^+_{Y | T,R} (g(h_1(\bar{r},u),\bar{r},\tilde{\lambda}^g(h_1(\bar{r},u),e)) | h_1(\bar{r},u),\bar{r}),
    \end{align*}
    where the second inequality follows from Lemma \ref{lm:local-control}. The above equality implies 
    \begin{align} \label{eqn:lambda-g-invariant}
        \tilde{\lambda}^g(h_0(\bar{r},u),e) = \tilde{\lambda}^g(h_1(\bar{r},u),e).
    \end{align}
    To see that, suppose there exists $e $ and $u$ such that $\tilde{\lambda}^g(h_0(\bar{r},u),e) \ne \tilde{\lambda}^g(h_1(\bar{r},u),e).$ Since $g$ satisfies Condition (\ref{eqn:refutable-implication}), 
    \begin{align*}
        &F^-_{Y | T,R} (g(h_0(\bar{r},u),\bar{r},\tilde{\lambda}^g(h_1(\bar{r},u),e)) | h_0(\bar{r},u),\bar{r}) \\
        = &  F^+_{Y | T,R} (g(h_1(\bar{r},u),\bar{r},\tilde{\lambda}^g(h_1(\bar{r},u),e)) | h_1(\bar{r},u),\bar{r}) \\
        = &  F^-_{Y | T,R} (g(h_0(\bar{r},u),\bar{r},\tilde{\lambda}^g(h_0(\bar{r},u),e)) | h_0(\bar{r},u),\bar{r}),
    \end{align*}
    which violates the fact that $F^-_{Y | T,R} (g(h_0(\bar{r},u),\bar{r},\cdot) | h_0(\bar{r},u),\bar{r})$ is strictly increasing.

    \bigskip
    \noindent \textbf{Step 2.}
    Consider the interval $(t_1,t_2)$. By construction, $\{t_1,t_2\} \cap \mathcal{T}^\times \ne \emptyset$. Notice that over $(t_1,t_2)$, $(h_0)^{-1}(\bar{r},\cdot)$ and $(h_1)^{-1}(\bar{r},\cdot)$ do not intersect. Then by continuity, one of them is always strictly greater than the other. The goal is to show that $\tilde{\lambda}^g(t,e)$ is constant as a function of $t$ over $(t_1,t_2)$. There are four cases to consider, depending on whether $t_1 \in \mathcal{T}^\times$ or $t_2 \in \mathcal{T}^\times$ and whether $(h_0)^{-1}(\bar{r},\cdot)$ is strictly greater or smaller than $(h_1)^{-1}(\bar{r},\cdot)$ over $(t_1,t_2)$.
    
    We first focus on the case of $t_1 \in \mathcal{T}^\times$ and $(h_0)^{-1}(\bar{r},\cdot) < (h_1)^{-1}(\bar{r},\cdot)$ over $(t_1,t_2)$. The other cases are essentially the same. Define a mapping $\pi(t) = h_1(\bar{r},(h_0)^{-1}(\bar{r},t))$. 
    Such a mapping $\pi$ maps the interval $(t_1,t_2)$ back to itself. To see that, we first notice that $\pi(t)$ is less than $t$ for any $t \in (t_1,t_2)$ since 
    \begin{align*}
        \pi(t) = h_1(\bar{r},(h_0)^{-1}(\bar{r},t)) \leq h_1(\bar{r},(h_1)^{-1}(\bar{r},t)) = t < t_2.
    \end{align*}
    Suppose $\pi(t) \leq t_1$, then
    \begin{align*}
        (h_0)^{-1}(\bar{r},t_1) & = (h_1)^{-1}(\bar{r},t_1) \\
        & \leq (h_1)^{-1}(\bar{r},\pi(t)) \\
        & = (h_0)^{-1}(\bar{r},t),
    \end{align*}
    where the first line follows from $t_1 \in \mathcal{T}^\times$, the second line follows from the monotonicity of $(h_1)^{-1}$ and $\pi(t) \leq t_1$, and the last line follows from the definition of $\pi$. This contradicts the strict monotonicity of $(h_0)^{-1}(\bar{r},\cdot)$ since $t_1 < t$.

    Now pick any $\tilde{t}_0 \in (t_1,t_2)$, the recursive sequence $\tilde{t}_{k+1} = \pi(\tilde{t}_k)$ is well-defined. This sequence is non-increasing and bounded below by $t_1$. Therefore, $\lim_{k \rightarrow \infty} \tilde{t}_k$ exists and lies in the interval $[t_1,t_2)$. By the continuity of $(h_0)^{-1}$ and $(h_1)^{-1}$,
    \begin{align*}
        (h_0)^{-1}(\bar{r},\lim \tilde{t}_k) & = \lim (h_0)^{-1}(\bar{r}, \tilde{t}_k) \\
        & = \lim (h_1)^{-1}(\bar{r}, \pi(\tilde{t}_k)) \\
        & = \lim (h_1)^{-1}(\bar{r}, \tilde{t}_{k+1}) \\
        & = (h_1)^{-1}(\bar{r}, \lim \tilde{t}_{k+1}),
    \end{align*}
    where the second line follows from the definition of $\pi$ and the third line follows from the construction of the sequence $\{\tilde{t}_k\}$. Then it must be true that $\lim \tilde{t}_k = t_1$ since we are studying the case where $(h_0)^{-1}(\bar{r},\cdot) < (h_1)^{-1}(\bar{r},\cdot)$ over $(t_1,t_2)$.

    Equation (\ref{eqn:lambda-g-invariant}) implies that $\tilde{\lambda}^g(\cdot,e)$ is invariant with respect to the transformation $\pi$: 
    \begin{align*}
        \tilde{\lambda}^g(\pi(t),\bar{r},e) = \tilde{\lambda}^g(t,e), 
    \end{align*}
    for every $t \in [t_0',t_0'']$ and $e \in \mathcal{E}$. 
    Then $\tilde{\lambda}^g$ is invariant along the sequence $\{ \tilde{t}_k \}$. By the continuity of $\tilde{\lambda}^g$,  
    \begin{align*}
        \tilde{\lambda}^g(\tilde{t}_0,e) & = \lim_{k \rightarrow \infty} \tilde{\lambda}^g(\tilde{t}_0,e) = \lim_{k \rightarrow \infty} \tilde{\lambda}^g(\tilde{t}_k,e) = \tilde{\lambda}^g(\lim_{k \rightarrow \infty} \tilde{t}_k,e) = \tilde{\lambda}^g(t_1,e) = \lambda^g(e),
    \end{align*}
    where the first equality holds since $\tilde{\lambda}^g(\tilde{t}_0,\bar{r},e)$ is constant with respect to $k$, the second equality holds since $\tilde{\lambda}^g$ is invariant along the sequence $\{ \tilde{t}_k \}$, the third equality follows from the continuity of $\tilde{\lambda}^g$, and the last equality is the definition of the function $\lambda$ on $\mathcal{E}$.

    Since the initial point $\tilde{t}_0$ is chosen arbitrarily from the interval $(t_1,t_2)$, the above analysis shows that $\tilde{\lambda}^g(t,e) = \lambda^g(e)$ for $t \in (t_1,t_2)$ (hence for $t \in [t_1,t_2]$, by continuity) and $e \in \mathcal{E}$. 
    
    Recall that this analysis is conducted for the case where $t_1 \in \mathcal{T}^\times$ and $(h_0)^{-1}(\bar{r},\cdot)$ is strictly smaller than $(h_1)^{-1}(\bar{r},\cdot)$ over $(t_1,t_2)$. The other three cases reach the same conclusion that $\tilde{\lambda}^g(t,e)$ is equal to $\lambda^g(e)$ over $[t_1,t_2] \times \mathcal{E}$ through symmetric arguments. More specifically, we can switch $h_1$ and $h_0$ in defining $\pi$ so that the sequence $\{ \tilde{t}_k \}$ tends to a point in $\mathcal{T}^\times$.

    \bigskip
    \noindent \textbf{Step 3.} 
    Repeat step 2 on each interval $(t_l,t_{l+1}),l=2,\cdots,L$. It follows that $\tilde{\lambda}^g(t,e) = \lambda^g(e)$ over $([t_0',t_0''] \cap [t_1',t_1'']) \times [0,1]$.

    \bigskip
    \noindent \textbf{Step 4.} 
    Pick any $t' \in \text{int}([t_0',t_0''] \cup [t_1',t_1'']) \setminus [t_1,t_L]$ (if this set is nonempty). There are four cases to consider, depending on whether $t' \in [t_0',t_0'']$ or $t' \in [t_1',t_1'']$ and whether $t_0'' < t_1''$ or $t_0'' > t_1''$. Without loss of generality, assume that $t' \in [t_0',t_0'']$ and $t_0'' < t_1''$. The other three cases can be dealt with symmetric arguments. In this case, $t_1' < t_0''$ because $[t_0',t_0''] \cap [t_1',t_1'']$ is a non-degenerate interval. Denote $t'' \in [t_1',t_1'']$ such that $(h_0)^{-1}(\bar{r},t) = (h_1)^{-1}(\bar{r},t'')$. It must be the case that $t'' \in [t_1',t_0'']$. If that is not the case, then $t'' > t_0''$. Then for any $t \in [t_1',t_0'']$,
    \begin{align*}
        (h_0)^{-1}(\bar{r},t) > (h_0)^{-1}(\bar{r},t') = (h_1)^{-1}(\bar{r},t'') > (h_1)^{-1}(\bar{r},t), t \in [t_0',t_0''] \cap [t_1',t_1''],
    \end{align*}
    by the strict monotonicity of $(h_0)^{-1}(\bar{r},\cdot)$ and $(h_1)^{-1}(\bar{r},\cdot)$. However, this contradicts the assumption that $\mathcal{T}^\times$ is nonempty. Then by (\ref{eqn:lambda-g-invariant}), 
    \begin{align*}
        \tilde{\lambda}^g(t',e) = \tilde{\lambda}^g(t'',e) = \lambda^g(e).
    \end{align*}

    \bigskip
    \noindent \textbf{Step 5.}
    By the definition of $\tilde{\lambda}^g$ in (\ref{eqn:Ig-defn}), we now have
    \begin{align*}
        g^*(t,\bar{r},e) = g(t,\bar{r},\lambda^g(e)), \text{ for } t \in [t_0',t_0''] \cup [t_1',t_1''] , e \in \mathcal{E}.
    \end{align*} 
    By the properties of $\tilde{\lambda}^g$, we know $\lambda^g$ is continuous and strictly increasing. The above statement holds for any $g \in \mathcal{G}$ that satisfies Equation (\ref{eqn:refutable-implication}). 

\end{proof}

\begin{proof} [Proof of Corollary \ref{cor:CASF-identification}]

    Based on the definition of $F_{\varepsilon | U,R}^{g^{\circ}}$ in (\ref{eqn:def-F-varepsilon-g}), Theorem \ref{thm:id-rd-cutoff} and Lemma \ref{lm:local-control}, we have
    \begin{align*}
        F^{g^{\circ}}_{\varepsilon | U,R}(\lambda^g(e) | u,\bar{r}) & = F^-_{Y | T,R} ({g^{\circ}}(h_0(\bar{r},u),\bar{r},\lambda^{g^{\circ}}(e)) | h_0(\bar{r},u),\bar{r}) \\
        & = F^-_{Y | T,R} ({g^{\circ}}(h_0(\bar{r},u),\bar{r},e) | h_0(\bar{r},u),\bar{r}) \\
        & = F_{\varepsilon | U,R}(e | u,\bar{r}).
    \end{align*}
    The second claim follows from a change of variable.

\end{proof}

\begin{proof} [Proof of Corollary \ref{cor:gamma-star-identification}]
    By construction, $\big\lVert D_{\gamma,h^*} \big\rVert_{w} \geq 0$ for any $\gamma \in \Gamma$. By Lemma \ref{lm:local-control}, we have $\big\lVert D_{\gamma^*,h^*} \big\rVert_{w} = 0$. We want to show that $\gamma^*$ is the unique zero. Since $w>0$, we have 
    \begin{align*}
        \norm{ D_{\gamma,h^*} }_w = 0 & \implies D_{\gamma,h^*}(e,u) = 0, \text{ for all } e \in \mathcal{E}, u \in [0,1] \\
        & \implies \text{ Condition (\ref{eqn:refutable-implication}) is satisfied by } g_\gamma(\cdot,\bar{r},\cdot),
    \end{align*}
    where the second line follows by taking the partial derivative with respect to $u$ on both sides. 
    By Theorem \ref{thm:id-rd-cutoff}, this implies that $g_\gamma(\cdot,\bar{r},\cdot) = g_{\gamma'}(\cdot,\bar{r},\lambda(\cdot))$. By Assumption \ref{ass:normalization-Gamma}, it must be that $\gamma = \gamma^*$. Therefore, $\gamma^*$ is the unique minimizer of $\norm{ D_{\gamma,h^*} }_w$.
\end{proof}

\section{Proof of estimation results} \label{sec:proof-est}

This section proceeds as follows. Section \ref{ssec:proof-thm2} provides the proofs of Theorem \ref{thm:estimation} and Proposition \ref{prop:fs-quantile}. Section \ref{ssec:ept} introduces the empirical process theory and presents the lemmas on the uniform convergence results used in Section \ref{ssec:proof-thm2}. Section \ref{ssec:cov-est} discusses the consistent estimation of the asymptotic covariance matrix.

\subsection{Proofs of Theorem \ref{thm:estimation} and Proposition \ref{prop:fs-quantile}} \label{ssec:proof-thm2}
\begin{proof} [Proof of Theorem \ref{thm:estimation}]
In this proof, the functions $h_1(r,u)$, $h_0(r,u)$, and $g(t,r,e)$ are only evaluated at $r = 
\bar{r}$. For simplicity, we omit this argument $\bar{r}$ throughout.
The proof proceeds with seven steps:
\begin{itemize}
    \item \hyperlink{thm2pf1}{Step 1} contains preliminary results on the LLR estimator of the condition distribution $Y|T,R$.
    \item \hyperlink{thm2pf2}{Step 2} derives the consistency of $\hat{\gamma}$.
    \item \hyperlink{thm2pf3}{Step 3} derives an initial estimate of the convergence rate of $\hat{\gamma}$.
    \item \hyperlink{thm2pf4}{Step 4} proves a stochastic equicontinuity condition on the criterion function.
    \item \hyperlink{thm2pf5}{Step 5} presents a linear approximation of the criterion function.
    \item \hyperlink{thm2pf6}{Step 6} shows the asymptotic normality of the minimizer of the linearized criterion function.
    \item \hyperlink{thm2pf7}{Step 7} derives the asymptotic normal distribution of $\hat{\gamma}$.
\end{itemize}

\noindent \hypertarget{thm2pf1}{\textbf{Step 1.}} (Preliminary results on LLR.) Let $ X_i(t) = (1,(T_i-t)/b_1,(R_i-\bar{r})/b_1)'$ denote the vector containing the regressors in the LLR. For $\bm{x} = (1,x_1,x_2)$, let
\begin{align*}
    k_0(\bm{x}) = k_T(x_1) k_R(x_2) \mathbf{1}\{ x_2 < 0 \}, \\
    k_1(\bm{x}) = k_T(x_1) k_R(x_2) \mathbf{1}\{ x_2 \geq 0 \}.
\end{align*}
The kernel weights in the LLR can be written as
\begin{align*}
    k_0(X_i(t)) = k_T\left( (T_i - t)/b_1 \right) k_R\left( (R_i - \bar{r})/b_1 \right) \mathbf{1}\{R_i < \bar{r}\}, \\
    k_1(X_i(t)) = k_T\left( (T_i - t)/b_1 \right) k_R\left( (R_i - \bar{r})/b_1 \right) \mathbf{1}\{R_i \geq \bar{r}\}.
\end{align*}
From \cite{xie2021uniform}, we have the following uniform asymptotic linear representation for the LLR estimator of the conditional distribution functions:
\begin{align*}
	\hat{F}^-_{Y|T,R}(y|t,\bar{r}) - F^-_{Y|T,R}(y|t,\bar{r}) & = b_1^2 \mu_0(y,t) + \iota' \Xi_0(t)^{-1} \frac{1}{nb_1^2} \sum_{i=1}^n s_0(Y_i,T_i,R_i,y,t) \\
    & \quad + O_p \left( b_1^3 + \frac{|\log b_1|}{nb_1^2} \right), 
\end{align*}
uniformly over $y \in \mathbb{R},t \in [t_0',t_0'']$, and 
\begin{align*}
    \hat{F}^+_{Y|T,R}(y|t,\bar{r}) - F^+_{Y|T,R}(y|t,\bar{r}) & = b_1^2 \mu_1(y,t) +  \iota' \Xi_1(t)^{-1} \frac{1}{nb_1^2} \sum_{i=1}^n s_1(Y_i,T_i,R_i,y,t) \\
    & \quad + O_p \left( b_1^3 + \frac{|\log b_1|}{nb_1^2} \right),
\end{align*}
uniformly over $y \in \mathbb{R},t \in [t_1',t_1'']$. In the above expressions, $\iota = (1,0,\cdots,0)$. The functions $\mu_0(y,t)$ and $\mu_1(y,t)$ are defined by
\begin{align*}
    \mu_0(y,t) & = \frac{b_1^2}{2}  \iota' \Omega_0(t)^{-1} \frac{\partial^2 }{\partial t^2} F^-_{Y|T,R}(y|t,\bar{r}) \int \bm{x} x_1^2 k_0(\bm{x}) \mathbf{1}\{ t+ b_1x_1 \in [t_0',t_0''] \} dx_1dx_2  \\
    & \quad + \frac{b_1^2}{2}  \iota' \Omega_0(t)^{-1} \frac{\partial^2 }{\partial r^2} F^-_{Y|T,R}(y|t,\bar{r}) \int \bm{x} x_2^2 k_0(\bm{x}) \mathbf{1}\{ t+ b_1x_1 \in [t_0',t_0''] \} dx_1dx_2  \\
    & \quad + b_1^2 \iota' \Omega_0(t)^{-1} \frac{\partial^2 }{\partial t \partial r} F^-_{Y|T,R}(y|t,\bar{r}) \int \bm{x} x_1x_2 k_0(\bm{x}) \mathbf{1}\{ t+ b_1x_1 \in [t_0',t_0''] \} dx_1dx_2  \\
    & \quad + \frac{b_2^2}{2} \Omega_0(t)^{-1} \frac{\partial^2 }{\partial y^2} F^-_{Y|T,R}(y|t,\bar{r}) \int \bm{x} k_0(\bm{x}) \mathbf{1}\{ t+ b_1x_1 \in [t_0',t_0''] \} dx_1dx_2, 
\end{align*}
and
\begin{align*}
    \mu_1(y,t) & = \frac{b_1^2}{2}  \iota' \Omega_1(t)^{-1} \frac{\partial^2 }{\partial t^2} F^+_{Y|T,R}(y|t,\bar{r}) \int \bm{x} x_1^2 k_1(\bm{x}) \mathbf{1}\{ t+ b_1x_1 \in [t_1',t_1''] \} dx_1dx_2  \\
    & \quad + \frac{b_1^2}{2}  \iota' \Omega_1(t)^{-1} \frac{\partial^2 }{\partial r^2} F^+_{Y|T,R}(y|t,\bar{r}) \int \bm{x} x_2^2 k_1(\bm{x}) \mathbf{1}\{ t+ b_1x_1 \in [t_1',t_1''] \} dx_1dx_2  \\
    & \quad + b_1^2 \iota' \Omega_1(t)^{-1} \frac{\partial^2 }{\partial t \partial r} F^+_{Y|T,R}(y|t,\bar{r}) \int \bm{x} x_1x_2 k_1(\bm{x}) \mathbf{1}\{ t+ b_1x_1 \in [t_1',t_1''] \} dx_1dx_2  \\
    & \quad + \frac{b_2^2}{2} \Omega_1(t)^{-1} \frac{\partial^2 }{\partial y^2} F^+_{Y|T,R}(y|t,\bar{r}) \int \bm{x} k_1(\bm{x}) \mathbf{1}\{ t+ b_1x_1 \in [t_1',t_1''] \} dx_1dx_2.
\end{align*}
The matrices $\Omega_0(t)$, $\Omega_1(t)$, $\Xi_0(t)$, and $\Xi_1(t)$ are defined by
\begin{align*}
    \Omega_0(t) & = \int \bm{x} \bm{x}' k_0(\bm{x}) \mathbf{1}\{ t+ b_1x_1 \in [t_0',t_0''] \} dx_1dx_2, \\
    \Omega_1(t) & = \int \bm{x} \bm{x}' k_1(\bm{x}) \mathbf{1}\{ t+ b_1x_1 \in [t_1',t_1''] \} dx_1dx_2,
\end{align*}
and
\begin{align*}
    \Xi_0(t) & = \int \bm{x} \bm{x}' k_0(\bm{x}) f_{T,R}^-(t+b_1x_1,\bar{r}+b_1x_2) dx_1dx_2, \\
    \Xi_1(t) & = \int \bm{x} \bm{x}' k_1(\bm{x}) f_{T,R}^-(t+b_1x_1,\bar{r}+b_1x_2) dx_1dx_2. 
\end{align*}
The terms $s_0$ and $s_1$ are defined by 
\begin{align*}
	s_0(Y_i,T_i,R_i,y,t) & = X_i(t) \tilde{K}_Y(Y_i,T_i,R_i;y,t) k_0(X_i(t)), \\
    s_1(Y_i,T_i,R_i,y,t) & = X_i(t) \tilde{K}_Y(Y_i,T_i,R_i;y,t) k_1(X_i(t)), \\
    \tilde{K}_Y(Y_i,T_i,R_i,y) & = K_Y \left( (y - Y_i)/b_1 \right) - \mathbb{E} \left[ K_Y \left( (y - Y_i)/b_1 \right) | T_i,R_i \right].
\end{align*}
Under Assumption \ref{ass:smoothness-distribution}(i) and (ii) and Assumption \ref{ass:kernels}(i), we can apply Lemma 1 in \cite{xie2021uniform}, which is a modification of Lemma 11 in \cite{fan2016multivariate}, and obtain that the eigenvalues of $\Omega_0(t)$, $\Omega_1(t)$, $\Xi_0(t)$ and $\Xi_1(t)$ are bounded and bounded away from zero for all values of $t$ and $b_1$. Consequently, the norm of these matrices and there inverses are bounded.

Notice that in \cite{xie2021uniform}, the remainder term from the bias expansion is $o(b_1^2)$ while in the above asymptotic linear representation, the corresponding term is $O(b_1^3)$. This is because we assume that $F^-_{Y|T,R}$ and $F^+_{Y|T,R}$ are three-times continuously differentiable (Assumption \ref{ass:smoothness-distribution}). Under this assumption, we can go through the same steps as in the proof of Theorem 1 in \cite{xie2021uniform} and show that the remainder from the bias expansion is $O(b_1^3)$. The details are omitted for brevity. The bias terms $\mu_0(y,t)$ and $\mu_1(y,t)$ are continuously differentiable under Assumption \ref{ass:smoothness-distribution}. The indicator functions inside the integral, for example, $\mathbf{1}\{ t+ b_1x_1 \in [t_0',t_0''] \}$, can be eliminated by explicitly indicating the lower and upper limits of the corresponding integral. The derivative can then be taken by using the Leibniz rule.

\bigskip
\noindent \hypertarget{thm2pf2}{\textbf{Step 2.}} (Consistency of $\hat{\gamma}$.) Because $w$ is positive and integrates to $1$, we have $\lVert\hat{D}_{\gamma,\hat{h}} - D_{\gamma,\hat{h}} \rVert_w$ $\leq$ $\lVert \hat{D}_{\gamma,\hat{h}} - D_{\gamma,\hat{h}} \rVert_\infty$. For any $\gamma \in \Gamma$, we can apply Fubini's theorem to the uniform asymptotic linear representation and obtain that
\begin{align*}
    \hat{D}_{\gamma,\hat{h}}(e,u) - D_{\gamma,\hat{h}}(e,u) = \RomNum{1} + \RomNum{2} + O_p \left( b_1^2 + \log n/(nb_1^2) \right)
\end{align*}
uniformly over $\gamma \in \Gamma, e \in \mathcal{E},$ and $u \in (0,1)$, where
\begin{align*}
    \RomNum{1} & = \frac{1}{nb_1^2} \sum_{i=1}^n \int_0^u \iota' \Xi_0(\hat{h}_0(v))^{-1}   s_0(Y_i,T_i,R_i;g_{\gamma}(\hat{h}_0(v),e),\hat{h}_0(v)) dv, \\
    \RomNum{2} &=   \frac{1}{nb_1^2} \sum_{i=1}^n \int_0^u \iota' \Xi_1(\hat{h}_1(v))^{-1}   s_1(Y_i,T_i,R_i;g_{\gamma}(\hat{h}_1(v),e),\hat{h}_1(v)) dv.
\end{align*}
By symmetry, we only need to study the term $\RomNum{1}$. Denote $\mathbf{1}_{\hat{h}_0} = \mathbf{1}\{ \hat{h}_0 \in \mathcal{H}_0(\mathcal{P}_0^n) \}$ as the indicator of whether $\hat{h}_0 \in \mathcal{H}_0(\mathcal{P}_0^n)$. Then $\RomNum{1} \leq \RomNum{1}.1 + \RomNum{1}.2$,
where 
\begin{align*}
    \RomNum{1}.1 & = \sup_{\gamma \in \Gamma, h_0 \in \mathcal{H}_0(\mathcal{P}_0^n), e \in \mathcal{E}, u \in (0,1)} \left| \frac{1}{nb_1^2} \sum_{i=1}^n \int_0^u \iota' \Xi_0(h_0(v))^{-1}   s_0(Y_i,T_i,R_i;g_{\gamma}(h_0(v),e),h_0(v)) dv \right|, \\
    \RomNum{1}.2 & = (1- \mathbf{1}_{\hat{h}_0}) \sup_{y \in \mathcal{Y}, t \in [t_0',t_0'']} \left| \frac{1}{nb_1^2} \sum_{i=1}^n \iota' \Xi_0(t)^{-1}   s_0(Y_i,T_i,R_i;y,t) \right|
\end{align*}
Define
\begin{align} \label{eqn:alpha-n}
    \tilde{\alpha}_n =  n^{-1/2 } b_1^{-7/12 - \bar{\epsilon}/5} \text{, and } \alpha_n = \left( b_1^2 + \sqrt{\log n / (nb_1^4)} \right) (b_1^2 + \tilde{\alpha}_n) ,
\end{align}
where $\bar{\epsilon}$ is defined in Assumption \ref{ass:bandwidth}.
In Lemma \ref{lm:uc1}, we show that $\RomNum{1}.1 = O_p(\tilde{\alpha}_n)$. Combining Assumption \ref{ass:h-tilde}(i) and Lemma \ref{lm:uc3}, we have $\RomNum{1}.2 = O_p(\tilde{\alpha}_n)$. Therefore,
\begin{align} \label{eqn:Dhat-uniform-rate}
    \sup_{\gamma \in \Gamma}\norm{\hat{D}_{\gamma,\hat{h}} - D_{\gamma,\hat{h}}}_w & O_p(b_1^2 + \tilde{\alpha}_n).
\end{align}
By the smoothness of $F^-_{Y,T,R}$ and $g_\gamma$, the following term is $O(\lVert \hat{h} - h^* \rVert_\infty)$:
\begin{align*}
    \sup_{\gamma \in \Gamma} \left| \int_0^u \left( F^-_{Y|T,R}( g_{\gamma}(\hat{h}_0(v),e) | h_0(v),\bar{r} ) - F^-_{Y|T,R}( g_{\gamma}(h^*_0(v),e) | h_0(v),\bar{r} ) \right) dv \right| 
\end{align*}
Therefore, we obtain that uniformly over $\gamma \in \Gamma$, 
\begin{align*}
    \norm{D_{\gamma,\hat{h}} - D_{\gamma,h^*} }_w = O(\lVert \hat{h} - h^*) \rVert_\infty = O_p\left( b_1^2 +\sqrt{\log n / (nb_1)} \right).
\end{align*}
By the triangle inequality, we have
\begin{align*}
    \norm{D_{\hat{\gamma},h^*}}_w & \leq \norm{D_{\hat{\gamma},h^*} - D_{\hat{\gamma},\hat{h}}}_w + \norm{\hat{D}_{\hat{\gamma},\hat{h}} - D_{\hat{\gamma},\hat{h}}}_w + \norm{\hat{D}_{\hat{\gamma},\hat{h}}}_w .
\end{align*}
By the definition of $\hat{\gamma}$ in (\ref{eqn:gamma-hat-def}), we have
\begin{align*}
    \norm{\hat{D}_{\hat{\gamma},\hat{h}}}_w & \leq \norm{\hat{D}_{\gamma^*,\hat{h}}}_w + o_p\left( \alpha_n \right) \leq \norm{\hat{D}_{\gamma^*,\hat{h}} - D_{\gamma^*,\hat{h}} }_w + \norm{D_{\gamma^*,\hat{h}} - D_{\gamma^*,h^*} }_w + o_p\left( \alpha_n \right).
\end{align*}
Combining the above two inequalities together, we obtain that
\begin{align} \label{eqn:D-hatgamma-hstar-rate}
    \norm{D_{\hat{\gamma},h^*}}_w & \leq 2 \sup_{\gamma \in \Gamma} \norm{D_{\gamma,\hat{h}} - D_{\gamma,h^*} }_w + 2 \sup_{\gamma \in \Gamma} \norm{\hat{D}_{\gamma,\hat{h}} - D_{\gamma,\hat{h}} }_w + o_p\left(\alpha_n \right) = O_p\left( b_1^2 + \tilde{\alpha}_n \right).
\end{align}
In particular, the above quantity is $o_p(1)$.
Because $\Gamma$ is compact, and $\norm{D_{\cdot,h^*}}_w$ is continuous and has a unique minimizer $\gamma^*$ (Corollary \ref{cor:gamma-star-identification}), for any $\varepsilon>0$ there exists $\delta>0$ such that $\norm{\gamma - \gamma^*}_2 > \epsilon$ $\implies$ $\norm{D_{\gamma,h^*}}_w > \delta$. Therefore, $\mathbb{P}(\norm{\hat{\gamma} - \gamma^*}_2 > \epsilon )$ $\leq$ $\mathbb{P}(\norm{D_{\gamma,h^*}}_w > \delta) = o(1) $. This proves that $\hat{\gamma}$ is a consistent estimator.

\bigskip
\noindent \hypertarget{thm2pf3}{\textbf{Step 3.}} (Convergence rate of $\hat{\gamma}$.)
Since $\hat{\gamma}$ is consistent, we can Taylor expand $D_{\hat{\gamma},h^*}$ around $\gamma^*$. Together with the reverse triangle inequality and the fact that $D_{\gamma^*,h^*} = 0$, the expansion gives that
         \begin{align*}
             \big \lVert D_{\hat{\gamma},h^*} \big \rVert_{w} & = \big \lVert \nabla_\gamma D_{\gamma^*,h^*} (\hat{\gamma} - \gamma^*) + (\hat{\gamma} - \gamma^*)' \nabla^2_\gamma D_{\tilde{\gamma},h^*} (\hat{\gamma} - \gamma^*) \big \rVert_{w} \\
             & \geq \big \lVert \nabla_\gamma D_{\gamma^*,h^*} (\hat{\gamma} - \gamma^*) \big \rVert_{w} - \lVert \hat{\gamma} - \gamma^* \rVert_2^2 \int_0^1 \int_\mathcal{E} \big \lVert \nabla^2_\gamma D_{\tilde{\gamma},h^*}(u,e) \big \rVert_2 w(e,u) dedu \\
             & \geq \big \lVert \nabla_\gamma D_{\gamma^*,h^*} (\hat{\gamma} - \gamma^*) \big \rVert_{w} + O\big(\lVert \hat{\gamma} - \gamma^* \rVert^2_2 \big),
         \end{align*} 
         where $\tilde{\gamma}$ is some point on the line segment connecting $\hat{\gamma}$ and $\gamma^*$ and the last line follows from Assumption \ref{ass:complexity-parametric-model}(iii) that $\big \lVert \nabla^2_\gamma D_{\tilde{\gamma},h^*}(u,e) \big \rVert_2$ is bounded. We claim that there exists a universal constant $C>0$ such that 
         \begin{align} \label{eqn:linearly-independent}
            \lVert \nabla_\gamma D_{\gamma^*,h^*} \zeta \big \rVert_{w} \geq C \lVert \zeta \rVert, \text{ for all } \zeta \in \mathbb{R}^{d_\Gamma}.
         \end{align}
         If this claim is true, then by Equation (\ref{eqn:D-hatgamma-hstar-rate}), we obtain a bound on the convergence rate of $\hat{\gamma}$:
         \begin{align*}
            \lVert \hat{\gamma} - \gamma^* \rVert_2=O_p\left( b_1^2 + \tilde{\alpha}_n \right).
         \end{align*}

         The remaining part of this step is devoted to the proof of (\ref{eqn:linearly-independent}). Suppose that claim is false, then for each integer $k \geq 1$, there exists $\zeta_k \in \mathbb{R}^{d_\Gamma}$ such that $\lVert \nabla_\gamma D_{\gamma^*,h^*} \zeta_k \big \rVert_{w} < 1/k \lVert \zeta_k \rVert$. Without loss of generality, we can assume $\lVert \zeta_k \rVert = 1$ (or simply redefine the sequence as $\zeta_k/\lVert \zeta_k \rVert$). By the Bolzano–Weierstrass theorem, the sequence $\{\zeta_k\}$ has a convergent subsequence. Without loss of generality, we assume $\{\zeta_k\}$ itself is convergent with the limit denoted by $\zeta_\infty$. Then it must be the case that $\lVert \nabla_\gamma D_{\gamma^*,h^*} \zeta_\infty \big \rVert_{w} = 0$. Since $\nabla_\gamma D_{\gamma^*,h^*}$ is a continuous function, the previous equation implies that $\nabla_\gamma D_{\gamma^*,h^*} \zeta_\infty = 0$. This violates Assumption \ref{ass:complexity-parametric-model}(iv) that $\nabla_\gamma D_{\gamma^*,h^*}$ is a vector of linearly independent functions.

    \bigskip
    \noindent \hypertarget{thm2pf4}{\textbf{Step 4.}} (Stochastic equicontinuity of the criterion function.) Let $\gamma_n \overset{p}{\rightarrow} \gamma^*$ be such that $\norm{\gamma_n - \gamma^*} = O_p(b_1^2 + \tilde{\alpha}_n)$. We want to find the asymptotic order of the term $\lVert \hat{D}_{\gamma_n,\hat{h}} - D_{\gamma_n,\hat{h}} - \hat{D}_{\gamma^*,h^*} \rVert_w$, which is bounded by
    \begin{align*}
        \sup_{e \in \mathcal{E}, u \in (0,1)} \left| \hat{D}_{\gamma_n,\hat{h}}(e,u) - D_{\gamma_n,\hat{h}}(e,u) - ( \hat{D}_{\gamma^*,h^*}(e,u) - D_{\gamma^*,h^*}(e,u)) \right| \leq \RomNum{1} + \RomNum{2},
    \end{align*}
    where
    \begin{align*}
        \RomNum{1} & = \sup_{e \in \mathcal{E}, u \in (0,1)} \Big| \hat{F}^-_{Y|T,R}( g_{\gamma_n}(\hat{h}_0(v),e) | \hat{h}_0(v),\bar{r} ) - F^-_{Y|T,R}( g_{\gamma_n}(\hat{h}_0(v),e) | \hat{h}_0(v),\bar{r} ) \\
        & \quad - \left( \hat{F}^-_{Y|T,R}( g_{\gamma^*}(h^*_0(v),e) | h^*_0(v),\bar{r} ) - F^-_{Y|T,R}( g_{\gamma^*}(h^*_0(v),e) | h^*_0(v),\bar{r} ) \right) \Big|, 
    \end{align*}
    and
    \begin{align*}
        \RomNum{2} & = \sup_{e \in \mathcal{E}, u \in (0,1)} \Big| \hat{F}^+_{Y|T,R}( g_{\gamma_n}(\hat{h}_1(v),e) | \hat{h}_1(v),\bar{r} ) - F^-_{Y|T,R}( g_{\gamma_n}(\hat{h}_1(v),e) | \hat{h}_1(v),\bar{r} ) \\
        & \quad - \left( \hat{F}^+_{Y|T,R}( g_{\gamma^*}(h^*_1(v),e) | h^*_1(v),\bar{r} ) - F^+_{Y|T,R}( g_{\gamma^*}(h^*_1(v),e) | h^*_1(v),\bar{r} ) \right) \Big|.
    \end{align*}
    By symmetry, we only need to study the term $\RomNum{1}$. The uniform asymptotic linear representation of the LLR estimators gives a bias-variance decomposition that 
    \begin{align*}
        \RomNum{1} \leq \RomNum{1}.1 + \RomNum{1}.2 + O_p \left( b_1^3 + |\log b_1|/nb_1^2 \right),
    \end{align*}
    where
    \begin{align*}
        \RomNum{1}.1 & = b_1^2 \left( \mu_0(g_{\gamma_n}(\hat{h}_0(v),e),\hat{h}_0(v)) - \mu_0(g_{\gamma^*}(h^*_0(v),e),h^*_0(v)) \right), \\
        \RomNum{1}.2 & = \sup_{e \in \mathcal{E}, u \in [0,1]} \Big| \frac{1}{nb_1^2} \sum_{i=1}^n  \iota' \Xi_0(\hat{h}_0(v))^{-1} s_0(Y_i,T_i,R_i;g_{\gamma_n}(\hat{h}_0(v),e),\hat{h}_0(v))  \\
        & \quad - \iota' \Xi_0(h^*_0(v))^{-1} s_0(Y_i,T_i,R_i;g_{\gamma^*}(h^*_0(v),e),h^*_0(v)) \Big|.
    \end{align*}
    By the smoothness of $\mu_0$ (Step 1) and $g_\gamma$ (Assumption \ref{ass:complexity-parametric-model}), we can bound the term $\RomNum{1}.1$ by 
    \begin{align*}
        \RomNum{1}.1 \leq Cb_1^2 (\norm{\gamma_n - \gamma^*}_2 + \lVert\hat{h} - h^*\rVert_\infty) = O_p(b_1^4 + b_1^2 \tilde{\alpha}_n).
    \end{align*}
    For the term $\RomNum{1}.2$, consider the decomposition that $\RomNum{1}.2 \leq \RomNum{1}.2.1 + \RomNum{1}.2.2$, where 
    \begin{align*}
        \RomNum{1}.2.1 & = \sup_{e \in \mathcal{E}, u \in [0,1]} \Big| \iota' \Xi_0(\hat{h}_0(v))^{-1} \frac{1}{nb_1^2} \sum_{i=1}^n s_0(Y_i,T_i,R_i;g_{\gamma_n}(\hat{h}_0(v),e),\hat{h}_0(v))  \\
        & \quad - s_0(Y_i,T_i,R_i;g_{\gamma^*}(h^*_0(v),e),h^*_0(v)) \Big|, \\
        \RomNum{1}.2.2 & = \sup_{e \in \mathcal{E}, u \in [0,1]} \Big| \iota' \Big( \Xi_0(\hat{h}_0(v))^{-1} - \Xi_0(h^*_0(v))^{-1} \Big) \frac{1}{nb_1^2} \sum_{i=1}^n s_0(Y_i,T_i,R_i;g_{\gamma^*}(h^*_0(v),e),h^*_0(v)) \Big| .
    \end{align*}
    As mentioned in Step 1, we know that $\norm{\Xi(t)^{-1}}_2$ is bounded for $t \in [t_0',t_0'']$ by Lemma 1 in \cite{xie2021uniform}. Applying the mean value theorem, we obtain that
    \begin{align*}
        \RomNum{1}.2.1 \leq C(\RomNum{1}.2.1.1 + \RomNum{1}.2.1.2 + \RomNum{1}.2.1.3)
    \end{align*}
    where
    \begin{align*}
        \RomNum{1}.2.1.1 & = \sup_{y \in \mathbb{R}, t \in [t_0',t_0'']} \norm{ \frac{1}{nb_1^2} \sum_{i=1}^n \frac{\partial }{\partial y} s_0(Y_i,T_i,R_i;y,t)}_2 \norm{ \nabla_\gamma g_{\gamma^*} }_\infty \norm{\gamma_n - \gamma^*}_2, \\
        \RomNum{1}.2.1.2 & = \sup_{y \in \mathbb{R}, t \in [t_0',t_0'']} \norm{ \frac{1}{nb_1^2} \sum_{i=1}^n \frac{\partial }{\partial y} s_0(Y_i,T_i,R_i;y,t)}_2 \norm{ \frac{\partial }{\partial T} g_{\gamma^*} }_\infty \lVert\hat{h} - h^*\rVert_\infty, \\
        \RomNum{1}.2.1.3 & = \sup_{y \in \mathbb{R}, t \in [t_0',t_0'']} \norm{ \frac{1}{nb_1^2} \sum_{i=1}^n \frac{\partial }{\partial t} s_0(Y_i,T_i,R_i;y,t)}_2 \lVert\hat{h} - h^*\rVert_\infty.
    \end{align*}
    In Lemma \ref{lm:uc2}, we show that the following two terms are of order $O_p\left( \sqrt{\log n /(nb_1^4)} \right)$:
    \begin{align*}
        \sup_{y \in \mathbb{R}, t \in [t_0',t_0'']} \left| \frac{1}{nb_1^2} \sum_{i=1}^n \frac{\partial }{\partial y} s(Y_i,T_i,R_i;y,t)\right|, \sup_{y \in \mathbb{R}, t \in [t_0',t_0'']} \left| \frac{1}{nb_1^2} \sum_{i=1}^n \frac{\partial }{\partial t} s(Y_i,T_i,R_i;y,t)\right|.
    \end{align*}
    Because $\norm{ \nabla_\gamma g_{\gamma^*} }_\infty$ and $\norm{ \partial g_{\gamma^*}  /\partial T }_\infty$ are finite, we know that 
    \begin{align*}
        \RomNum{1}.2.1 = O_p\left( \sqrt{\log n /(nb_1^4)}\right) \times(\norm{\gamma_n - \gamma^*}_2 + \lVert\hat{h} - h^*\rVert_\infty) = O_p\left( \sqrt{\log n /(nb_1^4)} \tilde{\alpha}_n \right).
    \end{align*}
    Applying the mean value theorem to $\RomNum{1}.2.2$, we obtain that
    \begin{align*}
        \RomNum{1}.2.2 & \leq \sup_{t \in [t_0',t_0'']} \left| \iota' \frac{\partial }{\partial t} \Xi_0(t)^{-1}  \right| \sup_{y \in \mathbb{R}, t \in [t_0',t_0''] }\norm{ \frac{1}{nb_1^2} \sum_{i=1}^n s_0(Y_i,T_i,R_i;y,t) }_2 \lVert\hat{h} - h^*\rVert_\infty.
    \end{align*}
    In Lemma \ref{lm:uc3}, we show that
    \begin{align*}
        \sup_{y \in \mathbb{R}, t \in [t_0',t_0''] }\norm{  \sum_{i=1}^n s_0(Y_i,T_i,R_i;y,t)/(nb_1^2) }_2 = O_p\left( \sqrt{\log n /(nb_1^2)} \right) .
    \end{align*}
    Therefore, $\RomNum{1}.2.1$ asymptotically dominates $\RomNum{1}.2.2$. Hence, the term $\RomNum{1}$ is of the following order:
    \begin{align*}
        \RomNum{1} = O_p\left( \left(b_1^2 + \sqrt{\log n /(nb_1^4)} \right) \tilde{\alpha}_n \right)= O_p(\alpha_n) .
    \end{align*}
    Based on the same argument, the above asymptotic order also applies to the term $\RomNum{2}$. 
    Thus, we have the following stochastic equicontinuity result:
    \begin{align} \label{eqn:se}
        \lVert \hat{D}_{\gamma_n,\hat{h}} - D_{\gamma_n,\hat{h}} - \hat{D}_{\gamma^*,h^*} \rVert_w = O_p(\alpha_n).
    \end{align}

    \bigskip
    \noindent \hypertarget{thm2pf5}{\textbf{Step 5.}} (Linearization of the criterion function.)
         Let $\partial_h^{[\hat{h} - h^*]} D_{\gamma,h^*}(e,u)$ be the Fr\'echet derivative of $D_{\gamma,h}(e,u)$ with respect to $h$ at $h^*$, in the direction of $h-h^*$. That is,
    \begin{align*}
        \partial_h^{[\hat{h} - h^*]} D_{\gamma,h^*}(e,u) & = \int_0 ^u (\phi^-_\gamma(e,v) - \phi^+_\gamma(e,v)) (\hat{h}_0(v) - h^*_0(v)) dv ,
    \end{align*}
    where
    \begin{align*}
        \phi^-_\gamma(e,v) & = \frac{\partial}{\partial Y} F^-_{Y|T,R}( g_\gamma(h^*_0(v),e)|h^*_0(v),\bar{r}) \frac{\partial}{\partial T} g_\gamma(h^*_0(v),e) + \frac{\partial}{\partial T} F^-_{Y|T,R}( g_\gamma(h^*_0(v),e)|h^*_0(v),\bar{r}), \\
        \phi^+_\gamma(e,v) & = \frac{\partial}{\partial Y} F^+_{Y|T,R}( g_\gamma(h^*_1(v),e)|h^*_1(v),\bar{r}) \frac{\partial}{\partial T} g_\gamma(h^*_1(v),e) + \frac{\partial}{\partial T} F^+_{Y|T,R}( g_\gamma(h^*_1(v),e)|h^*_1(v),\bar{r}).
    \end{align*}
    It is straightforward to see that $\norm{\partial_h^{[\hat{h} - h^*]} D_{\gamma,h^*}}_\infty$ $= O(\lVert\hat{h} - h^*\rVert_\infty)$.
    Following the same steps as in Lemma 4 of \cite{TORGOVITSKY2017minimum}, we can show that $\lVert  D_{\gamma,\hat{h}} - D_{\gamma,h^*} - \partial_h^{[\hat{h} - h^*]} D_{\gamma,h^*} \rVert_w = O(\lVert\hat{h} - h^*\rVert_\infty^2)$, uniformly over $\gamma \in \Gamma$, and
    \begin{align*}
        \norm{ \partial_h^{[\hat{h} - h^*]} D_{\gamma_n,h^*} - \partial_h^{[\hat{h} - h^*]} D_{\gamma^*,h^*} }_w = O(\norm{\gamma_n - \gamma^*}_2 \lVert\hat{h} - h^*\rVert_\infty),
    \end{align*}
    for any sequence $\gamma_n \overset{p}{\rightarrow} \gamma^*$. Define 
    \begin{align*}
        \hat{L}_\gamma(e,u) = \hat{D}_{\gamma^*,h^*}(e,u) + \nabla_\gamma D_{\gamma^*,h^*}(e,u) (\gamma - \gamma^*) + \partial_h^{[\hat{h} - h^*]} D_{\gamma^*,h^*}(e,u),
    \end{align*}
    as a linear approximation of $\hat{D}_{\gamma,\hat{h}}(e,u)$ for $\gamma$ near $\gamma^*$. For any sequence $\norm{\gamma_n - \gamma^*}_2 = O_p(b_1^2 + \tilde{\alpha}_n )$, we have
    \begin{align} \label{eqn:Lgamma-rate}
        \norm{ \hat{L}_{\gamma_n} }_w & \leq \norm{ \hat{D}_{\gamma^*,h^*} - D_{\gamma^*,h^*}}_w + \norm{ \nabla_\gamma D_{\gamma^*,h^*}  }_w \norm{\gamma_n - \gamma^*}_2 + O( \lVert\hat{h} - h^*\rVert_\infty) = O_p(b_1^2 + \tilde{\alpha}_n ) ,
    \end{align}
    where the asymptotic order of the first term on the RHS is derived in (\ref{eqn:Dhat-uniform-rate}). 
    We want to bound the approximation error from the linearization of the criterion function. By adding and subtracting terms, we obtain that
    \begin{align} \label{eqn:approx-error-linearization1}
        & \norm{ \hat{L}_{\gamma_n} - \hat{D}_{\gamma_n,\hat{h}} }_w \nonumber \\
         \leq & \norm{ \hat{D}_{\gamma^*,h^*} - \hat{D}_{\gamma_n,\hat{h}} - (D_{\gamma^*,h^*} - D_{\gamma_n,\hat{h}}) } + \norm{D_{\gamma^*,h^*} + \nabla_\gamma D_{\gamma^*,h^*}(\gamma - \gamma^*) - D_{\gamma_n,h^*}}_w \nonumber \\
        \quad & + \norm{D_{\gamma_n,h^*} + \partial_h^{[\hat{h} - h^*]} D_{\gamma_n,h^*} - D_{\gamma_n,\hat{h}}}_w + \norm{ \partial_h^{[\hat{h} - h^*]} D_{\gamma^*,h^*} - \partial_h^{[\hat{h} - h^*]} D_{\gamma_n,h^*} }_w .
    \end{align}
    The four terms on the RHS of the above inequality can be analyzed as the following. The order of the first term on the RHS of (\ref{eqn:approx-error-linearization1}) is given by (\ref{eqn:se}) in the previous step.
    The second term is $O(\norm{\gamma_n - \gamma^*}_2)$ by the smoothness of $D_{\gamma,h^*}$. The third term is bounded by 
    \begin{align*}
        \sup_{\gamma \in \Gamma} \norm{  D_{\gamma,\hat{h}} - D_{\gamma,h^*} - \partial_h^{[\hat{h} - h^*]} D_{\gamma,h^*} }_w = O(\lVert\hat{h} - h^*\rVert_\infty^2).
    \end{align*}
    The fourth term is $O(\norm{\gamma_n - \gamma^*}_2 \lVert\hat{h} - h^*\rVert_\infty)$. Therefore, the leading term on the RHS of (\ref{eqn:approx-error-linearization1}) is the first term, and hence the approximation error from the linearization of the criterion function is of the following order:
    \begin{align} \label{eqn:approx-error-linearization2}
        \norm{ \hat{L}_{\gamma_n} - \hat{D}_{\gamma_n,\hat{h}} }_w = O_p(\alpha_n) .
    \end{align}

    \bigskip
    \noindent \hypertarget{thm2pf6}{\textbf{Step 6.}} (Minimizer of the linearized criterion function.)
    Define $\tilde{\gamma}$ as the minimizer of $\norm{\hat{L}_\gamma}_w$. The first-order condition gives that
    \begin{align*}
        \Delta(\tilde{\gamma} - \gamma^*) = \int \nabla_\gamma D_{\gamma^*,h^*}(e,u) \left( \hat{D}_{\gamma^*,h^*}(e,u) + \partial_h^{[\hat{h} - h^*]} D_{\gamma^*,h^*}(e,u) \right) w(e,u) dedu ,
    \end{align*}
    where 
    \begin{align} \label{eqn:nabla-D}
        \Delta = \int \nabla_\gamma D_{\gamma^*,h^*}(e,u) \nabla_\gamma D_{\gamma^*,h^*}(e,u)' w(e,u) dedu.
    \end{align}
    By the uniform asymptotic linear representation of the LLR estimators and $\hat{h}$, we can write
    \begin{align*}
        & \int \nabla_\gamma D_{\gamma^*,h^*}(e,u) \left( \hat{D}_{\gamma^*,h^*}(e,u) + \partial_h^{[\hat{h} - h^*]} D_{\gamma^*,h^*}(e,u) \right) w(e,u) dedu \\
        = & \frac{1}{nb_1} \sum_{i=1}^n (\zeta_-^{\textit{DF}}(Y_i,T_i,R_i) + \zeta_-^{\textit{Q}}(Y_i,T_i,R_i)) - \frac{1}{nb_1} \sum_{i=1}^n (\zeta_+^{\textit{DF}}(Y_i,T_i,R_i) + \zeta_+^{\textit{Q}}(Y_i,T_i,R_i)) \\
        & + b_1^2(B_- - B_+) + O_p(b_1^3) + o_p(1/\sqrt{nb_1}).
    \end{align*}
    The terms $B_-$ and $B_+$ are deterministic bias terms defined by
    \begin{align}
        B_- & = \int  w(e,u) \nabla_\gamma D_{\gamma^*,h^*}(e,u) \left( \int_0^u \mu_0(g_{\gamma^*}(h_0^*(v),e),h_0^*(v)) + \phi^-_{\gamma^*}(e,v)\nu_0(v) \right) dedu, \label{eqn:bias-}\\
        B_+ & = \int  w(e,u) \nabla_\gamma D_{\gamma^*,h^*}(e,u) \left( \int_0^u \mu_1(g_{\gamma^*}(h_1^*(v),e),h_1^*(v)) + \phi^+_{\gamma^*}(e,v)\nu_1(v) \right) dedu. \label{eqn:bias+}
    \end{align}
    The functions $\zeta_-^{\textit{DF}}$ and $\zeta_+^{\textit{DF}}$ represent stochastic terms from the LLR estimation of the conditional distribution $F_{Y|T,R}$. They are defined by 
    \begin{align*}
        &\zeta_-^{\textit{DF}}(Y,T,R) \\
        = & \frac{1}{b_1} \int_{\mathcal{E}} \int_0^1 w(e,u) \nabla_\gamma D_{\gamma^*,h^*}(e,u) \int_0^u \iota' \Xi_0(h^*_0(v))^{-1} s_0(Y,T,R,g_{\gamma^*}(h^*_0(v),e),h^*_0(v))  dv dedu\\
        = & \int_{\mathcal{E}} \int_0^1 w(e,u) \nabla_\gamma  D_{\gamma^*,h^*}(e,u) \int_{(T-h^*_0(u))/b_1}^{(T-t_0')/b_1} \iota' \Xi_0(T + b_1 v)^{-1} (1,v,(R-\bar{r})/b_1)' \\
        & \times \tilde{K}_Y(Y,T,R;g_{\gamma}(T + b_1 v,e)) k_T(v) k_R^-((R-\bar{r})/b_1) ((h_0^*)^{-1})'(T+b_1v) dv, 
    \end{align*}
    and
    \begin{align*}
        &\zeta_+^{\textit{DF}}(Y,T,R) \\
        = & \frac{1}{b_1} \int_{\mathcal{E}} \int_0^1 w(e,u) \nabla_\gamma D_{\gamma^*,h^*}(e,u) \int_0^u \iota' \Xi_1(h^*_1(v))^{-1} s_1(Y_i,T_i,R_i,g_{\gamma^*}(h^*_1(v),e),h^*_1(v))  dv dedu\\
        = & \int_{\mathcal{E}} \int_0^1 w(e,u) \nabla_\gamma D_{\gamma^*,h^*}(e,u) \int_{(T-h_1(u))/b_1}^{(T-t_1')/b_1} \iota' \Xi_1(T + b_1 v)^{-1} (1,v,(r-\bar{r})/b_1)' \\
        & \times \tilde{K}_Y(Y,T,R;g_{\gamma^*}(T + b_1 v,e)) k_T(v) k_R^+((R-\bar{r})/b_1) ((h_0^*)^{-1})'(T+b_1v) dv.
    \end{align*}
    In the above notations, $k_R^-(x) = k_R(x) \mathbb{1}\{x < 0\}$ and $k_R^+(x) = k_R(x) \mathbb{1}\{x \geq 0\}$. Similarly define $k_{Q,0}^-(x) = k_{Q,0}(x)\mathbb{1}\{x < 0\}$ and $k_{Q,1}^+(x) = k_{Q,1}(x)\mathbb{1}\{x \geq 0\}$.
    The functions $\zeta_-^{\textit{Q}}$ and $\zeta_+^{\textit{Q}}$ represent stochastic terms from the nonparametric estimation of the conditional quantile function $h^*$. They are defined by 
    \begin{align*}
        \zeta_-^{\textit{Q}}(T,R) & = \int_{\mathcal{E}} \int_0^1 w(e,u) \nabla_\gamma D_{\gamma^*,h^*}(e,u) \int_0^u \phi^-_{\gamma^*}(e,v) q_{0}(T,R;v) k_{Q,0}^-\left( (R - \bar{r})/b_1 \right)  dv  dedu, \\
        \zeta_+^{\textit{Q}}(T,R) & = \int_{\mathcal{E}} \int_0^1 w(e,u) \nabla_\gamma D_{\gamma^*,h^*}(e,u) \int_0^u \phi^+_{\gamma^*}(e,v) q_{1}(T,R;v) k_{Q,1}^+\left( (R - \bar{r})/b_1 \right)  dv  dedu.
    \end{align*}
    By Fubini's theorem, we have 
    \begin{align*}
        \mathbb{E}[\zeta_\pm^{\textit{DF}}(Y,T,R) | T,R] = \mathbb{E} [\zeta_\pm^{\textit{Q}}(T,R)) | T,R] = 0.
    \end{align*}
    Notice that $((h_0^*)^{-1})'(\cdot) = f_{T|R}^-(\cdot | \bar{r})$ and $((h_1^*)^{-1})'(\cdot) = f_{T|R}^+(\cdot | \bar{r})$. The variance matrix can be computed as follows, where to save space, we use the notation of squaring a vector to mean the tensor product of that vector with itself. 
    \begin{align*}
        & \mathbb{E} \left[ (\zeta_-^{\textit{DF}}(Y,T,R) + \zeta_-^{\textit{Q}}(T,R)) \otimes (\zeta_-^{\textit{DF}}(Y,T,R) + \zeta_-^{\textit{Q}}(T,R)) \right] \\
        = & \int_{\mathcal{Y} \times [t_0',t_0''] \times [r_0,\bar{r}]} \Big( \int_{\mathcal{E}} \int_0^1  w(e,u) \nabla_\gamma D_{\gamma^*,h^*}(e,u) \int_{(t-h_0(u))/b_1}^{(t-t_0')/b_1} \iota' \Xi_0(t + b_1 v)^{-1} (1,v,(r-\bar{r})/b_1)' \\
        & \times \tilde{K}_Y(y,t,r;g_{\gamma^*}(t + b_1 v,e)) k_T(v) k_R^-((r-\bar{r})/b_1)f_{T|R}^-(t+b_1v | \bar{r}) dv, \\
        & + \int_0^u \phi^-_{\gamma*}(e,v) q_{0}(t,r;v) k_{Q,0}^-\left( (r - \bar{r})/b_1 \right) dv  dedu \Big)^2 f^-_{Y,T,R}(y,t,r) dydtdr,
    \end{align*}
    where $f^\pm_{Y,T,R}$ is defined analogously as $f^\pm_{T|R}$ and $F^\pm_{Y|T,R}$.
    Applying the change of variables $\tilde{r} = (r-\bar{r})/b_1$, we obtain that the above matrix is equal to $b_1$ times the matrix 
    \begin{align*} 
        & \int_{\mathcal{Y} \times [t_0',t_0''] \times [-1,0]} \Big( \int_{\mathcal{E}} \int_0^1 w(e,u) \nabla_\gamma D_{\gamma^*,h^*}(e,u) \int_{(t-h_0(u))/b_1}^{(t-t_0')/b_1} \iota' \Xi_0(t + b_1 v)^{-1} (1,v,\tilde{r})' \nonumber  \\
        & \times \tilde{K}_Y(y,t,\bar{r}+b_1\tilde{r};g_{\gamma^*}(t + b_1 v,e)) k_T(v) k_R^-(\tilde{r})f_{T|R}^-(t + b_1v| \bar{r}) dv \nonumber \\
        & + \int_0^u \phi^-_{\gamma*}(e,v) q_{0}(t,\bar{r}+b_1\tilde{r};v) k_{Q,0}^-\left( \tilde{r} \right) dv  dedu \Big)^2 f^-_{Y,T,R}(y,t,\bar{r}+b_1\tilde{r}) dydtd\tilde{r}.
    \end{align*}
    For any $t \in [t_0'+b_1,t_0''-b_1]$, we have $\Xi_0(t) = f^-_{T,R}(t,\bar{r}) \bar{\Omega}_0$ with $\bar{\Omega}_0 = \int \bm{x} \bm{x}' k_0(\bm{x}) dx_1 dx_2$.
    By letting $n \rightarrow \infty$ (so that $b_1 \rightarrow 0$) and using the continuity of the relevant functions and the dominated convergence theorem, we know that the above matrix is asymptotically equivalent to 
    \begin{align} \label{eqn:sigma-}
        \Sigma_- & = \int \Big( \int_{\mathcal{E}} \int_0^1 w(e,u) \nabla_\gamma D_{\gamma^*,h^*}(e,u) \Big( \iota' \bar{\Omega}_0^{-1} (1,0,\tilde{r})' \big(\mathbf{1}\{y \leq g_{\gamma^*}(t,e) \} \nonumber \\
        & \quad - F^-_{Y|T,R}(g_{\gamma^*}(t,e)|t,\bar{r}) \big) k_R^-(\tilde{r})/f_{R}(\bar{r}) \nonumber \\
        & \quad + \int_0^u \phi^-_{\gamma*}(e,v) q_{0}(t,\bar{r};v)k_{Q,0}^-\left( \tilde{r} \right) dv \Big) dedu \Big)^2 f^-_{Y,T,R}(y,t,\bar{r}) dydtd\tilde{r}.
    \end{align}
    In particular, we have used the following convergence result in the above expression:
    \begin{align*}
        K_Y((y-y')/b_1) & \rightarrow \mathbf{1}\{y' \leq y\}, \\
        \tilde{K}_Y(y,t,\bar{r}+b_1\tilde{r};g_{\gamma^*}(t + b_1 v,e)) & \rightarrow \mathbf{1}\{y \leq g_{\gamma^*}(t,e) \} - F^-_{Y|T,R}(g_{\gamma^*}(t,e)|t,\bar{r}).
    \end{align*}
    The above derivation shows that
    \begin{align*}
        \mathbb{E} \left[ (\zeta_-^{\textit{DF}}(Y,T,R) + \zeta_-^{\textit{Q}}(T,R)) \otimes (\zeta_-^{\textit{DF}}(Y,T,R) + \zeta_-^{\textit{Q}}(T,R)) \right] \sim b_1 \Sigma_-.
    \end{align*}
    Similarly, we can show that 
    \begin{align*}
        \mathbb{E} \left[ (\zeta_+^{\textit{DF}}(Y,T,R) + \zeta_+^{\textit{Q}}(T,R)) \otimes (\zeta_+^{\textit{DF}}(Y,T,R) + \zeta_+^{\textit{Q}}(T,R)) \right] \sim b_1 \Sigma_+,
    \end{align*}
    where
    \begin{align} \label{eqn:sigma+}
        \Sigma_+ & = \int \Big( \int_{\mathcal{E}} \int_0^1 w(e,u) \nabla_\gamma D_{\gamma^*,h^*}(e,u) \Big( \iota' \bar{\Omega}_1^{-1} (1,0,\tilde{r})' \big(\mathbf{1}\{y \leq g_{\gamma^*}(t,e) \} \nonumber \\
        & \quad - F^+_{Y|T,R}(g_{\gamma^*}(t,e)|t,\bar{r}) \big) k_R^+(\tilde{r})/f_{R}(\bar{r}) \nonumber \\
        & \quad  + \int_0^u \phi^+_{\gamma*}(e,v) q_{1}(t,\bar{r};v)k_{Q,1}^+\left( \tilde{r} \right) dv \Big) dedu \Big)^2 f^+_{Y,T,R}(y,t,\bar{r}) dydtd\tilde{r},
    \end{align}
    and $\bar{\Omega}_1 = \int \bm{x} \bm{x}' k_1(\bm{x}) dx_1 dx_2$.
    The terms $\zeta_-^{\textit{DF}}(Y,T,R)$ and $\zeta_-^{\textit{Q}}(T,R)$ contain the factor $\mathbf{1}\{R < 0\}$ while the terms $\zeta_+^{\textit{DF}}(Y,T,R)$ and $\zeta_+^{\textit{Q}}(T,R)$ contain the factor $\mathbf{1}\{R \geq 0\}$. Hence, we can compute the variance matrix of their sum as
    \begin{align*}
        \textit{var} \left( (\zeta_-^{\textit{DF}}(Y,T,R) + \zeta_-^{\textit{Q}}(T,R)) - (\zeta_+^{\textit{DF}}(Y,T,R) + \zeta_+^{\textit{Q}}(T,R))\right) = \Sigma_- + \Sigma_+.
    \end{align*}
    Since $\Sigma_-$ and $\Sigma_+$ do not vary with $n$, Chebyshev's inequality implies that the following term is $O_p\big( 1/\sqrt{nb_1} \big)$:
    \begin{align*}
        \frac{1}{nb_1} \sum_{i=1}^n \left( \zeta_-^{\textit{DF}}(Y_i,T_i,R_i) + \zeta_-^{\textit{Q}}(T_i,R_i) \right) - \frac{1}{nb_1} \sum_{i=1}^n \left( \zeta_+^{\textit{DF}}(Y_i,T_i,R_i) + \zeta_+^{\textit{Q}}(T_i,R_i) \right).
    \end{align*}
    Moreover, $\mathcal{E}$ is compact and the relevant functions in the expressions of $\zeta_\pm^{\textit{DF}}$ and $\zeta_\pm^{\textit{Q}}$ are bounded (Assumptions \ref{ass:w}, \ref{ass:smoothness-distribution}, \ref{ass:complexity-parametric-model}, and \ref{ass:h-tilde}). We can apply the Lyapnov's central limit theorem (for example, Theorem 5.11 in \cite{white2001asymptotic}) to obtain that 
    \begin{align*}
        & \big( \sqrt{nb_1} (\Sigma_- + \Sigma_+)^{-1/2} \big) \Big( \frac{1}{nb_1} \sum_{i=1}^n \left( \zeta_-^{\textit{DF}}(Y_i,T_i,R_i) + \zeta_-^{\textit{Q}}(T_i,R_i) \right) \\
        & \quad - \frac{1}{nb_1} \sum_{i=1}^n \left( \zeta_+^{\textit{DF}}(Y_i,T_i,R_i) + \zeta_+^{\textit{Q}}(T_i,R_i) \right) \Big) \overset{d}{\rightarrow} N(0,\bm{I}_{d_\Gamma}).
    \end{align*}
    Therefore, we obtain, for $\tilde{\gamma}$, the convergence rate: $\norm{\tilde{\gamma} - \gamma^*}_2 = O_p\big( b_1^2 + 1/\sqrt{nb_1} \big)$, and asymptotic normality:
    \begin{align} \label{eqn:asym-normal-tildegamma}
        \big( \sqrt{nb_1} (\Sigma_- + \Sigma_+)^{-1/2} \big) (\Delta (\tilde{\gamma} - \gamma^*) - b_1^2( B_- - B_+)) \overset{d}{\rightarrow} N(0,\bm{I}_{d_\Gamma}),
    \end{align}
    under the condition that $nb_1^7 \rightarrow 0$ (Assumption \ref{ass:bandwidth}).

    \bigskip
    \noindent \hypertarget{thm2pf7}{\textbf{Step 7.}} (Asymptotic normality of $\hat{\gamma}$.)
    By Equation (\ref{eqn:approx-error-linearization2}), we can apply the triangle inequality repeatedly and obtain that
    \begin{align*}
        \norm{\hat{L}_{\hat{\gamma}}}_w & \leq \norm{\hat{Q}_{\hat{\gamma},\hat{h}}}_w + O_p(\alpha_n) \leq  \norm{\hat{Q}_{\tilde{\gamma},\hat{h}}}_w + O_p(\alpha_n)  \leq \norm{\hat{L}_{\tilde{\gamma}}}_w + O_p(\alpha_n) ,
    \end{align*}
    where the second inequality uses the definition of $\hat{\gamma}$ in (\ref{eqn:gamma-hat-def}). Squaring the above inequality and using (\ref{eqn:Lgamma-rate}) to bounded $\lVert \hat{L}_{\tilde{\gamma}} \rVert_w$, we obtain that 
    \begin{align*}
        \norm{\hat{L}_{\hat{\gamma}}}^2_w & \leq \norm{\hat{L}_{\tilde{\gamma}}}^2_w + O_p(\alpha_n^2) + O_p(\alpha_n \tilde{\alpha}_n) \leq \norm{\hat{L}_{\tilde{\gamma}}}^2_w + O_p(\alpha_n \tilde{\alpha}_n) \\
        & = \norm{\hat{L}_{\tilde{\gamma}}}^2_w + O_p\left( \left(b_1^2 + \sqrt{\log n /(nb_1^4)}\right) \left( b_1^4 + n^{-1} b_1^{-7/6- 2\bar{\epsilon}/5} \right) \right).
    \end{align*}
    Thus, we have 
    \begin{align*}
        \norm{\hat{L}_{\hat{\gamma}}}_w^2 - \norm{\hat{L}_{\tilde{\gamma}}}_w^2 & = O_p \left( b_1^6 + n^{-1} b_1^{5/6- 2\bar{\epsilon}/5} + b_1^2 \sqrt{ \log n / n} + \sqrt{\log n} n^{- 3/2} b_1^{-3\frac{1}{6} - 2\bar{\epsilon}/5} \right) .
    \end{align*}
    We want to show that the four terms inside the $O_p$-notation in the above equation is $o(1/(nb_1))$.
    Both the terms $b_1^6$ and $b_1^2 \sqrt{\log n /n}$ are $o(1/(nb_1))$ under Assumption \ref{ass:bandwidth}(ii). The term $n^{-1} b_1^{5/6- 2\bar{\epsilon}/5}$ is $o(1/(nb_1))$ since $b_1=o(1)$. For the fourth term, we have
    \begin{align*}
        \sqrt{\log n} n^{- 3/2} b_1^{-3\frac{1}{6} -2\bar{\epsilon}/5} = o(1/(nb_1)) \iff nb_1^{4 \frac{1}{3} + \bar{\epsilon}} b_1^{-\bar{\epsilon}/5} / \log n \rightarrow \infty,
    \end{align*} 
    where the statement on the RHS is true by Assumption \ref{ass:bandwidth}(iii). The above derivations show that 
    \begin{align*}
        \norm{\hat{L}_{\hat{\gamma}}}_w^2 - \norm{\hat{L}_{\tilde{\gamma}}}_w^2 = o_p(b_1^2 + 1/(nb_1)).
    \end{align*}
    By adding and subtracting $(\tilde{\gamma} - \gamma^*)\nabla_\gamma D_{\gamma^*,h^*}$, we obtain that
    \begin{align*}
        \norm{\hat{L}_{\hat{\gamma}}}_w^2 = \norm{\tilde{L}_{\hat{\gamma}}}_w^2 + \norm{(\hat{\gamma} - \tilde{\gamma}) \nabla_\gamma D_{\gamma^*,h^*} }_w^2 + 2 (\hat{\gamma} - \tilde{\gamma}) \int \tilde{L}_{\tilde{\gamma}} (e,u) \nabla_\gamma D_{\gamma^*,h^*}(e,u) w(e,u) dedu.
    \end{align*}
    The last term (the innner product term) above is zero because $\tilde{L}_{\tilde{\gamma}}$ is orthogonal to $\nabla_\gamma D_{\gamma^*,h^*}$ from the projection perspective. This can also be verified by using the definition of $\tilde{\gamma}$. Hence, we have $\norm{(\hat{\gamma} - \tilde{\gamma}) \nabla_\gamma D_{\gamma^*,h^*} }_w^2 = o_p(1/(nb_1))$. By the same argument as in Step 3, we can show that $\norm{\hat{\gamma} - \tilde{\gamma}}_2 = o_p(1/\sqrt{nb_1})$. Therefore, by (\ref{eqn:asym-normal-tildegamma}) and Slutsky's theorem, we obtain the desired asymptotic distribution of $\hat{\gamma}$:
    \begin{align*}
        & \big( \sqrt{nb_1} (\Sigma_- + \Sigma_+)^{-1/2} \big) (\Delta (\hat{\gamma} - \gamma^*) - b_1^2( B_- - B_+)) \\
        = & \big( \sqrt{nb_1} (\Sigma_- + \Sigma_+)^{-1/2} \big) (\Delta (\tilde{\gamma} - \gamma^*) - b_1^2( B_- - B_+)) + o_p(1) \overset{d}{\rightarrow} N(0,\bm{I}_{d_\Gamma}).
    \end{align*}

\end{proof}

\begin{proof} [Proof of Proposition \ref{prop:fs-quantile}]
    We only prove the results for $\hat{h}_0(\bar{r},\cdot)$ since the results for $\hat{h}_1(\bar{r},\cdot)$ can be proved analogously. For part (i) of Assumption \ref{ass:h-tilde}, we can set the partition $\mathcal{P}_0^n$ to be the class of intervals $\{[u_j,u_{j+1}]:j=0,\cdots,J_n\}$. The estimator $\hat{h}_0(\bar{r},\cdot)$ is a linear function within each interval and hence is contained in the class $\mathcal{H}_0^n(\mathcal{P}_0^n)$. 

    For part (ii), notice that, under Assumption \ref{ass:smoothness-distribution}(i), the estimator $\hat{h}_0(\bar{r},u_0)$ and $\hat{h}_0(\bar{r},u_{J_n + 1})$ converge to $t_0'$ and $t_0''$, respectively, at the $1/n$ rate. Therefore, we can replace $\hat{h}_0(\bar{r},u_0)$ by $t_0'$ and $\hat{h}_0(\bar{r},u_{J_n + 1})$ by $t_0''$ without affecting the asymptotics. Let $\tilde{h}_0(\bar{r},u)$ denote the solution of (\ref{eqn:first-step-quantile-est}) at given $u$. The uniform asymptotic linear representation for $\hat{h}_0(\bar{r},u), u \in (0,1)$ follows from Lemma 3 in the Appendix of \cite{dong2021regression}, which is a slight modification of Theorem 1.2 of \cite{QU2015nonparametric}. Then we can use Step 2 in the proof of Theorem 2 in \cite{QU2015nonparametric} to show that the error induced by linear interpolation is asymptotically negligible.

    The uniform convergence rate in Part (iii) of Assumption \ref{ass:h-tilde} can be shown by using the uniform asymptotic linear representation. Since $\nu_0$ is bounded, the bias term is $O(b_1^2)$. In Lemma \ref{lm:uc4}, we show that the stochastic term is satisfies 
    \begin{align*}
        \sup_{u \in [0,1]} \left| \frac{1}{nb_1} \sum_{i=1}^n q_0(T_i,R_i;u) k_{Q,0}\left( \frac{R_i - \bar{r}}{b_1} \right) \mathbf{1}\{R_i < \bar{r}\} \right| = O_p \left( \sqrt{\log n / nb_1} \right).
    \end{align*}
    This proves the desired result.

\end{proof}

\subsection{Uniform convergence rates and the empirical process theory} \label{ssec:ept}

Below are some basic concepts and results from the empirical process theory which are used to prove several uniform convergence results.

Let $\mathcal{F}$ be a class of uniformly bounded measurable matrix-valued functions, that is, there exists $M>0$ such that, for all $f \in \mathcal{F}$, $\norm{f}_2 \leq M$. Let $N(\mathcal{F},P,\epsilon)$ be the $\epsilon$-covering number of the metric space $(\mathcal{F},L_2(P))$, that is, $N(\mathcal{F},P,\epsilon)$ is defined as the minimal number of open $\norm{\cdot}_{L_2(P)}$-balls of radius $\epsilon$ and centers in $\mathcal{F}$ required to cover $\mathcal{F}$. 

We say that a uniformly bounded function class $\mathcal{F}$ is \textit{Euclidean} if there exists $A_1,A_2 >0$ (that only depend on the uniform bound) such that for every probability measure $P$ and every $\epsilon \in (0,1]$, $N(\mathcal{F},P,\epsilon) \leq A_1 / \epsilon^{A_2}$. We say that a function class $\mathcal{F}$ is \textit{log-Euclidean} with coefficient $\rho \in (0,1)$ if there exists $A>0$ (that only depends on the uniform bound) such that for every probability measure $P$ and every $\epsilon \in (0,1]$, $\log N(\mathcal{F},P,\epsilon) \leq A /\epsilon^{2\rho}$.


The above definition of Euclidean classes is introduced by \cite{nolan1987uprocess}. The same concept is also studied by \cite{gine1999laws}, but they refer to what we call ``Euclidean'' as ``VC.'' There is a slight difference that \cite{nolan1987uprocess} use the $L_1$-norm while \cite{gine1999laws} use the $L_2$-norm. We ignored the envelope in their definition since we only work with uniformly bounded $\mathcal{F}$. The following two lemmas demonstrates how to generate function classes that are Euclidean and log-Euclidean.

\begin{lemma} \label{lm:generate-Euclidean}
    Let $\mathcal{F}_1$ and $\mathcal{F}_2$ be uniformly bounded and Euclidean classes of functions. The following classes of functions are also uniformly bounded and Euclidean.
    \begin{enumerate} [label = (\roman*)]
        \item $\mathcal{F}_1 \oplus \mathcal{F}_2 = \{f_1 + f_2:f_1 \in \mathcal{F}_1, f_2 \in \mathcal{F}_2\}$.
        \item $\mathcal{F}_1\mathcal{F}_2 = \{f_1\cdot f_2:f_1 \in \mathcal{F}_1, f_2 \in \mathcal{F}_2\}$.
        \item $\{ \mathbb{E}[f_1(\cdot) | X] : f_1 \in \mathcal{F}_1 \}$.
        \item $\left\{ k\left( (\cdot - x)/b \right): x \in \mathbb{R}, b>0  \right\}$, where $k:\mathbb{R} \rightarrow \mathbb{R}$ is a function of bounded variation.
    \end{enumerate}
\end{lemma}

\begin{proof} [Proof of Lemma \ref{lm:generate-Euclidean}]
    See Appendix B in \cite{xie2021uniform}.
\end{proof}

\begin{lemma} \label{lm:log-Eucl-times-Eucl}
    Let $\mathcal{F}_1$ be a uniformly bounded and Euclidean class of functions and $\mathcal{F}_2$ be a uniformly bounded and log-Euclidean class of functions with coefficient $\rho$. Then $\mathcal{F}_1\mathcal{F}_2$ is uniformly bounded and log-Euclidean with coefficient $\rho+
    \epsilon$ for any $\epsilon>0$.
\end{lemma}

The following two lemmas give the asymptotic order of the supremum of empirical processes generated by Euclidean and log-Euclidean classes, respectively.

\begin{proof} [Proof of Lemma \ref{lm:log-Eucl-times-Eucl}]
    This follows from the definition of Euclidean and log-Euclidean classes.
\end{proof}

\begin{lemma} \label{lm:gine}

    Let $X_1,\cdots,X_n$ be an iid sample of a random vector $X$ in $\mathbb{R}^d$. Let $\mathcal{G}_n$ be a sequence of classes of measurable real-valued functions defined on $\mathbb{R}^d$. Assume that there is a fixed uniformly bounded Euclidean class $\mathcal{F}$ such that $\mathcal{F}_n \subset \mathcal{F} $ for all $n$. Let $\sigma^2_{n} \geq \sup_{f \in \mathcal{F}_n} \mathbb{E}[f(X)^2]$. Then
        \begin{align*}
             \sup_{f \in \mathcal{F}_n}  \left| \sum_{i=1}^n (f(X_i) - \mathbb{E}f(X_i)) \right| = O_p \left( \sqrt{n \sigma_n^2 |\log \sigma_n|} + |\log \sigma_n| \right) .
        \end{align*}
        In particular, if $n \sigma_n^2 / |\log \sigma_n| \rightarrow \infty$, then the above rate simplifies to $O_p \left( \sqrt{n \sigma_n^2 |\log \sigma_n|} \right) .$
\end{lemma}

\begin{proof} [Proof of Lemma \ref{lm:gine}]
    This is Lemma 2 in \cite{xie2021uniform}.
\end{proof}

\begin{lemma} \label{lm:log-Euclidean-rate}

    Let $X_1,\cdots,X_n$ be an iid sample of a random vector $X$ in $\mathbb{R}^d$. Let $\mathcal{F}_n$ be a sequence of classes of measurable real-valued functions defined on $\mathbb{R}^d$. Assume that there is a fixed uniformly bounded log-Euclidean class $\mathcal{F}$ with coefficient $\rho$ such that $\mathcal{F}_n \subset \mathcal{F} $ for all $n$. Let $\sigma^2_{n} = \sup_{f \in \mathcal{F}_n} \mathbb{E}[f(X)^2]$. Then
        \begin{align*}
            \sup_{f \in \mathcal{F}_n}  \left| \sum_{i=1}^n (f(X_i) - \mathbb{E}f(X_i)) \right| = O_p \left( \sqrt{n} \sigma_n^{1-\rho} + n^{\rho/(1+\rho)} \right) .
        \end{align*}
\end{lemma}

\begin{proof} [Proof of Lemma \ref{lm:log-Euclidean-rate}]
    Let $M > 0$ be the uniform bound of $\mathcal{F}$. Since $\mathcal{F}$ is log-Euclidean with coefficient $\rho$, there exists $A>0$ such that $\log N(\mathcal{F},P_n,\epsilon) \leq A/\epsilon^\rho$ for every $\epsilon \in (0,1]$, where $P_n$ is the empirical measure. Since each $\mathcal{F}_n$ is contained in $\mathcal{F}$, the above result also holds when $\mathcal{F}$ is replaced by $\mathcal{F}_n$. Denote $\textit{Rad}_i, 1 \leq i \leq n,$ as a sequence of iid Rademacher variables. By Equation (3.19) in \cite{koltchinskii2011oracle} (which is a result of Theorem 3.12 in the same book), there exists a universal constant $C>0$ such that
    \begin{align*}
        \mathbb{E}\sup_{f \in \mathcal{F}}  \left| \sum_{i=1}^n \textit{Rad}_i f(X_i) \right| \leq C A^\rho M^\rho \sqrt{n} \sigma_n^{1-\rho} \vee C A^{2\rho/(\rho+1)} M n^{\rho/(1+\rho)} = O_p \left( \sqrt{n} \sigma_n^{1-\rho} + n^{\rho/(1+\rho)} \right).
    \end{align*}
    Then the desired result follows from the usual symmetrization argument (for example, Theorem 2.1 in \cite{koltchinskii2011oracle}) and Chebyshev's inequality.
\end{proof}

The following three lemmas give uniform convergence results that are used in the proof of Theorem \ref{thm:estimation}.

\begin{lemma} \label{lm:uc1}
    Under the assumptions of Theorem \ref{thm:estimation}, the following term is $O_p(\tilde{\alpha}_n)$:
    \begin{align*}
        & \sup_{e\in\mathcal{E},u \in [0,1],\gamma \in \Gamma,h_0 \in \mathcal{H}_0(\mathcal{P}_0^n)} \left| \frac{1}{nb_1^2} \sum_{i=1}^n \int_0^u \iota' \Xi_0(h_0(v))^{-1}   s_0(Y_i,T_i,R_i;g_{\gamma}(h_0(v),\bar{r},e),h_0(v)) dv \right| .
    \end{align*}
\end{lemma}

\begin{proof} [Proof of Lemma \ref{lm:uc1}]
    Since $\mathcal{P}_0^n$ is a finite partition, we can without loss of generality assume that $\mathcal{P}_0^n$ only contains the whole interval $[t_0',t_0'']$ so that there is effectively no partition. To simply notation, we omit the term $\mathcal{P}_0^n$. 
By the change of variables $\tilde{v} = (T_i - h_0(v))/b_1$ and Fubini's theorem, we have
    \begin{align*}
        & \Big| \frac{1}{nb_1^2} \sum_{i=1}^n \int_0^u \iota' \Xi_0(h_0(v))^{-1}   s_0(Y_i,T_i,R_i;g_{\gamma}(h_0(v),\bar{r},e),h_0(v)) dv \Big| \\
        \leq & \int \Big| \frac{1}{nb_1} \sum_{i=1}^n \iota' \Xi_0(T_i + b_1 \tilde{v})^{-1}   s_0(Y_i,T_i,R_i;g_{\gamma}(T_i + b_1 \tilde{v},\bar{r},e),T_i + b_1 \tilde{v}) \\
        & \quad \times (h_0^{-1})'(T_i+b_1\tilde{v}) \mathbf{1}\{(T_i - h_0(u))/b_1 < \tilde{v} < (T_i - t_0')/b_1\}  d\tilde{v} \Big| \\
        \leq & \sup_{\tilde{v} \in (-1,1)} \Big| \frac{1}{nb_1} \sum_{i=1}^n \iota' \Xi_0(T_i + b_1 \tilde{v})^{-1}   s_0(Y_i,T_i,R_i;g_{\gamma}(T_i + b_1 \tilde{v},\bar{r},e),T_i + b_1 \tilde{v}) \\
        & \quad \times (h_0^{-1})'(T_i+b_1\tilde{v}) \mathbf{1}\{(T_i - h_0(u))/b_1 < \tilde{v} < (T_i - t_0')/b_1\}  \Big|,
    \end{align*}
    where, in the last inequality, the supremum is taken over $\tilde{v} \in (-1,1)$ because of the support of $k_T$. Define the following function of $(Y,T,R)$ indexed by $(v,u,e,\gamma,h_0)$:
    \begin{align*}
        \psi_n(Y,T,R;v,u,e,\gamma,h_0) & = \iota' \Xi_0(T + b_1 v)^{-1}   s_0(Y,T,R;g_{\gamma}(T + b_1 v,\bar{r},e),T + b_1 v) \\
        & \quad \times (h_0^{-1})'(T+b_1v) \mathbf{1}\{(T - h_0(u))/b_1 < v < (T - t_0')/b_1\}.
    \end{align*}
    Let $\Psi_n = \{\psi_n(\cdot,\cdot,\cdot;v,u,e,\gamma,h): v \in (-1,1),u \in (0,1),e \in \mathcal{E}, \gamma \in \Gamma, h_0 \in \mathcal{H}_0\}$. Our goal is to use empirical process theory to derive the asymptotic order of 
    \begin{align*}
        \sup_{\psi_n \in \Psi_n} | \sum_{i=1}^n \psi_n(Y,T,R;v,u,e,\gamma,h_0) |.
    \end{align*}
    Consider a larger class $\Psi $ as the product $\Psi = \Psi_\Xi \Psi_Y \Psi_{TR} \Psi_{\mathcal{H}_0} $, where
    \begin{align*}
        \Psi_{\Xi_0} & = \{ T \mapsto \iota' \Xi_0(T + v)^{-1}: v \in (-1,1) \}, \\
        \Psi_Y & = \{ (Y,T,R) \mapsto \tilde{K}_Y(Y,T,R;g_{\gamma}(T + v,\bar{r},e),T + v) : v \in (-1,1),\gamma \in \Gamma, e \in \mathcal{E} \}, \\
        \Psi_{TR} & = \{ (Y,T,R) \mapsto (1,v,(R-\bar{r})/b)' k_T(v) k_R^-((R-\bar{r})/b) \\
        & \quad \times  \mathbf{1}\{T - h_0(u) < bv < T - t_0'\} : b,u  \in (0,1),v \in (-1,1) \}, \\
        \Psi_{\mathcal{H}_0} & = \{ T \mapsto (h_0^{-1})'(T+v): h_0 \in \mathcal{H}_0, v \in (0,1) \}.
    \end{align*} 
    Notice that in the above definition of $\Psi_{\Xi_0}$, $\Psi_Y$, and $\Psi_{\mathcal{H}_0}$, omitting the parameter $b$ does not change the class under consideration. The class $\Psi$ does not vary with $n$ and $\Psi_n \subset \Psi, n \geq 1$.
    In the following paragraphs, we show that the classes $\Psi_{\Xi_0}$, $\Psi_Y$, and $\Psi_{TR}$ are Euclidean while the class $\Psi_{\mathcal{H}_0}$ is log-Euclidean. 
    
    For $\Psi_{\Xi_0}$, we know that $\norm{\Xi_0}$ and $\lVert \Xi_0^{-1} \rVert$ are uniformly bounded by Lemma 1 in \cite{xie2021uniform}. By the smoothness of $f^{-}_{T,R}$ in Assumption \ref{ass:smoothness-distribution}, the class $\{ T \mapsto \iota' \Xi_0(T + v): v \in (-1,1) \}$ is Lipschitz in the parameter $v \in (-1,1)$ and hence, by Theorem 2.7.11 in \cite{wellner1996}, has covering numbers bounded by that of one-dimensional intervals. This implies that $\{ T \mapsto \iota' \Xi_0(T + v): v \in (-1,1) \}$ is uniformly bounded and Euclidean. Then by Theorem 3 in \cite{ANDREWS1994empirical}, we know that $\Psi_{\Xi_0}$ is uniformly bounded and Euclidean. 
    
    The class $\Psi_Y$ can be written as $\Psi_Y = \Psi_{Y1} + \Psi_{Y2}$, where
    \begin{align*}
        \Psi_{Y1} & = \{ (Y,T) \mapsto K_Y((g_{\gamma}(T + v,\bar{r},e)-Y)/b): b,v \in (0,1), \gamma \in \Gamma, e \in \mathcal{E} \}, \\
        \Psi_{Y2} & = \{ (Y,T,R) \mapsto - \mathbb{E}[K_Y((g_{\gamma}(T + v,\bar{r},e)-Y)/b)|T,R] : b,v \in (0,1), \gamma \in \Gamma, e \in \mathcal{E} \}.
    \end{align*}
    In view of Lemma \ref{lm:generate-Euclidean}(i) and (iii), we only need to show that $\Psi_Y$ is Euclidean. The class $\Psi_{Y}$ is uniformly bounded by $1$. The function $K_Y$ is increasing since $k_Y$ is positive. The subgraph class of $\Psi_Y$ can be written as 
    \begin{align*}
        & \{ \{ (y,t,s) : K_Y((g_{\gamma}(t + v,\bar{r},e)-y)/b) \leq s \} : b,v \in (0,1), \gamma \in \Gamma, e \in \mathcal{E}\} \\
        = & \{ \{ (y,t,s) : g_{\gamma}(t + v,\bar{r},e)-y - b K^{-1}(s) \leq 0 \} : b,v \in (0,1), \gamma \in \Gamma, e \in \mathcal{E}\}.
    \end{align*}
    By Assumption \ref{ass:complexity-parametric-model}, the following function class is finite-dimensional:
    \begin{align*}
        \{ (t,y,s) \mapsto g_{\gamma}(t + v,\bar{r},e)-y - b K^{-1}(s) : b,v \in (0,1), \gamma \in \Gamma, e \in \mathcal{E}\}.
    \end{align*}
    By Lemma 18(ii) in \cite{nolan1987uprocess}, the subgraph class of $\Psi_Y$ is a polynomial class, which implies (by Theorem 2.6.7 in \cite{wellner1996}) that $\Psi_Y$ is Euclidean.
    
    For the class $\Psi_{TR}$, notice that the function $k_T$, $k_R$, and the indicator function are all of bounded variation. The kernel functions $k_T$ and $k_R$ are supported on $[-1,1]$. Therefore, the term $(R-\bar{r})/b$ is bounded between $[-1,1]$. By Lemma \ref{lm:generate-Euclidean}(ii) and (iv), we know that $\Psi_{TR}$ is uniformly bounded and Euclidean.

    Lastly, by Assumption \ref{ass:h-tilde}(i), the class $\Psi_{\mathcal{H}_0}$ is contained in the class of twice continuously diferentiable functions whose second-order derivatives are Lipschitize continuous. By the well-known bounds on the entropy of Lipschitz classes (see, for example, Example 5.11 in Chapter 5 of \cite{wainwright2019high}), we know the class $\Psi_{\mathcal{H}_0}$ is log-Euclidean with coefficient $1/2 \times 1/(2+1) = 1/6$. Then by Lemma \ref{lm:log-Eucl-times-Eucl}, we know that $\Psi$ is log-Euclidean with coefficient $1/6 + \epsilon$ for any small $\epsilon>0$.

    Next, we want to derive a uniform variance bound for the class $\Psi$ and appeal to Lemma \ref{lm:log-Euclidean-rate}. By the uniform boundedness of the classes studied above and applying the usual change of variables, we obtain that
    \begin{align*}
        \mathbb{E}[\psi_n(Y,T,R;v,u,e,\gamma,h_0)^2] \leq C \mathbb{E}[ k_R((R-\bar{r})/b_1)^2 ] = C b_1 \int k_R(\tilde{r}) f_{R}(\bar{r} + b_1 \tilde{r}) d\tilde{r} = O(b_1).
    \end{align*}
    Lemma \ref{lm:log-Euclidean-rate} then gives that
    \begin{align*}
        \sup_{\psi \in \Psi_n} \Big| \sum_{i=1}^n \psi(Y,T,R;v,u,e,\gamma,h_0) \Big| = O_p \left( n^{1/2 } b_1^{(1-1/6)/2 - \epsilon} + n^{1/7} \right),
    \end{align*}
    for any small $\epsilon > 0$. Notice that, in the rate specified above, the term $n^{1/7}$ is dominated in view of Assumption \ref{ass:bandwidth}. 
    Then the desired convergence rate follows from dividing by $nb_1$ on both sides.

\end{proof}

\begin{lemma} \label{lm:uc2}
    Under the assumptions of Theorem \ref{thm:estimation}, we have
    \begin{align*}
        \sup_{y \in \mathbb{R}, t \in [t_0',t_0'']} \norm{ \frac{1}{nb_1^2} \sum_{i=1}^n \frac{\partial }{\partial y} s_0(Y_i,T_i,R_i;y,t)}_2 & = O_p\left( \sqrt{\log n /(nb_1^4)} \right),\\
         \sup_{y \in \mathbb{R}, t \in [t_0',t_0'']} \norm{ \frac{1}{nb_1^2} \sum_{i=1}^n \frac{\partial }{\partial t} s_0(Y_i,T_i,R_i;y,t)}_2 & = O_p\left( \sqrt{\log n /(nb_1^4)} \right).
    \end{align*}
\end{lemma}

\begin{proof} [Proof of Lemma \ref{lm:uc2}]
    The partial derivative of $s_0$ with respect to $y$ is a vector of length three whose generic element can be denoted by $\dot{s}_0(Y_i,T_i,R_i;y,t,b_1,b_2)/b_2,$
    where
    \begin{align*}
        \dot{s}_0(Y,T,R;y,t) = \left(\frac{T-t}{b_1}\right)^{\ell_1}\left(\frac{R-\bar{r}}{b_1}\right)^{\ell_2} \left(  k_Y \left( (y - Y)/b_2 \right) - \mathbb{E} \left[ k_Y \left( (y - Y)/b_2 \right) | T,R \right] \right) k_0(X(t))
    \end{align*}
    with $(\ell_1,\ell_2) = (0,0),(1,0),(0,1)$. We use the empirical process theory to derive the uniform convergence rate of the sample average of $\dot{s}_0$.
    Recall that the kernel functions $k_Y,k_T,$ and $k_R$ are of bounded variation. Then by Lemma \ref{lm:generate-Euclidean}, we know that the following function class is uniformly bounded and Euclidean:
    \begin{align*}
        &  \{(Y,T,R) \mapsto \dot{s}_0(Y,T,R;y,t,b,b') : y \in \mathbb{R}, t \in [t_0',t_0''], b,b' > 0 \} \\
        = & \{ (Y,T,R) \mapsto ((T-t)/b)^{\ell_1}\left((R-\bar{r})/b\right)^{\ell_2} \left(  k_Y \left( (y - Y)/b' \right) - \mathbb{E} \left[ k_Y \left( (y - Y)/b' \right) | T,R \right] \right) \\
        & \quad \times k_T((T-t)/b) k_R^-((R-\bar{r})/b) : y \in \mathbb{R}, t \in [t_0',t_0''], b,b' > 0 \}.
    \end{align*}
    By the law of iterated expectations and differentiation under the integral, we know that $\dot{s}_0$ is centered. By using the fact that $k_T$ and $k_R$ are supported on $[-1,1]$ and $k_Y$ is bounded and applying the standard change of variables, we can bound the variance of $\dot{s}_0$ by 
    \begin{align*}
         2\norm{k_Y}_\infty^2 \mathbb{E} \left[ k_T((T-t)/b_1)^2 k_R^-((R-\bar{r})/b_1)^2 \right] 
         = b_1^2 2\norm{k_Y}_\infty^2 \int k_T(x_1)^2 k_R^-(x_2)^2 = O(b_1^2),
    \end{align*}
    uniformly over $y \in \mathbb{R}$ and $t \in [t_0',t_0'']$. Then by Lemma \ref{lm:gine}, we know that the uniform convergence rate of the sample average of $\dot{s}_0$ is $O_p\left(\sqrt{nb_1^2\log n} \right)$. Therefore,
    \begin{align*}
        \sup_{y \in \mathbb{R}, t \in [t_0',t_0'']} \norm{ \frac{1}{nb_1^2} \sum_{i=1}^n \frac{\partial }{\partial y} s_0(Y_i,T_i,R_i;y,t)}_2 = \frac{1}{nb_1^2 b_2} O_p\left(\sqrt{nb_1^2\log n} \right) =  O_p\left( \sqrt{\log n /(nb_1^4)} \right)
    \end{align*}
    under the condition that $b_1/b_2 \in [1/C,C]$ (Assumption \ref{ass:bandwidth}). This proves the first claim of the lemma. For the second claim, the same argument applies. We just want to point out that $k_T$ is differentiable on the entire real line by Assumption \ref{ass:kernels} even though its support is $[-1,1]$.

\end{proof}

\begin{lemma} \label{lm:uc3}
    Under the assumptions of Theorem \ref{thm:estimation}, we have
    \begin{align*}
        \sup_{y \in \mathcal{Y}, t \in [t_0',t_0'']} \left| \frac{1}{nb_1^2} \sum_{i=1}^n \iota' \Xi_0(t)^{-1}   s_0(Y_i,T_i,R_i;y,t) \right| & = O_p\left( \sqrt{\log n /(nb_1^2)} \right), \\
        \sup_{y \in \mathbb{R}, t \in [t_0',t_0''] }\norm{ \frac{1}{nb_1^2} \sum_{i=1}^n s_0(Y_i,T_i,R_i;y,t) }_2 & = O_p\left( \sqrt{\log n /(nb_1^2)} \right) .
    \end{align*}
\end{lemma}

\begin{proof} [Proof of Lemma \ref{lm:uc3}]
    Following the same steps as in the proofs of the previous two lemmas, we can show that the relevant function classes are uniformly bounded and Euclidean. By the usual change of variables, we can show that the uniform variance bound is $O(b_1^2)$ before taking into account the factor $1/(nb_1)$ in the two terms. Then the desired results follow from Lemma \ref{lm:gine}. The details are omitted for brevity.
\end{proof}

\begin{lemma} \label{lm:uc4}
    \begin{align*}
        \sup_{u \in [0,1]} \left| \frac{1}{nb_1} \sum_{i=1}^n q_0(T_i,R_i;u) k_{Q,0}\left( \frac{R_i - \bar{r}}{b_1} \right) \mathbf{1}\{R_i < \bar{r}\} \right| = O_p \left( \sqrt{\log n / nb_1} \right).
    \end{align*}

\end{lemma}

\begin{proof} [Proof of Lemma \ref{lm:uc4}]
    Without loss of generality, let $c=1$. Define 
    \begin{align*}
        \psi_n(T,R;u) & = q_0(T,R;u) k_{Q,0}\left( \frac{R - \bar{r}}{b_1} \right) \mathbf{1}\{R < \bar{r}\} \\
        & = \frac{u - \mathbf{1}\{ T \leq h_0^*(\bar{r},u) \}}{f_R(\bar{r}) f^-_{T | R}(h_0^*(\bar{r},u) | \bar{r}) } \iota' \Omega_{Q,0}^{-1} (1,(R - \bar{r})/b_1)' K_{\textit{FS}}((R - \bar{r})/b_1)
    \end{align*}
    and $\Psi_n = \{ (T,R) \mapsto \psi_n(T,R;u): u \in [0,1] \}$. By the law of iterated expectations, $\psi_n$ is centered. Let $M = \sup_{u \in [0,1]} |f_R(\bar{r}) f^-_{T | R}(h_0^*(\bar{r},u) | \bar{r})|$. Define a product class $\Psi = \Psi_T \Psi_R$ where 
    \begin{align*}
        \Psi_T & = \{ (T,R) \mapsto C (u - \mathbf{1}\{ T \leq t \}): u \in [0,1],t \in [t_0',t_0''], |C| \leq M \}, \\
        \Psi_R & = \{ (T,R) \mapsto \iota' \Omega_{Q,0}^{-1} (1,(R - \bar{r})/b)' K_{\textit{FS}}((R - \bar{r})/b): b > 0 \}.
    \end{align*}
    The class $\Psi$ does not vary with $n$, and $\Psi_n \subset \Psi, n \geq 1$. The class $\Psi_T$ is uniformly bounded and Euclidean since the set of indicator functions $\mathbf{1}\{ T \leq t \},t \in [t_0',t_0'']$ is Euclidean. The class $\Psi_R$ is uniformly bounded and Euclidean since $K_{\textit{FS}}$ is of bounded variation and compactly supported. By the usual change of variables, we can show that the uniform variance bound for $\Psi_n$ is $O(b_1)$. Then the desired convergence rate follows from Lemma \ref{lm:gine}.

\end{proof}

\subsection{Covariance Matrix Estimation} \label{ssec:cov-est}

In this section, we discuss the estimation of the asymptotic variance matrix of $\hat{\gamma}$, which involves the estimation of $\Delta$, $\Sigma_-$, and $\Sigma_+$. For concreteness, we consider the first-step nonparametric conditional quantile estimation procedure described in Section \ref{ssec:quantile-est} and Proposition \ref{prop:fs-quantile}. In the expressions of $\Delta$, $\Sigma_-$, and $\Sigma_+$, the functions that require estimation include $\nabla_\gamma D_{\gamma^*,h^*}$, $\phi_{\gamma^*}^{\pm}$, $f^\pm_{Y,T,R}$, $f^\pm_{T|R}$, and $f_R$. By definition, 
\begin{align*}
    \nabla_\gamma D_{\gamma^*}(e,u) & = \int_0^u \big[ \frac{\partial}{\partial Y}F^-_{Y|T,R}( g_{\gamma^*}(h_0^*(\bar{r},v),\bar{r},e) | h^*_0(\bar{r},v),\bar{r} ) \nabla_\gamma g_{\gamma^*}(h_0^*(\bar{r},v),\bar{r},e)  \\
    & \quad - \frac{\partial}{\partial Y} F^+_{Y|T,R}( g_{\gamma^*}(h_1^*(\bar{r},v),\bar{r},e) | h_1^*(\bar{r},v),\bar{r} ) \nabla_\gamma g_{\gamma^*}(h_1^*(\bar{r},v),\bar{r},e) \big] dv .
\end{align*}
In the above quantity, we only need to estimate $\frac{\partial}{\partial Y} F^-_{Y|T,R} = f^\pm_{Y|T,R}$ since we already have estimators for $\gamma^*$ and $h^*$. By observing the definition of $\phi_{\gamma^*}^{\pm}$, we know that the additional term that requires estimation is $\frac{\partial}{\partial T} F^\pm_{Y|T,R}$. To summarize, we want to estimate $f^\pm_{Y,T,R}$ and $\frac{\partial}{\partial T} F^\pm_{Y|T,R}$. Once $f^\pm_{Y,T,R}$ is obtained, we can operate to get the marginal and conditional density functions.

For estimation of $f^\pm_{Y,T,R}$, we can employ the method developed by \cite{cattaneo2020simple}. They use the second-order local polynomial regression to estimate the joint density. Due to the nature of local polynomial regressions, the estimator is boundary adaptive and particularly suitable for RD designs.
To estimate the partial derivative $\frac{\partial}{\partial T} F^\pm_{Y|T,R}$, we can employ a second-order local polynomial regression. The procedure is similar to STEP 2 in the construction of $\hat{\gamma}$. We add two quadratic terms into the minimization problem:
\begin{align*} 
    &\sum_{i:R_i < \bar{r}} \left( K_Y\left( \frac{y - Y_i}{b_2} \right) - a^- - a_{T}^-(T_i-t) - a_{T,2}^-(T_i-t)^2 - a_{R}^-(R_i - \bar{r}) - a_{R,2}^-(R_i - \bar{r})^2 \right)^2 \\
    & \quad \times k_T\left( \frac{T_i - t}{b_1} \right) k_R\left( \frac{R_i - \bar{r}}{b_1} \right).
\end{align*}
The minimizer $\hat{a}_{T}^-$ is the estimate of $\frac{\partial}{\partial T} F^-_{Y|T,R}(y|t,\bar{r})$. The estimate of $\frac{\partial}{\partial T} F^+_{Y|T,R}(y|t,\bar{r})$ can be analogously constructed.

We assume that the resulting estimators $\hat{f}^\pm_{Y,T,R}$ and $\frac{\partial}{\partial T} \hat{F}^\pm_{Y|T,R}$ are uniformly consistent. That is,
\begin{align*}
    \sup_{y,t,r} |\hat{f}^\pm(y,t,r) - f^\pm(y,t,r)| & = o_p(1), \\
    \sup_{y,t} \left| \frac{\partial}{\partial T} \hat{F}^\pm_{Y|T,R}(y | t,\bar{r}) - \frac{\partial}{\partial T} F^\pm_{Y|T,R}(y | t,\bar{r}) \right| & = o_p(1).
\end{align*}
Such uniform convergence results can be proved along the lines of, for examples, \cite{fan2016multivariate} and \cite{xie2021uniform}. The details are omitted here. We can construct the following distributional estimates:
\begin{align*}
    \hat{f}^\pm_{T,R}(t,\bar{r}) & = \int \hat{f}^\pm_{Y,T,R}(y,t,\bar{r}) dy, \\
    \hat{f}^\pm_{Y|T,R}(y,t|\bar{r}) & = \hat{f}^\pm_{Y,T,R}(y,t,\bar{r}) / \hat{f}^\pm_{T,R}(t,\bar{r}).
\end{align*}
Under the assumption that $f^\pm_{T,R}$ is bounded away from zero, the estimator $\hat{f}^\pm_{Y|T,R}(y,t|\bar{r})$ is uniformly consistent.
Let 
\begin{align*}
    \hat{\Delta} & = \int w(e,u) \Big( \int_0^u \big[ \hat{f}^-_{Y|T,R}( g_{\hat{\gamma}}(\hat{h}_0(\bar{r},v),\bar{r},e) | \hat{h}_0(\bar{r},v),\bar{r} ) \nabla_\gamma g_{\hat{\gamma}}(\hat{h}_0(\bar{r},v),\bar{r},e)  \\
    & \quad - \hat{f}^+_{Y|T,R}( g_{\hat{\gamma}}(\hat{h}_1(\bar{r},v),\bar{r},e) | \hat{h}_1(\bar{r},v),\bar{r} )  \nabla_\gamma g_{\hat{\gamma}}(\hat{h}_1(\bar{r},v),\bar{r},e) \big] dv \Big)^2 dedu .
\end{align*}
Under the uniform consistency of $\hat{f}^\pm_{Y|T,R}$ and $\hat{h}$ and the consistency of $\hat{\gamma}$, we can show that $\hat{\Delta}$ is a consistent estimator of $\Delta$. For the estimation of $\Sigma_-$, define 
\begin{align*}
    \hat{f}_R (\bar{r}) & = \int \hat{f}^\pm_{T,R}(t,\bar{r}) dt, \\
    \hat{f}^\pm_{T|R}(t|\bar{r}) & =\hat{f}^\pm_{T,R}(t,\bar{r}) / \hat{f}_R(\bar{r}), \\
    \hat{\phi}^-_{\hat{\gamma}}(e,v) & = \hat{f}^-_{Y|T,R}( g_{\hat{\gamma}}(\hat{h}_0(v),e)|\hat{h}_0(v),\bar{r}) \frac{\partial}{\partial T} g_{\hat{\gamma}}(\hat{h}_0(v),e) + \frac{\partial}{\partial T} \hat{F}^-_{Y|T,R}( g_{\hat{\gamma}}(\hat{h}_0(v),e)|\hat{h}_0(v),\bar{r}).
\end{align*}
The above estimators are also uniformly consistent. In particular,
\begin{align*}
    \sup_{e,v} \left| \hat{\phi}^-_{\hat{\gamma}}(e,v) - \phi^-_{\gamma^*}(e,v) \right| = o_p(1).
\end{align*}
Let 
\begin{align*}
    \hat{\Sigma}_- & = \int \Big( \int_{\mathcal{E}} \int_0^1 w(e,u) \Big( \int_0^u \big[ \hat{f}^-_{Y|T,R}( g_{\hat{\gamma}}(\hat{h}_0(\bar{r},v),\bar{r},e) | \hat{h}_0(\bar{r},v),\bar{r} ) \nabla_\gamma g_{\hat{\gamma}}(\hat{h}_0(\bar{r},v),\bar{r},e)  \\
    & \quad - \hat{f}^+_{Y|T,R}( g_{\hat{\gamma}}(\hat{h}_1(\bar{r},v),\bar{r},e) | \hat{h}_1(\bar{r},v),\bar{r} )  \nabla_\gamma g_{\hat{\gamma}}(\hat{h}_1(\bar{r},v),\bar{r},e)\big] dv \Big) \\
    & \quad \Big( \iota' \bar{\Omega}_0^{-1} (1,0,\tilde{r})' \big(\mathbf{1}\{y \leq g_{\hat{\gamma}}(t,e) \}  - \hat{F}^-_{Y|T,R}(g_{\hat{\gamma}}(t,e)|t,\bar{r}) \big) k_R^-(\tilde{r})/\hat{f}_{R}(\bar{r}) \\
    & \quad + \frac{k_{Q,0}^-\left( \tilde{r} \right)}{c\hat{f}^-_{T,R}(t,\bar{r})} \int_0^u \hat{\phi}^-_{\hat{\gamma}}(e,v) (v - \mathbf{1}\{t \leq \hat{h}_0(\bar{r},v)\})  dv \Big) dedu \Big)^2 \hat{f}^-_{Y,T,R}(y,t,\bar{r}) dydtd\tilde{r}.
\end{align*}
Under the uniform consistency of $\hat{f}_R$, $\hat{f}_{T,R}^\pm$, $\hat{f}^\pm_{Y|T,R}$, $\hat{F}_{Y|T,R}$, $\hat{f}^\pm_{Y,T,R}$, $\hat{\phi}^-_{\hat{\gamma}}$, and $\hat{h}$, and the consistency of $\hat{\gamma}$, $\hat{\Sigma}_-$ is a consistent estimator of $\Sigma_-$. The estimation of $\Sigma_+$ can be performed analogously.

\bibliographystyle{chicago}
\bibliography{references.bib}

\end{document}